\documentclass[11pt]{article}

\usepackage[numbers]{natbib}
\usepackage{amsmath}
\usepackage{amssymb}
\usepackage{amsthm}
\usepackage{mathtools}
\usepackage{xcolor}
\usepackage{array}
\usepackage{booktabs}
\usepackage{hyperref}
\usepackage{algorithm}
\usepackage{algpseudocode}
\usepackage{enumitem}\usepackage[margin=1in]{geometry}

\hypersetup{colorlinks=true,citecolor=blue,linkcolor=blue,hypertexnames=false}
\allowdisplaybreaks

\newcommand{\paperTitle}{A General Framework for Metropolis-Adjusted Dikin Walks:\\
Dimension-Square Mixing on Polytopes and\\Log-Det Walks on Spectrahedra%
}

\newcommand{\paperAuthor}{Zhao Song\thanks{Independent researcher.
  Email: \texttt{magic.linuxkde@gmail.com}.}
\and
Lichen Zhang\thanks{Massachusetts Institute of Technology.
  Email: \texttt{lichenz@mit.edu}.}}

\DeclareMathOperator{\nnz}{nnz}

\DeclareMathOperator{\tr}{tr}
\DeclareMathOperator{\Diag}{Diag}
\DeclareMathOperator{\diag}{diag}
\DeclareMathOperator{\lev}{lev}
\DeclareMathOperator{\Lip}{Lip}
\DeclareMathOperator{\gap}{gap}
\DeclareMathOperator{\dist}{dist}

\newcommand{\R}{\mathbb{R}}
\newcommand{\E}{\mathbb{E}}
\newcommand{\Sd}{\mathbb{S}^{d}}
\newcommand{\Sn}{\mathbb{S}^{n}}
\newcommand{\norm}[1]{\lVert #1\rVert}
\newcommand{\opnorm}[1]{\lVert #1\rVert_{\mathrm{op}}}
\newcommand{\fnorm}[1]{\lVert #1\rVert_{F}}
\newcommand{\snorm}[1]{\lVert #1\rVert_{S_{1}}}
\newcommand{\dtv}{d_{\mathrm{TV}}}
\newcommand{\intK}{K^{\circ}}
\newcommand{\mA}{A}
\newcommand{\wb}{\bar{w}}
\newcommand{\creg}{\mathcal{C}}
\newcommand{\Ktube}{\mathcal{K}_{0}}
\newcommand{\Ktubep}{\mathcal{K}_{0}^{+}}
\newcommand{\Lb}{\bar{\Lambda}}
\newcommand{\Rr}{\mathcal{R}}
\newcommand{\Rres}{\mathsf{R}}% Part II resolvent (I+Delta)^{-1}
\newcommand{\Bh}{\widehat{B}}
\newcommand{\lsep}{\ell_{\mathrm{sep}}}
\newcommand{\eone}{\epsilon_{0}}
\newcommand{\Ot}{\widetilde{O}}

\newcommand{\ind}[1]{\mathbf{1}_{#1}}
\newcommand{\Ab}{\bar{A}}
\newcommand{\Psib}{\Psi}
\newcommand{\Eb}{\bar{E}}
\newcommand{\Nt}{\widetilde{N}}

\newcommand{\psistar}{\psi^{+}_{*}}
\newcommand{\nubar}{\bar{\nu}}
\newcommand{\Hld}{H_{\mathrm{ld}}}

\theoremstyle{plain}
\newtheorem{theorem}{Theorem}[section]
\newtheorem{lemma}[theorem]{Lemma}
\newtheorem{definition}[theorem]{Definition}

\newtheorem{corollary}[theorem]{Corollary}

\theoremstyle{definition}

\newtheorem{remark}[theorem]{Remark}

\begin{document}

\date{}
\title{\paperTitle}
\author{\paperAuthor}
\maketitle

\begin{abstract}
Sampling from a log-concave density over a convex body is a fundamental
problem.  The Metropolis-adjusted Dikin walk draws Gaussian proposals shaped by
a local barrier metric, and the main difficulty in its analysis is to control
the proposal determinant and the reverse quadratic form at a random endpoint.
In this paper, we analyze these walks with exact metrics by keeping the two
terms together.  In the complete logarithmic acceptance ratio, their leading
uncentered terms cancel, and second-order tools control the remaining centered
fluctuations.  For a polytope given by $n$ inequalities and a convex
$L$-Lipschitz potential, we show that the regularized Lee--Sidford walk mixes
from a warm start in $\Ot((d^{2}+dL^{2}R^{2})\log(w/\delta))$ steps.  For a
spectrahedron with $n\times n$ blocks, we show that the log-det walk scaled by
$\kappa=\Ot(\psistar)$ mixes in $\Ot((\psistar nd+dL^{2}R^{2})\log(w/\delta))$
steps, where $\psistar$ measures matrix leverage.  Both analyses share an
acceptance-to-mixing reduction.  Moreover, we give exact-arithmetic
implementations.  For polytopes, a proposal-comparison argument transfers the
bound to an appropriately padded $O(1/d)$-accurate metric computed from
high-precision Lewis weights.  For spectrahedra, given
$\widehat\psi\ge\psistar$, our implementation with the exact or a
TensorSRHT-sketched log-det metric obtains the same bound with $\psistar$
replaced by $\widehat\psi$.

\end{abstract}

\tableofcontents

%%% This file is the structure for main body content
%%% This file should only contain use \input{xxx}.
%%% TeX files for body contents should be named as:
%%% 01_xxxx.tex
%%% 02_xxxx.tex
%%% ...
%%% 49_xxxx.tex

\section{Introduction}
\label{sec:introduction}

Sampling from a log-concave density constrained to a convex body is a fundamental problem in theoretical computer science, statistics, and optimization.  Dikin walks approach this task by adapting
a Gaussian proposal to a local barrier metric, which allows the walk to exploit
the geometry of the feasible region.  For nearly universal barriers, this
affine-invariant approach can also avoid an artificial dependence on the number
of constraints.

Kannan and Narayanan introduced the logarithmic-barrier Dikin
walk~\cite{kn12}.  Chen, Dwivedi, Wainwright, and Yu then introduced the Vaidya
and John walks, whose weighted barriers reduce the dependence on the number of
inequalities~\cite{cdwy18}, and they conjectured that an improved John walk
mixes in $\Ot(d^{2})$ iterations from a warm
start~\cite[Conjecture~5]{cdwy18}.  Laddha, Lee, and Vempala developed strong
self-concordance and symmetry as a framework for this question~\cite{llv20},
and Kook and Vempala isolated average self-concordance as the condition that
controls the reverse local norm at a random proposal~\cite{kv24}.  With the
$\sqrt d$-scaled Lewis-weight metric, Kook and Vempala obtain
$\Ot(d^{5/2}\log(w/\delta))$ warm-start mixing for uniform and exponential
targets~\cite[Theorems~3.1 and~3.4, Lemma~E.3]{kv24}, matching the John-walk
rate of~\cite{cdwy18}, which assumes $n\le\exp(\sqrt d)$.

For convex $L$-Lipschitz potentials, Mangoubi and Vishnoi add a multiple of
the identity to the log-barrier Hessian~\cite{mv23}, a soft-threshold
regularization that they first introduced in~\cite{mv21}.  In their later
implementation, each step takes $\Ot(\nnz(\mA)+d^{2})$ arithmetic operations
plus the cost of evaluating $f$, where $\nnz(\mA)$ counts the nonzero
constraint coefficients~\cite[Theorem~2.1]{mv24}.  \cite{gkmsz24}
extends this regularization beyond the log barrier and uses its Lee--Sidford
form.  Our metric, the scaled Lee--Sidford metric plus
$L^{2}I$, also has this form.  For smooth potentials, Jiang and Chen analyze
regularized Dikin walks, including a regularized Lewis-weight
metric~\cite{jc25}.

However, adapting the proposal to the local metric also creates the main
analytic difficulty.  The Metropolis filter compares a forward Gaussian defined
at the current point with a reverse Gaussian defined at the proposal.
Therefore, it is not enough to show that the forward step remains in a local
ellipsoid, since we must also control the determinant and the reverse local
quadratic form at the random endpoint.  Indeed, the original dimension-square
claim of Laddha, Lee, and Vempala leaves open this reverse-ellipsoid step, and
a direct repair, reported in~\cite{k26} and attributed there to Vempala
(personal communication), gives a dimension-cube bound.

In the Lipschitz setting, \cite{gkmsz24} claimed
$\Ot(d^{2}+dL^{2}R^{2})$ mixing for a robust Lee--Sidford walk, and Kook
subsequently identified a gap in the proof of their Lemma~D.5~\cite{k26}.
Specifically, while closeness of the endpoint Lewis weights gives only
$O(r^{2})$ control of the squared-local-norm change, the
average-self-concordance argument requires $O(r^{2}/d)$.  Moreover, since the
endpoint weights depend on the Gaussian proposal direction, the
fixed-coefficient polynomial argument for the log
barrier~\cite[Proposition~7]{sv16} does not transfer
directly~\cite[footnote~4]{k26}.  We note that this gap lies in the proof and is not a counterexample to
the dimension-square theorem.  Motivated by these difficulties, we ask the
following question:
\begin{center}
\emph{Can we control the determinant and the reverse quadratic form at the
random endpoint directly, without estimating them separately?}
\end{center}

In this paper, we answer this question affirmatively.  Our key idea is to keep
the complete logarithmic acceptance ratio intact.  Within this ratio, the
linear part of the log determinant cancels the uncentered mean of the reverse
quadratic term, and we control the centered fluctuations that remain by
Gaussian integration by parts and second-order concentration.  Consequently,
we obtain direct analyses for Lee--Sidford geometry on polytopes and for
log-det geometry on spectrahedra.

\subsection{Our results}

We first state our result for polytopes.

\begin{theorem}[Polytopes, informal version of Theorem~\ref{I:thm:main-formal}]
\label{thm:intro-polytope}
Let $K\subseteq B_{2}(0,R)$ be a bounded full-dimensional polytope in $\R^{d}$
given by $n$ linear inequalities, and let $\pi\propto e^{-f}\ind{K}$ with $f$
convex and $L$-Lipschitz on $K$.  The $\tfrac12$-lazy
Metropolis-adjusted Dikin walk with the exact Lee--Sidford metric, scaled by a polylogarithmic factor
and regularized by $L^{2}I$, reaches total-variation distance $\delta$ from
$\pi$ from a $w$-warm start in
\[
  \Ot((d^{2}+dL^{2}R^{2})\log(w/\delta))
\]
steps.
\end{theorem}

Let us interpret Theorem~\ref{thm:intro-polytope} by comparing it with prior
bounds for Lipschitz potentials.  For a general convex $L$-Lipschitz
potential, Mangoubi and Vishnoi obtain $O((nd+dL^{2}R^{2})\log(w/\delta))$
steps from a $w$-warm start~\cite[Theorem~2.1 of arXiv v2]{mv23}, a bound
that~\cite{mv24} matches up to polylogarithmic factors.  In comparison, our
bound improves on that of~\cite{mv23,mv24} by the factor
$(n+L^{2}R^{2})/(d+L^{2}R^{2})$, up to polylogarithmic factors.  This gain is
real when $n$ exceeds $\max\{d,L^{2}R^{2}\}$ by more than polylogarithmic
factors, and there is no gain when $n=\Ot(d)$ or $L^{2}R^{2}\gtrsim n$.  With a
barrier of parameter $\Ot(d)$, the conductance bound of~\cite{nr17} gives
$\Ot((d^{5}+d^{3}L^{2}R^{2})\log(w/\delta))$ steps~\cite[footnote~1 of arXiv v2]{mv23},
which is smaller than the bound of~\cite{mv23} only when $n\gg d^{4}$.
Moreover, apart from~\cite{mv23,mv24} and~\cite{gkmsz24}, the prior polytope
bounds in Table~\ref{tab:result_history} cover only uniform or exponential
targets, or require a smooth potential.

To our knowledge, we give the first complete proof of the rate
$\Ot((d^{2}+dL^{2}R^{2})\log(w/\delta))$ for convex $L$-Lipschitz potentials.
\cite[Theorem~1.2]{gkmsz24} claimed this rate, but its proof
contains the gap that Kook identified~\cite{k26}.  Our polytope proof
addresses the same difficulty with the Metropolis filter without proving average
self-concordance, and instead controls the complete acceptance ratio after the
determinant--quadratic cancellation.  We first prove the main theorem for the
exact metric, and then, in Appendix~\ref{I:sec:approx}, we transfer the
acceptance and mixing bounds to a padded metric constructed from
high-precision approximate Lewis weights.

Our next result concerns spectrahedra.

\begin{theorem}[Spectrahedra, informal version of Theorem~\ref{II:thm:main-formal}]
\label{thm:intro-spectral}
Let
\[
 K:=\{x\in\R^{d}:A_{0}+\sum_{i=1}^{d}x_{i}A_{i}\succeq0\}
\]
be contained in $B_{2}(0,R)$, bounded, full-dimensional, and strictly
feasible after minimal-face reduction, with $n\times n$ matrix blocks, and let $\pi\propto e^{-f}\ind{K}$ with $f$
convex and $L$-Lipschitz on $K$.  For the log-det metric, let $\psistar$ be
the maximum of one and the supremum, over $x\in\intK$, of the spectral norm of
the whitened matrix leverage operator at $x$.  The
$\tfrac12$-lazy Metropolis-adjusted Dikin walk with the exact log-det metric, scaled by
$\kappa=\Ot(\psistar)$ and regularized by $L^{2}I$, reaches total-variation
distance $\delta$ from $\pi$ from a $w$-warm start in
\[
  \Ot((\psistar nd+dL^{2}R^{2})\log(w/\delta))
\]
steps.
\end{theorem}

Next, we compare Theorem~\ref{thm:intro-spectral} with prior bounds.  For the uniform target on a spectrahedron, Narayanan gives
$O(nd^{2}\log(1/\delta))$ mixing from an $O(1)$-warm start with the log-det
Hessian scaled by $d$~\cite{n16}, and Kook and Vempala recover this bound
through their Theorem~3.1 for a polytope intersected with a hyperbolic
cone~\cite[Section~3.5]{kv24}.  Specifically, for $m$ linear constraints and a
barrier of parameter $\nu_{h}$, their bound $O(d(m+d\nu_{h}))$ becomes
$O(nd^{2})$ from a warm start when $m=0$ and $\nu_{h}=n$.  For the uniform
target ($L=0$), the bound in our proof is $\psistar nd$ times a factor of order
$\log^{2}(nd)$, with $1\le\psistar\le d$.  Hence, our theorem improves on
these bounds only when $\psistar\log^{2}(nd)\ll d$.  In addition, on an arrow spectrahedron with $\psi(0)=\Theta(d)$ and the uniform target, the walk with $\kappa=1$ and $L=0$ can have exponentially small acceptance at the center when the proposal radius is a fixed universal constant.  We
remark that this example concerns acceptance and does not give a lower bound
on the mixing time.

Now that we have mixing bounds for the walks with the exact metrics, we turn
both analyses into exact-arithmetic implementations in the appendix, by
comparing the implemented proposals with the exact proposals.  For polytopes,
we obtain a suitably padded $O(1/d)$-accurate metric from deterministic
high-precision Lewis weights.  For spectrahedra, independent TensorSRHT
sketches at the current point and at the proposal give an exactly reversible
chain.  In both settings, the proposal comparison transfers the corresponding
mixing bound without differentiating the approximate metric.  We summarize the
resulting bounds in the following theorem, and we state the formal arithmetic
result for polytopes in Corollary~\ref{I:cor:arithmetic-runtime}.

\begin{theorem}[Approximate metrics and exact-arithmetic complexity, informal]
\label{thm:intro_approximate_metrics}
In each setting, assume the hypotheses of Theorem~\ref{I:thm:main-formal} or
Theorem~\ref{II:thm:main-formal}, respectively, and
work in the exact-arithmetic model of Appendix~\ref{I:sec:approx}.  Let $\omega$
denote the exponent of square-matrix multiplication.  Starting from an exact
draw $X_0\sim\mu_0$, where $\mu_0$ is $w$-warm, the following implementations
reach total-variation error $\delta$.
\begin{enumerate}[label=(\roman*),leftmargin=2.4em,itemsep=2pt]
\item For polytopes, deterministic $O(1/d)$-relative Lewis-weight
approximation with one-sided metric padding gives an exactly reversible walk
using
$
 \Ot(
 (d^2+dL^2R^2)(nd^{\omega-1}+d^\omega)
 \log(e w/\delta))
$
arithmetic operations and
$\Ot((d^2+dL^2R^2)\log(e w/\delta))$ evaluations of $f$.

\item For spectrahedra, given a supplied upper bound
$\widehat\psi\ge\psistar$, the implementation of
Corollary~\ref{II:thm:sketched-arithmetic}, which uses either the exact metric
or independent sketches at both endpoints, gives
an exactly reversible walk using
\[
 \Ot(
 (\widehat\psi nd+dL^2R^2)
 (n^\omega+n^2d^{\omega-1}+\min\{dn^\omega,n^2d^{3\omega-5}\})
 \log(e w/\delta))
\]
arithmetic operations and
$\Ot((\widehat\psi nd+dL^2R^2)\log(e w/\delta))$ evaluations of $f$.
The unconditional choice $\widehat\psi=d$ replaces
$\widehat\psi nd$ by $nd^2$.
\end{enumerate}
These bounds count arithmetic operations and do not account for bit
complexity.  In particular, they exclude finite-bit conditioning, rounded
Gaussian sampling, inexact Metropolis decisions, and construction of the warm
start.
\end{theorem}

\subsection{Related work}

In Table~\ref{tab:result_history}, we place our two results in context.  We
list the target class because the uniform, exponential, and general Lipschitz
settings are not interchangeable, and a dagger marks a result whose published
proof has a gap, identified in~\cite{k26} or in this paper.

\begin{table}[t]
\caption{Selected warm-start bounds for Dikin-type walks.  A $\dagger$ marks a result whose published proof contains a known gap.  Here $n$ denotes
the number of inequalities for a polytope or the matrix size for a
spectrahedron, and $\tau=\log(w/\delta)$.  The iterations column suppresses
$O$ and $\widetilde O$.}
\label{tab:result_history}
\centering
\footnotesize
\setlength{\tabcolsep}{3pt}
\begin{tabular}{@{}>{\raggedright\arraybackslash}p{0.125\textwidth}>{\raggedright\arraybackslash}p{0.18\textwidth}>{\raggedright\arraybackslash}p{0.31\textwidth}>{\raggedright\arraybackslash}p{0.14\textwidth}>{\raggedright\arraybackslash}p{0.19\textwidth}@{}}
\toprule
Reference & Target & Walk or metric & Iterations & Status \\
\midrule
\multicolumn{5}{@{}l}{\textbf{(a) Polytopes}}\\
\cite{kn12}
& Uniform
& Log-barrier Dikin walk
& $nd\tau$
& \\
\cite{cdwy18}
& Uniform
& John walk, $n\le\exp(\sqrt d)$
& $d^{5/2}\log^{4}(n/d)\tau$
& \\
\cite{llv20}\textsuperscript{$\dagger$}
& Uniform
& Lee--Sidford Dikin walk
& $d^{2}\tau$
& gap identified~\cite{k26}\\
\cite{mv23,mv24}
& $L$-Lipschitz
& Soft-threshold log barrier
& $d(n+L^{2}R^{2})\tau$
& \\
\cite{kv24}
& Exponential
& $\sqrt d$-scaled Lewis-weight metric
& $d^{5/2}\tau$
& \\
\cite{gkmsz24}\textsuperscript{$\dagger$}
& $L$-Lipschitz
& Robust Lee--Sidford walk
& $d(d+L^{2}R^{2})\tau$
& gap identified~\cite{k26}\\
\cite{jc25}
& $\alpha$-strongly convex, $\beta$-smooth
& Regularized Lewis-weight metric
& $d(d^{3/2}+\beta/\alpha)\tau$
& \\
\cite{k26}\textsuperscript{a}
& Exponential
& $\Ot(d^{1/4})$-scaled Lee--Sidford metric
& $d^{9/4}\tau$
& preprint\\
\cite{lu26}\textsuperscript{a}
& Exponential
& Polylog-scaled Lee--Sidford metric
& $d^{2}\tau$
& preprint, concurrent\\
\cite{ck26}\textsuperscript{a}
& Exponential
& Polylog-scaled Lee--Sidford metric
& $d^{2}\tau$
& preprint, concurrent\\
Theorem~\ref{thm:intro-polytope}
& $L$-Lipschitz
& Exact regularized polylog-scaled Lee--Sidford metric
& $d(d+L^{2}R^{2})\tau$
& this paper\\
\midrule
\multicolumn{5}{@{}l}{\textbf{(b) Spectrahedra}}\\
\cite{n16}\textsuperscript{b}
& Uniform
& $d$-scaled log-det Hessian
& $nd^{2}\tau$
& \\
\cite{kv24}\textsuperscript{c}
& Uniform
& $d$-scaled log-det Hessian
& $nd^{2}\tau$
& \\
\cite{gkmsz24}\textsuperscript{$\dagger$}
& $L$-Lipschitz
& Log-det metric
& $d(n+L^{2}R^{2})\tau$
& gap identified here\\
Theorem~\ref{thm:intro-spectral}
& $L$-Lipschitz
& Exact regularized log-det metric scaled by $\kappa=\Ot(\psistar)$
& $d(\psistar n+L^{2}R^{2})\tau$
& this paper\\
\bottomrule
\end{tabular}

\smallskip
{\raggedright
Exponential targets include the uniform one.
\textsuperscript{a}\,A chi-square bound from any start with finite chi-square
divergence, converted with $\chi^{2}\le w-1$ and
$\dtv\le\tfrac12\sqrt{\chi^{2}}$.  For~\cite{ck26} it is Corollary~4.
\textsuperscript{b}\,The source assumes an $O(1)$-warm start, so here
$\tau=O(\log(1/\delta))$.
\textsuperscript{c}\,Our instantiation, with $m=0$ and $\nu_{h}=n$, of the
bound $d(m+d\nu_{h})$ of~\cite[Section~3.5]{kv24} for $m$ linear constraints
and a hyperbolic cone with barrier parameter $\nu_{h}$.\par}
\end{table}

For exponential targets $\pi\propto e^{-\theta^{\top}x}\ind{K}$, which include
the uniform target, Kook gives an $\Ot(d^{9/4})$ bound by scaling the
Lee--Sidford metric by $\Ot(d^{1/4})$~\cite[Theorem~1.1]{k26}.  This bound
allows any start with finite chi-square divergence and has no term in $\theta$
or $R$.

For log-concave targets with Lipschitz potentials, Narayanan and Rakhlin prove
a conductance bound for a Dikin walk with a general self-concordant barrier~\cite[Section~3.1 and Lemma~4]{nr17}. However, they do not instantiate it for spectrahedra.  With the log-det barrier,
their bound gives $\Ot(n^{2}d^{3}+n^{2}dL^{2}R^{2})$ steps, as tabulated
in~\cite[Table~2]{gkmsz24}.  Truncated PSD cones, that is, PSD cones cut by
halfspaces, have their own bounds~\cite[Proposition~F.1]{kv24},
\cite[Corollary~3]{ck26}, which Table~\ref{tab:result_history}(b) omits.

We now turn to per-step costs, which are secondary in this paper.  For
polytopes, \cite[Theorems~1.1 and~1.2]{gkmsz24} states
$\Ot(nd^{\omega-1})$ for an approximate Lee--Sidford Hessian plus $d^{\omega}$
per step, which is the per-step factor of
Theorem~\ref{thm:intro_approximate_metrics}(i).  Laddha, Lee, and Vempala state
$\Ot(nd^{\omega-1})$ per step~\cite[Theorem~1.4 of arXiv v2]{llv20}, a factor
that Kook also notes informally~\cite{k26}.  For spectrahedra,
\cite[Theorem~1.3]{gkmsz24} states $\Ot(n^{\omega}+n^{2}d^{3\omega-5})$
per step and records the cost $O(dn^{\omega}+n^{2}d^{\omega-1})$ of the exact
log-det Hessian~\cite[Section~2.4]{gkmsz24}.  Up to polylogarithmic factors,
the per-step factor of Theorem~\ref{thm:intro_approximate_metrics}(ii) is the
smaller of these two.  To the best of our knowledge, what is new here is a
complete proof that the walks with approximate metrics mix at the rates of
Theorem~\ref{thm:intro_approximate_metrics}.

Since our polytope implementation uses a deterministic metric, the Metropolis
filter alone makes it exactly reversible.  For the sketched spectrahedral
metric, exact reversibility comes from the independent-sketch construction
of~\cite[Lemma~D.12]{gkmsz24}.  However, in total arithmetic,
Theorem~\ref{thm:intro_approximate_metrics}(i) does not always improve
on~\cite{mv24}.  When $LR\le\sqrt d$ and $f$ costs $O(1)$ to evaluate,
Theorem~\ref{thm:intro_approximate_metrics}(i) uses fewer operations only when
$\nnz(\mA)\gg d^{\omega}$, up to polylogarithmic factors.  When
$L^{2}R^{2}\gtrsim n$, the walk of~\cite{mv24} uses fewer operations, again up
to polylogarithmic factors.  Thus, the gain of Theorem~\ref{thm:intro-polytope}
lies in the number of steps and evaluations of $f$, and it does not hold
uniformly for arithmetic.

\paragraph{Independent and concurrent works.}
Lu~\cite{lu26} and Chen and Kook~\cite{ck26} improve the bound of~\cite{k26} to $\Ot(d^{2})$ for
exponential targets on polytopes.  Lu shows that the $\tfrac12$-lazy walk with
a polylogarithmically scaled Lee--Sidford metric mixes in $\Ot(d^{2})$
iterations in chi-square divergence~\cite[Theorem~1.3]{lu26}.  The proof
of~\cite{lu26} controls the reverse Gaussian proposal through a quantitative
average self-concordance estimate at the $O(1/d)$ scale, derived from uniform analytic derivative bounds for the Lewis-weight subspace and an all-order contraction inequality for Gaussian chaoses.  Chen and Kook give a general total-variation bound under strong self-concordance, symmetry, and a mixed-trace condition on the metric, which yields $\Ot(d^{2})$ total-variation bounds for the Lee--Sidford, Lewis-weight, and John
metrics~\cite[Theorem~1 and Corollary~2]{ck26}.  The key idea of this first
framework is to control the acceptance probability on a high-probability region
rather than at every point.  The second framework of~\cite{ck26} adds
lower-trace self-concordance and a fourth-order bootstrap condition, and it
gives pointwise acceptance control and an $\Ot(d^{2})$ chi-square bound for a
suitably scaled Lee--Sidford metric~\cite[Theorem~2 and Corollary~4]{ck26}.
The bounds
of~\cite{lu26,ck26} allow any start with finite chi-square divergence and have
no term in $\theta$ or $R$.  In comparison, for these targets we apply
Theorem~\ref{thm:intro-polytope} with $L=\norm{\theta}_{2}$ and obtain
$\Ot((d^{2}+d\norm{\theta}_{2}^{2}R^{2})\log(w/\delta))$, which is worse than
$\Ot(d^{2})$ when $\norm{\theta}_{2}R\gg\sqrt d$.

While these results are stronger for exponential targets, our approach differs from those of~\cite{lu26,ck26} in both technique and
scope.  Unlike~\cite{lu26} and the second framework of~\cite{ck26}, we do not
prove average self-concordance.  Unlike the first framework of~\cite{ck26}, we
prove a uniform lower bound on the acceptance probability at every interior
point, by controlling the complete logarithmic acceptance ratio after the
determinant--quadratic cancellation.  Moreover, our
analysis applies to general convex $L$-Lipschitz potentials and to walks with
approximate metrics in exact arithmetic
(Theorem~\ref{thm:intro_approximate_metrics}), which~\cite{lu26,ck26} do not
consider.  It also applies to general spectrahedra with the log-det walk
(Theorem~\ref{thm:intro-spectral}), while~\cite{ck26} treats truncated PSD
cones, a special class of spectrahedra, with a hybrid barrier~\cite[Corollary~3]{ck26}.

\subsection{The spectrahedral analysis of~\texorpdfstring{\cite{gkmsz24}}{GKMSZ24}}

Besides the polytope rate, \cite{gkmsz24} separately states
$\Ot((nd+dL^{2}R^{2})\log(w/\delta))$ mixing for spectrahedra, a case that
Kook's note does not discuss.  Its spectrahedral proof has a gap of its own,
which we identify here.  In the notation of~\cite{gkmsz24}, where a check marks
their symbols, this proof runs the same generic lemmas for the unscaled log-det
metric with step scale $\check\alpha=\Theta(1/d)$.  However, two of its inputs
fail for the log-det barrier.

The first input, used in their Lemmas~D.3 and~D.4, is their
Assumption~B.1(iii), namely the squared bound
$\|\check H^{-1/2}\nabla\log\det\check H\|_{2}^{2}\le\Ot(d)$ on their barrier
Hessian $\check H$.  Their Lemma~G.9 proves only the unsquared bound for the
log-det barrier, and the squared bound fails for this barrier.  Indeed, on our
arrow spectrahedron (Section~\ref{II:sec:tight}), we show that the squared
quantity equals $\frac23(d+\frac43)^{2}$ at every point $x$ with
$\norm{x}^{2}=1$ (Lemma~\ref{II:lem:arrow-gkmsz}(iii)).

The second input is their Lemma~D.5, which the theorem uses through their
Lemma~D.8.  The proof sketch of Lemma~D.5 treats only the log
barrier and Lewis weights and does not address the log-det metric.  With
$f\equiv0$, the arrow spectrahedron satisfies the hypotheses of their
Theorem~1.3.  On this instance, the conclusion of Lemma~D.5 is false for the
log-det metric with the identity term that the lemma assumes, and the
conclusion of Lemma~D.8 is false as well, for the walk that its proof writes
out.  In their walk, each proposal comes from a Gaussian shaped by the metric
$\check\Phi=\check\alpha^{-1}\check H+\check\eta^{-1}I$, where $\check\eta$ is
their regularization parameter.  Their Lemma~D.5 assumes that the second term,
an added identity term, is large, namely $\check\eta\le0.001/d$.  The
acceptance bound that we prove covers this identity term.  Specifically, we
show that at the center of the arrow, the log-det walk for the uniform target
accepts a proposal with probability at most $e^{-3\mathcal{J}/16}+e^{-d/54}$
for $d\ge8$ (Lemma~\ref{II:lem:arrow} and Remark~\ref{II:rem:arrow}), where
$\mathcal{J}=d^{2}\check\alpha\check\eta^{2}/(\check\alpha+\check\eta)^{2}$ is
of order $d$ whenever $\check\alpha$ and $\check\eta$ are both of order $1/d$.

We now describe how the conclusion of Lemma~D.5 fails on this instance.  The local norm of a step $v$ at a point
$x$ is its length measured by the metric at $x$, given by
$\norm{v}_{\check\Phi(x)}=(v^{\top}\check\Phi(x)v)^{1/2}$.  Their Lemma~D.5
measures the squared local norm of the proposed step twice, once at its start
and once at the proposal, and it states that the two values differ by at most
$0.01$ with probability at least $0.999$.  Taking $\check\alpha=\check\eta=c/d$
with $c\le0.001$, we show that at the center of the arrow the difference is at
least $\frac{3cd}{8}(\norm{\xi}^{2}/d)^{2}$, where $\xi$ is the standard
Gaussian vector that generates the proposal.  Hence, the probability that the
difference is at most $0.01$ is below $0.999$ once $d\ge8/(75c)$, and it tends
to $0$ as $d$ grows (Lemma~\ref{II:lem:arrow-gkmsz}(i)).

Their Lemma~D.8 concerns the next position of the walk, started either at a
point $x$ or at a point $z$.  When the local norm of $z-x$ at $x$ is at most
$0.001$, the lemma bounds the total variation distance between the two
distributions of the next position by $0.99$.  Once $d\ge\max(10^{4},44/c)$,
the walk accepts a proposal from the center with probability at most $0.0096$
(Lemma~\ref{II:lem:arrow-gkmsz}(ii)).  Thus, the walk started at the center
stays there with probability above $0.99$, while the walk started at any other
point never lands exactly on the center.  Hence, the two distributions are
more than $0.99$ apart in total variation.  This separation holds
for the walk that the proof of Lemma~D.8 writes out, with the exact metric at
both endpoints.  However, our lemma does not cover the version of their
Definition~C.4, which places an approximate metric at the proposal.  In the
proof of Lemma~D.8, a step that rests on Lemma~D.5 would give acceptance at
least $0.0099$ from every point, but this step fails at the center.  Finally,
since $f\equiv0$ is $1$-Lipschitz, we may take the Lipschitz constant $L=1$,
and then the parameters $\check\alpha=\check\eta=c/d$ are the choice
$\check\alpha=\Theta(1/d)$ and $\check\eta=\Theta(1/(dL^{2}))$ of their
Algorithm~2.  For $c<0.001$, these parameters meet the hypotheses on
$\check\alpha$, $\check\eta$, and $L$ of every lemma used in the proof of their
Lemma~D.8.  As for their Assumption~B.1, part~(ii) holds on this instance and
part~(iii) fails.

We note that the preceding statements concern lemma conclusions and proof
steps at one point of one instance.  They do not show that the $\Ot(nd)$ rate
is false or that the walk fails to mix.  Indeed, the center carries tiny mass
under the uniform target, and Lemma~\ref{II:lem:arrow} is not a mixing lower
bound.  In summary, we identify a gap in the proof, but we do not give a
counterexample to the $\Ot(nd)$ rate.

We give a complete proof of Theorem~\ref{thm:intro-spectral}, which
establishes mixing for the log-det walk scaled by $\kappa=\Ot(\psistar)$ at
its own rate $\Ot((\psistar nd+dL^{2}R^{2})\log(w/\delta))$.  However, the
theorem does not establish the $\Ot(nd)$ rate, and it does not cover the
unscaled walk.  Moreover, the first term of its rate, $\psistar nd$, is larger
than the corresponding term $nd$ of the rate stated in~\cite{gkmsz24} by the
factor $\psistar$.

\paragraph{Organization.}
In Section~\ref{sec:technical_overview}, we explain the shared cancellation and
the transfer to implementable kernels, and we collect the common tools in
Section~\ref{I:sec:prelim}.  In Section~\ref{I:sec:polytope}, we prove the
polytope result via Lewis-weight stability, Stein concentration, and
conductance.  In Section~\ref{II:sec:spectral}, we prove the spectrahedral
result via the noncommutative proposal-ratio identity, a good set of steps controlled by matrix leverage, and conductance,
and we conclude with the arrow spectrahedron.  In
Appendix~\ref{I:sec:approx}, we give the formal approximate-metric kernels and
their exact-arithmetic complexity bounds.

\section{Technical overview}
\label{sec:technical_overview}

In this section, we treat the two proof mechanisms separately.
We explain the Lewis-weight argument for Theorem~\ref{thm:intro-polytope} in Section~\ref{subsec:tech_polytope} and the noncommutative log-det argument for Theorem~\ref{thm:intro-spectral} in Section~\ref{subsec:tech_spectrahedron}.
Then, in Section~\ref{subsec:tech_implementation}, we show how approximate metrics give implementable versions of both walks.

\subsection{Proof overview for Theorem~\ref{thm:intro-polytope}}
\label{subsec:tech_polytope}

\paragraph{Exact cancellation.}
Let $H(x)=A_x^\top W_x^{1-2/p}A_x$ and $G(x)=\kappa H(x)+L^2I$.
We first fix the current point and apply an affine change of coordinates that makes $H(x)=I$ and writes $G(x)=\kappa(I+\creg)$, where $\creg\succeq0$.
Thus, in these coordinates, a proposal has the form $y=x+\eta h$ with $h\sim N(0,Q)$, where $Q=(I+\creg)^{-1}\preceq I$ and $\eta=r/\sqrt{\kappa d}$.
With $U=W_x^{1/2-1/p}A_x$, which satisfies $U^\top U=I$, we can represent the endpoint metric exactly as $H(y)=I+U^\top\Diag(b)U$, where each coordinate of $b$ records the combined change in the corresponding slack and Lewis weight.
Setting $c_i=(u_i^\top h)^2$, $\wb_i=\E[c_i]$, and $E=Q^{1/2}(H(y)-I)Q^{1/2}$, we obtain the following identity from Lemma~\ref{I:lem:ledger}:
\[
 \log\frac{q_y(x)}{q_x(y)}
 =-\frac12(c-\wb)^\top b
  +\frac12(\log\det(I+E)-\tr[E]).
\]
The main cancellation of our analysis takes place in this identity.
Since the linear determinant term $\tr[E]=\wb^\top b$ is exactly the Gaussian mean of the uncentered quadratic term with $b$ held fixed, it disappears before we apply any tail bound, and we are left with a centered contraction and two second-order remainders.
We summarize the main steps of the argument for the Dikin walk on polytopes in Figure~\ref{fig:tech-polytope}.

\begin{figure}[!ht]
\centering
\includegraphics[width=\textwidth]{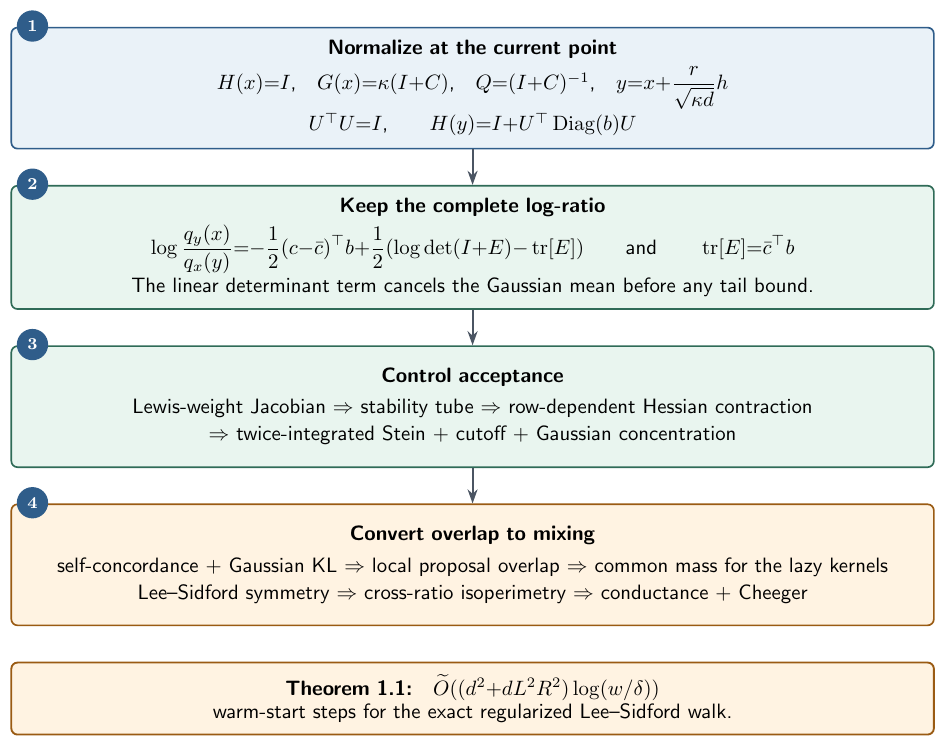}
\caption{Main steps for Theorem~\ref{thm:intro-polytope}: normalization, exact determinant--quadratic cancellation, Lewis-weight stability and acceptance, and conductance and mixing.}
\label{fig:tech-polytope}
\end{figure}

\paragraph{Lewis-weight stability and acceptance.}
It remains to bound the centered contraction and the two remainders.
Since the Lewis weights are implicit, differentiating their fixed-point equation directly would obscure the cancellation.
To circumvent this issue, we specialize the row-rescaling Jacobian of Lee and Sidford~\cite[Lemma~24]{ls19} to spatial slack changes in Lemma~\ref{I:lem:lewis-derivative}.
Using the resulting explicit identity, we show that with high probability the step lies in a good set on which all slacks and weights remain comparable to their values at $x$ (Lemmas~\ref{I:lem:tube-probability} and~\ref{I:lem:stability}).
Moreover, rowwise resolvent estimates give the row-dependent Hessian contraction in Lemma~\ref{I:lem:hessian-contraction}, a stronger, leverage-sensitive bound that prevents a factor depending on the number of inequalities from entering the fluctuation estimate.

Next, we split the centered term into its linear part $Z=(c-\wb)^\top a$ and the nonlinear remainder generated by $b-2a$.
By applying Gaussian integration by parts twice, we control $\E[Z]$ through the Hessian contraction, while Lemmas~\ref{I:lem:remainder-quad} and~\ref{I:lem:remainder-det} control the quadratic and determinant remainders.
Since $Z$ has a dimension-free Lipschitz constant on the good set, Gaussian concentration~\cite{l01} after a cutoff and McShane extension gives the required lower tail without conditioning the Gaussian law.
Feasibility follows from the Dikin containment of the unit $H(x)$-ellipsoid, and we bound the potential contribution using the facts that $G(x)\succeq L^2I$ and $f$ is $L$-Lipschitz.
Choosing $\rho=r/\sqrt\kappa$ as in Eq.~\eqref{I:eq:rho-choice} makes every remainder universally small. This proves the uniform acceptance statement of Lemma~\ref{I:lem:acceptance}.

\paragraph{From overlap to mixing.}
Now that we have uniform acceptance, we turn to mixing.
Strong self-concordance of the regularized metric controls the distortion between nearby covariance matrices. Substituting this distortion into the Gaussian KL formula and applying Pinsker's inequality gives proposal overlap whenever $\norm{x-z}_{G(x)}=O(r/\sqrt d)$.
Combined with the acceptance event, this overlap yields a fixed common mass for the two lazy Metropolis kernels.
Moreover, the scaled Lee--Sidford metric is $\bar\nu_{\mathrm{LS}}$-symmetric with $\bar\nu_{\mathrm{LS}}=\Ot(d)$, and regularization changes the symmetry parameter only to $\nu_G\le\bar\nu_{\mathrm{LS}}+4L^2R^2$.
Therefore, local-norm separation implies chord cross-ratio separation, and the isoperimetric inequality of Lov\'asz and Vempala~\cite{lv07} converts it into the conductance bound $\Phi\ge \frac{c r}{\sqrt{d(\bar\nu_{\mathrm{LS}}+L^2R^2)}}$.
Finally, using the reversible Cheeger inequality~\cite{ls88} and the direct $L^2(\pi)$-to-total-variation calculation for a $w$-warm start, we obtain the bound of $\Ot((d^2+dL^2R^2)\log(w/\delta))$ steps in Theorem~\ref{thm:intro-polytope}.

\subsection{Proof overview for Theorem~\ref{thm:intro-spectral}}
\label{subsec:tech_spectrahedron}

\paragraph{Exact cancellation.}
For spectrahedra, we follow the same outline.
At the current point, slack congruence makes $S(x)=I$, and we then apply a further linear change of variables that gives a Frobenius-orthonormal frame $\{\Ab_a\}_{a=1}^d$ for the whitened constraint directions and $H(x)=I_d$.
As before, a proposal is $y=x+\eta h$ with $h\sim N(0,Q)$ and $\eta=r/\sqrt{\kappa d}$, where $G(x)=\kappa(I+\creg)$ and $Q=(I+\creg)^{-1}$.
Let $\Delta=\eta\sum_a h_a\Ab_a$ and $\Rres=(I+\Delta)^{-1}$.
Using slack congruence and trace cyclicity, we obtain exact endpoint formulas for $H(y)-I=:\Eb$ and all of its first two derivatives.
Setting $E=Q^{1/2}\Eb Q^{1/2}$, we can write the identity of Lemma~\ref{II:lem:ledger} as
\[
 \log\frac{q_y(x)}{q_x(y)}
 =\frac12(\log\det(I+E)-\tr[E])-\frac12 Z(h),
 \qquad Z(h)=h^\top\Eb h-\tr[Q\Eb].
\]
Thus, the same determinant--quadratic cancellation is present, but the centered fluctuation is now a noncommutative trace polynomial.

\paragraph{The good set and acceptance.}
For the matrix leverage operator $\Psib(x)=\sum_a\Ab_a^2$, we use $\psi(x)=\opnorm{\Psib(x)}$ to denote its spectral norm and let $\psistar=\max(1,\sup_{x\in\intK}\psi(x))$.
Applying Tropp's self-adjoint Gaussian-series inequality~\cite[Theorem~4.1.1]{t15} to $\sum_a h_a\Ab_a$ gives a high-probability good set on which $\opnorm\Delta$ is small.
The resolvent is then stable on a slightly enlarged set, which makes the determinant remainder quadratic in $E$.
The key noncommutative estimate is the paired four-trace bound of Lemma~\ref{II:lem:four-trace}, which contracts the second derivatives arising from twice-integrated Stein to the single parameter $\psi(x)$ rather than a larger ambient-dimensional quantity.
As a result, the mean of $Z$ is $O(\rho^2\psi)$.

Moreover, the gradient formula for $Z$ and the same trace contractions give a dimension-free Lipschitz bound on the enlarged set, and hence we can control $Z-\E[Z]$ globally by Gaussian concentration after a cutoff and McShane extension.
By intersecting the concentration event with the good set, we also guarantee feasibility, since $I+t\Delta\succ0$ for $0\le t\le1$.
The potential term is handled by $G(x)\succeq L^2I$, exactly as in the proof for polytopes.
By choosing $\rho$ as in Eq.~\eqref{II:eq:rho-choice}, we obtain $\kappa=\Ot(\psistar)$ and make the determinant remainder, the Stein mean, and the fluctuation universally bounded. This proves the uniform acceptance lower tail in Lemma~\ref{II:lem:acceptance}.
Figure~\ref{fig:tech-spectrahedron} summarizes the main steps of the argument for the Dikin walk on spectrahedra.

\begin{figure}[!ht]
\centering
\includegraphics[width=\textwidth]{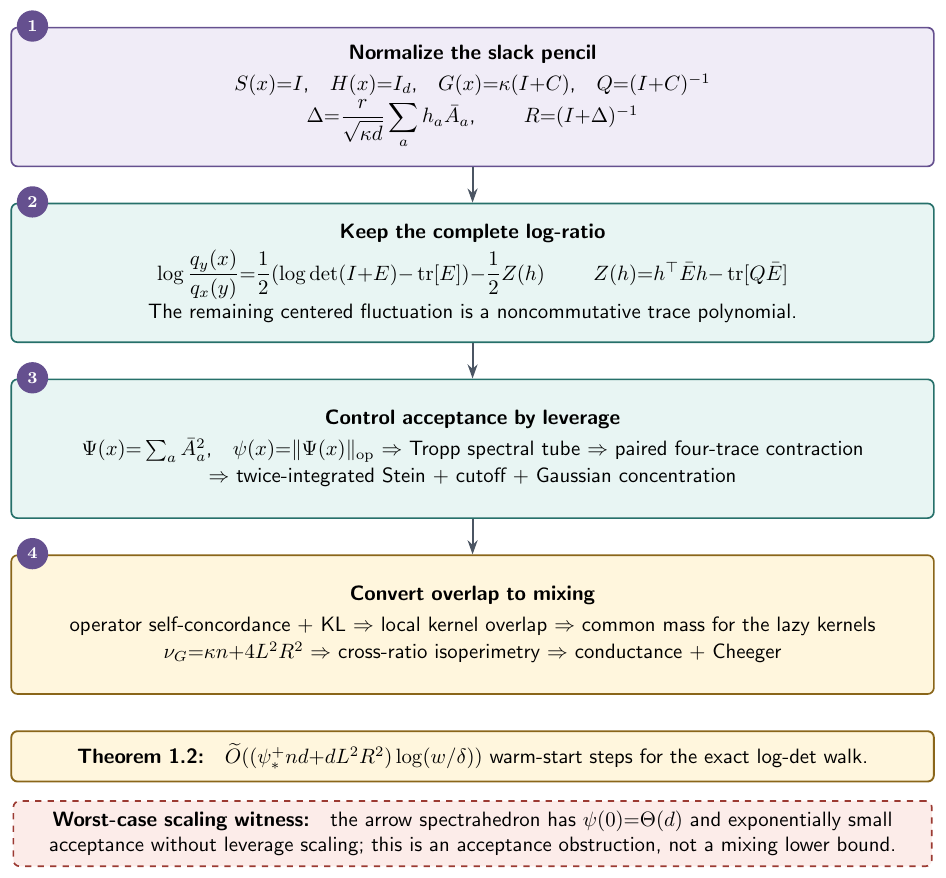}
\caption{Main steps for Theorem~\ref{thm:intro-spectral}: normalization, exact determinant--quadratic cancellation, leverage-controlled acceptance, and conductance and mixing.}
\label{fig:tech-spectrahedron}
\end{figure}

\paragraph{From overlap to mixing.}
Using operator self-concordance of the log-det metric to control nearby Gaussian covariances, we again obtain proposal overlap at local distance $O(r/\sqrt d)$ from the Gaussian KL identity and Pinsker's inequality.
The acceptance lower tail then transfers this overlap to common mass for the lazy kernels.
To bound the symmetry parameter, we define the block displacement $\Delta_{z-x}:=\sum_a(z_a-x_a)\Ab_a$ in normalized coordinates.
Then $z\in K\cap(2x-K)$ is equivalent to $\opnorm{\Delta_{z-x}}\le1$, while $\norm{z-x}_{H(x)}=\fnorm{\Delta_{z-x}}$.
Hence, $\fnorm{\Delta_{z-x}}^2\le n\opnorm{\Delta_{z-x}}^2$, and we obtain the symmetry parameter $\nu_G=\kappa n+4L^2R^2$ for the regularized metric.
With this parameter, cross-ratio isoperimetry and the same conductance argument as for polytopes yield $\Phi\ge\frac{c r}{\sqrt{d(\kappa n+L^2R^2)}}$.
The Cheeger inequality and warm-start contraction then give the bound of $\Ot((\psistar nd+dL^2R^2)\log(w/\delta))$ steps in Theorem~\ref{thm:intro-spectral}.

\paragraph{Why leverage scaling is necessary.}
Consider the uniform target and the arrow spectrahedron in Lemma~\ref{II:lem:arrow}, which has $\psi(0)=\Theta(d)$.
When $\kappa=1$, $L=0$ and the proposal radius is a fixed universal constant, the walk on this spectrahedron has exponentially small acceptance at the center.
This example explains why leverage scaling is needed for uniform acceptance. However, we make no claim of a lower bound on the mixing time.

\subsection{From exact metrics to implementable kernels}
\label{subsec:tech_implementation}

While we use the exact metric $G$ in both analyses above, an implementation constructs a positive-definite metric $\bar G(x;\omega_{\mathrm{s}})$, where we suppress the seed when the construction is deterministic, and proposes
\[
 Y\sim\bar q_{x,\omega_{\mathrm{s}}}
 :=N\!\left(x,\frac{r^{2}}{d}
       \bar G(x;\omega_{\mathrm{s}})^{-1}\right).
\]
If $Y$ is feasible, the algorithm constructs the reverse metric $\bar G(Y;\omega_{\mathrm{s}}')$ and accepts with probability
\[
 1\wedge
 \frac{e^{-f(Y)}\bar q_{Y,\omega_{\mathrm{s}}'}(x)}
      {e^{-f(x)}\bar q_{x,\omega_{\mathrm{s}}}(Y)}.
\]
Thus, the implemented chain applies the exact Metropolis filter to the metrics that were actually constructed. We note that it does not estimate the acceptance probability of the walk with the exact metric.

\paragraph{The two implementations.}
For polytopes, we compute the Lewis weights to coordinatewise relative accuracy $O(1/d)$ by a deterministic fixed-schedule iteration.
Inflating the result by its worst possible downward error gives a padded metric $\bar G(x)$ with $\bar G(x)\succeq G(x)$.
Since this metric is a deterministic measurable function of the state, the usual Metropolis--Hastings identity gives exact reversibility.
For spectrahedra, we use the direct log-det construction when $n<s$ and a TensorSRHT approximation when $n\ge s$, where $s$ is the sketch dimension of Eq.~\eqref{II:eq:sketch-size}.
In the sketched branch, we draw fresh and independent sketches at the forward and reverse endpoints.
After averaging over both seeds, we can write the accepted flow as
\[
 \E_{\omega_{\mathrm{s}},\omega_{\mathrm{s}}'}
 \!\left[
 \min\!\left\{
 e^{-f(x)}\bar q_{x,\omega_{\mathrm{s}}}(y),
 e^{-f(y)}\bar q_{y,\omega_{\mathrm{s}}'}(x)
 \right\}\right],
\]
which is symmetric under interchanging the endpoints and their seeds.
Hence, the marginal kernel of the sketched walk is exactly reversible.

\paragraph{Transfer from the walk with the exact metric.}
In both settings, either deterministically or on a pointwise sketch-success event, we have
\[
 G(x)\preceq\bar G(x;\omega_{\mathrm{s}})
 \preceq e^{\theta}G(x),
 \qquad \theta=O(1/d).
\]
By the centered Gaussian KL formula, $\dtv(q_x,\bar q_{x,\omega_{\mathrm{s}}})=O(\sqrt d\,\theta)$, while on the two-endpoint event from the acceptance proof for the exact metric, each conditional Gaussian log density changes by only $O(d\theta)=O(1)$.
Therefore, the acceptance lower tail and nearby-proposal overlap survive with different universal constants.
Since the conductance and warm-start arguments depend only on this common mass and the geometry of the exact metric, the mixing bounds transfer without differentiating the approximate metric.
In the randomized case, the sketch guarantee is only pointwise in the current state.
However, since failed seeds remain part of the transition kernel, we do not need a union bound over the mixing horizon.

In Appendix~\ref{I:sec:approx}, we give the formal kernels and perturbation arguments, and we combine the mixing times with the per-step costs $\Ot(nd^{\omega-1}+d^\omega)$ for polytopes and $\Ot(n^\omega+n^2d^{\omega-1}+\min\{dn^\omega,n^2d^{3\omega-5}\})$ for spectrahedra to obtain the exact-arithmetic bounds in Corollaries~\ref{I:cor:arithmetic-runtime} and~\ref{II:thm:sketched-arithmetic}.
 
\section{Preliminaries}\label{I:sec:prelim}\label{II:sec:prelim}

\subsection{Notation}

Throughout, we set
\[
p:=\max\{4,\;2(1+\log n)\},\qquad
\beta:=\tfrac12-\tfrac1p,\qquad
c_{p}:=1-\tfrac2p=2\beta .
\]
With this choice of $p$, the parameters satisfy both
$\beta\in[\tfrac14,\tfrac12)$ and $n^{2/p}\le e$. We use $C,c$ to denote
universal constants and $C_{p}$ to denote a quantity bounded by a fixed
polynomial in $p$, all of which may change from line to line. We use
the notation $\Ot(\cdot)$ to hide polynomial factors in $p$ and logarithmic
factors in $n,d$. For a symmetric matrix $M$, we use $\opnorm{M}$ and
$\fnorm{M}$ to denote its spectral and Frobenius norms, and $\Sd$ denotes
the space of symmetric $d\times d$ matrices with the Frobenius inner
product. For a positive definite $M$ and a positive vector $w$, we define
the weighted norms by $\norm{v}_{M}^{2}=v^{\top}Mv$ and
$\norm{v}_{W}^{2}=\sum_{i}w_{i}v_{i}^{2}$, where $W=\Diag(w)$. For vectors,
$\circ$ denotes the entrywise product and $v^{\circ2}=v\circ v$. For
matrices, $\circ$ is the Hadamard product.

\subsection{Polytope, slacks, and Lewis weights}

Given a vector $b\in\R^{n}$ and a matrix $\mA\in\R^{n\times d}$ with full
column rank and no zero rows, we define the sets
\[
K:=\{z\in\R^{d}:\mA z\ge b\},\qquad
\intK:=\{z\in\R^{d}:\mA z> b\}.
\]
We assume throughout that $K$ is bounded and full-dimensional and, after a
translation, that $K\subseteq B_{2}(0,R)$. For each point $x\in\intK$, we
define the positive slacks of the constraints and the slack-rescaled
constraint matrix by
\begin{equation*}
s_{x}:=\mA x-b,\qquad S_{x}:=\Diag(s_{x}),\qquad
A_{x}:=S_{x}^{-1}\mA .
\end{equation*}
The definition of the slacks fixes the sign convention, under which the slacks
of strictly feasible points are $\mA x-b>0$.

For a full-column-rank matrix $M$, let
$\lev(M):=\diag(M(M^{\top}M)^{-1}M^{\top})$ denote its vector of leverage
scores. The \emph{$\ell_{p}$-Lewis weights} of $A_{x}$ are the unique
positive vector $w_{x}$, with $W_{x}:=\Diag(w_{x})$, that satisfies the
fixed-point identity
\begin{equation}
w_{x}=\lev(W_{x}^{1/2-1/p}A_{x}).
\label{I:eq:lewis-fixedpoint}
\end{equation}
The existence and uniqueness of Lewis weights for $p\ge2$ go back to
Lewis~\cite{l78} (see also Cohen and Peng~\cite{cp15}). Besides these two
facts, the only external fact about Lewis weights that we use is one
Lewis-calculus identity. Specifically, under a smooth
positive diagonal rescaling of the rows, the unique Lewis weights vary
smoothly, and their row-rescaling Jacobian is the one from~\cite{ls19} that
we quote in Lemma~\ref{I:lem:lewis-derivative}. Finally, since the
associated leverage-score matrix is a rank-$d$ orthogonal projection, the
Lewis weights satisfy
\begin{equation}
0<w_{x,i}\le1,\qquad \sum_{i}w_{x,i}=d .
\label{I:eq:lewis-basic}
\end{equation}

\subsection{First-derivative Lewis calculus}\label{I:sec:lewis-calculus}

We use the following objects throughout. At a point $z\in\intK$, we set
\[
B_{z}=W_{z}^{1/2}A_{z},\qquad
V_{z}=W_{z}^{\beta}A_{z}H(z)^{-1/2},\qquad
\Bh_{z}=B_{z}H(z)^{-1/2}=W_{z}^{1/p}V_{z} .
\]
With these definitions, we have $V_{z}^{\top}V_{z}=I$, and by
Eq.~\eqref{I:eq:lewis-fixedpoint}, the $i$-th row $v_{z,i}$ of $V_{z}$ has
squared norm $w_{z,i}$. Using these rows, we define the row map
$\Phi_{z}:\Sd\to\R^{n}$ and its associated operators by
\begin{equation}
(\Phi_{z}X)_{i}=\frac{v_{z,i}^{\top}Xv_{z,i}}{\sqrt{w_{z,i}}},
\qquad
\Lb_{z}:=I-\Phi_{z}\Phi_{z}^{*},
\qquad
\Rr_{z}:=(I-c_{p}\Lb_{z})^{-1},
\qquad
N_{z}:=2\Lb_{z}\,\Rr_{z} .
\label{I:eq:frame-objects}
\end{equation}
Here $\Phi_{z}\Phi_{z}^{*}$ is an $n\times n$ matrix that depends only
on the Gram matrix of the rows $v_{z,i}$, and hence it is independent of
the choice of the square root $H(z)^{-1/2}$. Moreover, since
$2\beta=c_{p}$, these definitions give the following resolvent identity
\begin{equation}
\Rr_{z}=I+\beta N_{z}.
\label{I:eq:resolvent-identity}
\end{equation}

We repeatedly use two elementary facts about norms. The first is the rowwise
contraction
\begin{equation}
\sum_{i}\frac{(v_{z,i}^{\top}Xv_{z,i})^{2}}{w_{z,i}}
\le\sum_{i}v_{z,i}^{\top}X^{2}v_{z,i}=\fnorm{X}^{2}
\qquad(X\in\Sd),
\label{I:eq:phi-contraction}
\end{equation}
which follows from the Cauchy--Schwarz inequality in each row together
with $V_z^\top V_z=I$. In particular, the rowwise contraction implies
$\opnorm{\Phi_{z}}\le1$. The $i$-th component functional
of the row map has norm $\fnorm{(\Phi_{z})_{i}}=\sqrt{w_{z,i}}$. The second
fact is that the bound $A_{z}^{\top}W_{z}A_{z}\preceq H(z)$, which holds
because $w\le w^{1-2/p}$ for $w\le1$, implies $\opnorm{\Bh_{z}}\le1$.
Moreover, the rows satisfy $\norm{(\Bh_{z})_{i}}=w_{z,i}^{1/2+1/p}$.

\begin{lemma}[{Lewis derivative formula; \cite[Lemma 24]{ls19}}]
\label{I:lem:lewis-derivative}
The map $z\mapsto\log w_{z}$ is smooth on $\intK$ with
\begin{equation}
D\log w_{z}[v]=-W_{z}^{-1/2}N_{z}B_{z}v .
\label{I:eq:lewis-derivative}
\end{equation}
Moreover $0\preceq\Lb_{z}\preceq I$, and consequently
$\opnorm{\Rr_{z}}\le p/2$ and $0\preceq N_{z}\preceq pI$.
\end{lemma}

\begin{proof}
The only step that we import is the following row-rescaling Jacobian. We
spell out its specialization to spatial directions, since this
specialization fixes both the sign and the factor of two. Let
$\upsilon=s_z^{-1}$ be the vector that rescales the rows of the fixed
constraint matrix, and let $\Upsilon=\Diag(\upsilon)$. In the notation of
\cite[Lemma 24]{ls19}, the Jacobian of the Lewis weights under this
rescaling is
\[
 J_w(\upsilon)=2W(W-c_p\Lambda)^{-1}\Lambda \Upsilon^{-1},
 \qquad
 \Lambda:=W^{1/2}\Lb W^{1/2}.
\]
Along a spatial direction $h$, differentiating $\upsilon_i=s_{z,i}^{-1}$
gives $D\log\upsilon[h]=-A_zh$. Therefore, by the chain rule together with
$D\upsilon=\Upsilon D\log\upsilon$ and the definition of $\Lambda$ above,
we obtain
\begin{align*}
D\log w_z[h]
&=W^{-1}J_w(\upsilon)D\upsilon[h]\\
&=2(W-c_p\Lambda)^{-1}\Lambda D\log\upsilon[h]\\
&=-2W^{-1/2}(I-c_p\Lb)^{-1}\Lb W^{1/2}A_zh\\
&=-W^{-1/2}N B_zh,
\end{align*}
where the last step follows from the fact that $\Lb$ commutes with
$(I-c_p\Lb)^{-1}$, since this inverse is a function of $\Lb$.

It remains to prove the spectral bounds. By Eq.~\eqref{I:eq:phi-contraction}, we
have $0\preceq\Phi\Phi^{*}\preceq I$, which implies
$0\preceq\Lb=I-\Phi\Phi^{*}\preceq I$. Hence, the eigenvalues of
$\Rr=(I-c_p\Lb)^{-1}$ lie in $[1,(1-c_p)^{-1}]=[1,p/2]$, and those of
$N=2\Lb\Rr$ lie in $[0,p]$.
\end{proof}

\subsection{Analytic, metric, and geometric inputs}\label{I:sec:imports}

The following lemma is the only scalar concentration theorem used in the
polytope and spectrahedral proofs. We treat it as a non-elementary black
box~\cite{l01} and state it in the exact form that we need below.
\begin{lemma}[Gaussian Lipschitz concentration; black box]
\label{I:lem:gaussian-concentration}\label{II:lem:gaussian-concentration}
Let $\xi\sim N(0,I_m)$ and let $F:\R^m\to\R$ be
$L_F$-Lipschitz. For every $t\ge0$,
\begin{align*}
\Pr[F(\xi)-\E[F(\xi)]\ge t]
&\le \exp(-\frac{t^2}{2L_F^2}),\\
\Pr[|F(\xi)-\E[F(\xi)]|\ge t]
&\le 2\exp(-\frac{t^2}{2L_F^2}).
\end{align*}
The assertion for $L_F=0$ has the evident interpretation. In particular, if $Z$ is centered Gaussian with
$\operatorname{Var}(Z)\le\sigma^2$, then
$\Pr[|Z|\ge t]\le2e^{-t^2/(2\sigma^2)}$.
\end{lemma}

\begin{definition}[Strong self-concordance;~\cite{llv20}]
A matrix function $M:\intK\to\Sd_{\succ0}$ is \emph{strongly
self-concordant} (SSC) if for all $x\in\intK$ and $h\in\R^{d}$,
\[
\fnorm{M(x)^{-1/2}\,DM(x)[h]\,M(x)^{-1/2}}\le2\norm{h}_{M(x)} .
\]
\end{definition}

\begin{lemma}[SSC of the scaled LS metric]
\label{I:lem:ssc-ls}
The metric $2(1+p^{2})H$ is strongly self-concordant. Consequently,
$\kappa H$ is strongly self-concordant for every
$\kappa\ge2(1+p^{2})$.
\end{lemma}

\begin{proof}
Fix $z\in\intK$ and a direction $h$. For simplicity, we write
$A=A_z$, $W=W_z$, and $H=H(z)$, and we define
\[
\zeta:=Ah,\qquad
\dot\ell:=D\log w_z[h],\qquad
\alpha:=-2\zeta+c_p\dot\ell .
\]
Since $DA_z[h]=-\Diag(\zeta)A_z$ and
$D(W^{c_p})[h]=c_pW^{c_p}\Diag(\dot\ell)$, differentiating
$H=A^{\top}W^{c_p}A$ gives
\[
H^{-1/2}DH[h]H^{-1/2}=V_z^{\top}\Diag(\alpha)V_z .
\]
The columns of $V_z$ are orthonormal and its $i$-th row has squared
norm $w_{z,i}$, so $\opnorm{V_z^{\top}}=1$. By expanding the Frobenius
norm row by row and using $\alpha=-2\zeta+c_p\dot\ell$, $c_p\le1$, and
$(u+v)^2\le2u^2+2v^2$, we obtain
\begin{equation*}
\fnorm{H^{-1/2}DH[h]H^{-1/2}}^{2}
\le \fnorm{\Diag(\alpha)V_z}^{2}
=\sum_iw_{z,i}\alpha_i^2
\le8(\norm{\zeta}_{W}^{2}+\norm{\dot\ell}_{W}^{2}).
\end{equation*}
Next, we bound the two terms on the right-hand side. For the first term,
$\zeta=Ah$ and $W\preceq W^{c_p}$ imply
\[
\norm{\zeta}_{W}^{2}
=h^{\top}A^{\top}WAh\le\norm{h}_{H}^{2}.
\]
For the second term, Eq.~\eqref{I:eq:lewis-derivative},
$\opnorm{N_z}\le p$, and $W\preceq W^{c_p}$ give
\[
\norm{\dot\ell}_{W}^{2}=\norm{N_zB_zh}^{2}
\le p^{2}\norm{B_zh}^{2}
\le p^{2}\norm{h}_{H}^{2}.
\]

Thus, the square of the left side in the
SSC inequality for $H$ is at most $8(1+p^2)\norm{h}_{H}^{2}$. Since scaling
$H$ by a constant leaves that left side unchanged and
$4\norm{h}_{2(1+p^2)H}^{2}=8(1+p^2)\norm{h}_{H}^{2}$, the first claim
follows. Finally, the second claim holds because any larger scaling only
increases the right side of the SSC inequality.
\end{proof}

\begin{lemma}[Self-concordant distortion]
\label{I:lem:distortion}\label{II:lem:distortion}
Let $K$ be a convex body, and let $M$ be a $C^1$ matrix function on $\intK$
satisfying
\[
\opnorm{M(x)^{-1/2}DM(x)[h]M(x)^{-1/2}}\le 2\norm{h}_{M(x)}.
\]
Then for
$x,z\in\intK$ with $\Delta:=\norm{x-z}_{M(x)}<1$,
\[
\begin{aligned}(1-\Delta)^2M(x) &\preceq M(z), \\M(z) &\preceq (1-\Delta)^{-2}M(x).\end{aligned}
\]
\end{lemma}

\begin{proof}
Let $u=z-x$, $\gamma(t)=x+tu$, and
$\phi(t)=\norm{u}_{M(\gamma(t))}$. If $u=0$, the claim is trivial.
Otherwise, by differentiating $\phi(t)^2=u^{\top}M(\gamma(t))u$ and
applying the derivative hypothesis, we have
\[
|\phi'(t)|
=\frac{|u^{\top}DM(\gamma(t))[u]u|}{2\phi(t)}
\le\phi(t)^2,
\]
and hence $(\phi^{-1})'(t)\ge-1$. Since $\phi(0)=\Delta$, we obtain
\begin{equation}
\phi(t)\le\frac{\Delta}{1-t\Delta}
\qquad(0\le t\le1).
\label{I:eq:distortion-ode}
\end{equation}
Next, for a fixed $v$, let $\psi(t)=v^{\top}M(\gamma(t))v$. By the same
derivative hypothesis and Eq.~\eqref{I:eq:distortion-ode}, we have
\[
|\frac{d}{dt}\log\psi(t)|
\le2\phi(t)\le\frac{2\Delta}{1-t\Delta}.
\]

Integrating from $0$ to $1$ gives
$|\log(\psi(1)/\psi(0))|\le-2\log(1-\Delta)$, and exponentiating this bound
proves both Loewner inequalities, since $v$ was arbitrary.
\end{proof}

\begin{lemma}[{Cross-ratio isoperimetry; \cite[Theorem 2.5]{lv07}}]
\label{I:lem:isoperimetry-import}\label{II:lem:isoperimetry}
Let $\pi$ be a log-concave probability measure on a convex body $K$ and
let $K=S_{1}\,\dot\cup\,S_{2}\,\dot\cup\,S_{3}$ be a measurable
partition such that $\sigma(x,z)\ge\sigma_{0}$ for all $x\in S_{1}$,
$z\in S_{2}$, where $\sigma$ is the chord cross ratio of
Definition~\ref{I:def:crossratio}. Then
\[
\pi(S_{3})\ \ge\ c_{\mathrm{iso}}\,\min\{\sigma_{0},1\}\,
\pi(S_{1})\,\pi(S_{2})
\]
for a universal $c_{\mathrm{iso}}>0$.
\end{lemma}

The localization/isoperimetry theorem above is a non-elementary input, and we
apply it only after checking that $K$ is convex and that $e^{-f}\ind K$ is
log-concave, which holds because $f$ is convex. Its role in our analysis is
exactly to turn chord cross-ratio separation between two sets into stationary
mass in the intervening set.

From now on, we replace the universal constant $c_{\mathrm{iso}}$ of the
isoperimetry lemma by $\min\{c_{\mathrm{iso}},1\}$, which ensures
$0<c_{\mathrm{iso}}\le1$.

\begin{lemma}[{Cheeger's inequality for lazy reversible chains; \cite[Theorem 2.1]{ls88}}]
\label{I:lem:cheeger-import}\label{II:lem:cheeger}
A $\tfrac12$-lazy reversible chain with conductance $\Phi$ has spectral
gap $\gap(P)\ge c_{\mathrm{Ch}}\Phi^{2}$ for a universal
$c_{\mathrm{Ch}}>0$.
\end{lemma}

Cheeger's inequality above is the other deep input to our mixing analysis.
Exact Metropolis detailed balance gives reversibility of the chain, and
the explicit half-lazy step puts its spectrum in $[0,1]$, so the lemma
applies. We use this lemma to convert the conductance bound
proved later into an $L^{2}(\pi)$ spectral gap, while we derive the ensuing
warm-start contraction and the $L^{2}$-to-total-variation step directly in
Section~\ref{I:sec:mixing}. To match conventions precisely, we compare the
quantity used by Lawler and Sokal,
\[
k:=\inf_A\frac{\mathcal Q(A,A^c)}{\pi(A)\pi(A^c)},
\qquad
\mathcal Q(A,A^{c}):=\int_{A}P_{x}(A^{c})\,\pi(dx)
\]
(the stationary flow), with the conductance
$\Phi:=\inf_{\pi(A)\le1/2}\mathcal Q(A,A^c)/\pi(A)$ used in
Lemma~\ref{I:lem:conductance}. Since the chain is reversible, we can
replace $A$ by its smaller side, after which $\Phi\le k\le2\Phi$ holds
because $1\le1/\pi(A^c)\le2$. Hence, the bound of Lawler and Sokal in terms of $k$ gives the
displayed bound in terms of $\Phi$ up to a universal constant.

\renewcommand{\Hld}{H}
\subsection{Setting and the walk}

In the rest of this section and in Section~\ref{II:sec:spectral}, we let $K$ be
a spectrahedron with $n\times n$ matrix blocks and use $H$ to denote its
log-det metric, Eq.~\eqref{II:eq:metric-def}. In the proof for spectrahedra, we use only the general lemmas of Section~\ref{I:sec:imports},
Definition~\ref{I:def:crossratio}, and Lemma~\ref{I:lem:pinsker} from the proof for polytopes. Moreover, the Stein, cutoff, and McShane lemmas of
Section~\ref{II:sec:toolbox} are general, and we use them in both proofs.

Let $A_{0},A_{1},\dots,A_{d}\in\Sn$, and define
\[
K:=\{x\in\R^{d}:S(x)\succeq0\},\qquad
S(x):=A_{0}+\sum_{i=1}^{d}x_{i}A_{i}.
\]
We assume that $K$ is bounded and full-dimensional and, after a translation,
that $K\subseteq B_{2}(0,R)$. We also assume Slater strict feasibility, which
requires that $S(x_\circ)\succ0$ for some $x_\circ\in K$. This assumption is
equivalent to first removing any common kernel and working with the
positive-definite block on the minimal face. Let $\intK$ denote the
interior, on which $S(x)\succ0$ holds and on which the map
$x\mapsto S(x)^{-1}$ and everything built from it are smooth. Our target
distribution is $\pi\propto e^{-f}\ind{K}$, where $f$ is convex and
$L$-Lipschitz on $K$.

We define the \emph{log-det metric} as the Hessian of the barrier
$\phi(x)=-\log\det S(x)$:
\begin{equation}
\Hld(x)_{ij}
=\partial_{i}\partial_{j}\phi(x)
=\tr[S(x)^{-1}A_{i}\,S(x)^{-1}A_{j}]\ \succ0 .
\label{II:eq:metric-def}
\end{equation}
Indeed, for $v\in\R^{d}$ we have
$v^{\top}\Hld(x)v=\fnorm{S(x)^{-1/2}(\sum_i v_iA_i)S(x)^{-1/2}}^{2}$. If this
vanished for a nonzero $v$, then $S(x+tv)=S(x)$ for every $t$, which
contradicts the boundedness of $K$. Hence $\Hld(x)\succ0$. For a scale
$\kappa\ge1$, which we fix in Eq.~\eqref{II:eq:rho-choice} below with
$\kappa=\Ot(\psistar)$ in the notation of Definition~\ref{II:def:leverage}, we
set
\begin{equation*}
G(x):=\kappa\,\Hld(x)+L^{2}I .
\end{equation*}

\begin{definition}[Log-det walk with the exact metric]\label{II:def:walk}
Given the current point $x\in\intK$, propose
$Y\sim q_{x}:=N(x,\tfrac{r^{2}}{d}G(x)^{-1})$ with $r\in(0,1]$ a
universal constant, and accept with the exact Metropolis--Hastings
probability
$1\wedge R_{x}(Y)$,
$R_{x}(y):=e^{-f(y)}q_{y}(x)/e^{-f(x)}q_{x}(y)$. Proposals outside
$K$ are rejected, and the chain is made $\tfrac12$-lazy.
\end{definition}

The chain is reversible with respect to $\pi$. We say that an
initial law $\mu_{0}$ is \emph{$w$-warm} if $d\mu_{0}/d\pi\le w$.

\subsection{The matrix leverage operator}

\begin{definition}[Leverage operator]\label{II:def:leverage}
At $x\in\intK$, let $\Ab_{1}(x),\ldots,\Ab_{d}(x)\in\Sn$ be the
slack-congruenced constraint directions
$S(x)^{-1/2}A_{i}S(x)^{-1/2}$ after orthonormalization in the
Frobenius inner product (equivalently, in the metric
$\Hld(x)$). Define
\[
\Psib(x):=\sum_{a=1}^{d}\Ab_{a}(x)^{2},\qquad
\psi(x):=\opnorm{\Psib(x)},\qquad
\psistar:=\max(1,\ \sup_{x\in\intK}\psi(x)).
\]
\end{definition}

\begin{lemma}[Basic properties]\label{II:lem:leverage-basic}
$\Psib(x)$ is well-defined (independent of the choice of orthonormal
frame), $\tr[\Psib(x)]=d$, and $d/n\le\psi(x)\le d$. If all $A_{i}$ are
diagonal (the polytope case), $\Psib$ is the diagonal matrix of
leverage scores of the slack-rescaled constraint matrix, so
$\psi\le1$.
\end{lemma}

\begin{proof}
To see that the operator does not depend on the choice of frame, note that if $\{\Ab_{a}\}$ and
$\{\Ab_{a}'\}$ are two Frobenius-orthonormal bases of the same subspace of
$\Sn$, then $\Ab_{a}'=\sum_{b}O_{ab}\Ab_{b}$ for an orthogonal $O$, and
therefore $\sum_{a}\Ab_{a}'\otimes\Ab_{a}'=\sum_{a}\Ab_{a}\otimes\Ab_{a}$. By
applying this bilinear identity to the product map, we obtain
$\sum_{a}(\Ab_{a}')^{2}=\sum_{a}\Ab_{a}^{2}$. By orthonormality, we have
$\tr[\Psib]=\sum_{a}\fnorm{\Ab_{a}}^{2}=d$, which implies
$\psi\ge\tr[\Psib]/n=d/n$. For the upper bound, note that $\psi\le\sum_{a}\opnorm{\Ab_{a}^{2}}
=\sum_{a}\opnorm{\Ab_{a}}^{2}\le\sum_{a}\fnorm{\Ab_{a}}^{2}=d$. Finally, in the
diagonal case, where $\bar a_{k}\in\R^{d}$ denotes the vector of $k$-th
diagonal entries of the $\Ab_{a}$, orthonormality makes
$(\Psib)_{kk}=\norm{\bar a_{k}}^{2}$ the $k$-th leverage score of the
orthonormalized constraint matrix, and leverage scores are at most $1$.
\end{proof}

\subsection{Toolbox}\label{II:sec:toolbox}

We now collect the elementary and imported facts that we use
repeatedly. Recall the bounds $|\tr[XY]|\le\fnorm{X}\fnorm{Y}$ and
$\snorm{XY}\le\fnorm{X}\fnorm{Y}$, where $\snorm{\cdot}$ is the trace norm.

\begin{lemma}[PSD trace facts]\label{II:lem:psd-facts}
Let $X,Y,M,T\in\Sn$ with $X,Y,M\succeq0$ and $T\succeq0$.
\begin{enumerate}[label=(\roman*),leftmargin=2.4em,itemsep=1pt]
\item $\tr[XY]\ge0$, and consequently $0\preceq X\preceq Y$ implies
$\tr[X^{2}]\le\tr[XY]\le\tr[Y^{2}]$.
\item $\tr[TM]\le\opnorm{T}\tr[M]$ and, for any $B\in\Sn$,
$BTB\preceq\opnorm{T}\,B^{2}$.
\item $\fnorm{T^{1/2}MT^{1/2}}\le\opnorm{T}\fnorm{M}$.
\end{enumerate}
\end{lemma}

\begin{proof}
For (i), we have $\tr[XY]=\fnorm{X^{1/2}Y^{1/2}}^{2}\ge0$. Under the
hypothesis of the second claim, this fact yields $\tr[X(Y-X)]\ge0$ and also
$\tr[(Y-X)Y]\ge0$, which together give the chain of inequalities. Part (ii) follows
from $\tr[TM]=\tr[M^{1/2}TM^{1/2}]\le\opnorm{T}\tr[M]$ and
$u^{\top}BTBu\le\opnorm{T}\norm{Bu}^{2}$ for every $u$, and part (iii) follows from
$\fnorm{T^{1/2}MT^{1/2}}\le\opnorm{T^{1/2}}^{2}\fnorm{M}$.
\end{proof}

\begin{lemma}[Twice-integrated Gaussian Stein identity]\label{II:lem:stein}\label{I:lem:stein}
Let $h\sim N(0,Q)$ on $\R^{d}$ with $0\preceq Q\preceq I$, and let
$F:\R^{d}\to\R$ be $W^{2,\infty}$ with compact support. Then for all
indices $a,b$,
\begin{equation}
\E[(h_{a}h_{b}-Q_{ab})F(h)]
=\sum_{c,e}Q_{ac}Q_{be}\,\E[\partial_{c}\partial_{e}F(h)].
\label{II:eq:stein-identity}
\end{equation}
Equivalently, contracting with $v_{a}v_{b}$, for every $v\in\R^{d}$,
\[
\E[((v^{\top}h)^{2}-v^{\top}Qv)F(h)]
=\E[(Qv)^{\top}D^{2}F(h)\,(Qv)],
\]
where $D^{2}F$ is the (a.e.-defined, bounded) weak Hessian.
\end{lemma}

\begin{proof}
Let $\xi$ be a standard Gaussian vector with $h=Q^{1/2}\xi$. For any smooth
$G$ whose value and first two derivatives are bounded, by integrating by parts
twice against the Gaussian density, we obtain
$\E[(\xi_{\alpha}\xi_{\beta}-\delta_{\alpha\beta})G(\xi)]
=\E[\partial_{\alpha}\partial_{\beta}G(\xi)]$. We apply this identity to
$G=F\circ Q^{1/2}$. Even when $Q$ is singular, this $G$ has bounded
derivatives, which is all that the integration-by-parts argument requires.
Moreover, we have
\[
\begin{aligned}
h_{a}h_{b}-Q_{ab}
&=\sum_{\alpha,\beta}
(Q^{1/2})_{a\alpha}(Q^{1/2})_{b\beta}
(\xi_{\alpha}\xi_{\beta}-\delta_{\alpha\beta}),\\
\partial_{\alpha}\partial_{\beta}(F\circ Q^{1/2})
&=\sum_{c,e}(Q^{1/2})_{c\alpha}(Q^{1/2})_{e\beta}
(\partial_{c}\partial_{e}F)\circ Q^{1/2},
\end{aligned}
\]
and contracting these two identities gives
Eq.~\eqref{II:eq:stein-identity}. Finally, for $F\in W^{2,\infty}$ with
compact support, we mollify and pass to the limit by dominated convergence.
\end{proof}

\begin{lemma}[Dimension-free cutoff]\label{II:lem:cutoff}\label{I:lem:cutoff}
Let $\mathcal{K}\subset\R^{d}$ be closed and convex. There is
$\chi:\R^{d}\to[0,1]$ with $\chi\equiv1$ on $\mathcal{K}$, $\chi$
supported in the closed Euclidean $1$-neighborhood
$\mathcal{K}^{+}$, and
$\norm{\nabla\chi}\le12$, $\opnorm{D^{2}\chi}\le C_{0}$ a.e., for a
universal $C_{0}$.
\end{lemma}

\begin{proof}
The function $\phi(h)=\dist(h,\mathcal{K})^{2}$ is convex
and $C^{1,1}$ with the $2$-Lipschitz gradient
$\nabla\phi=2(h-\Pi_{\mathcal{K}}h)$. Hence $0\preceq D^{2}\phi\preceq2I$ a.e.
We take $\chi=\theta\circ\phi$ for a fixed smooth decreasing
$\theta:[0,\infty)\to[0,1]$ with $\theta\equiv1$ on $[0,\tfrac19]$,
$\theta\equiv0$ on $[\tfrac49,\infty)$, $|\theta'|\le9$, and
$|\theta''|\le C_{\theta}$ for a fixed universal constant $C_{\theta}$. On the
support of $\theta'\circ\phi$, we have $\norm{\nabla\phi}\le\tfrac43$. Thus,
$\norm{\nabla\chi}\le9\cdot\tfrac43=12$ and
$\opnorm{D^{2}\chi}\le|\theta''|\,\norm{\nabla\phi}^{2}
+|\theta'|\,\opnorm{D^{2}\phi}\le C_{0}$, and these bounds are
dimension-free.
\end{proof}

\begin{lemma}[McShane extension]\label{II:lem:mcshane}\label{I:lem:mcshane}
A function $Z$ that is $\Lambda$-Lipschitz on a nonempty set
$\mathcal{K}\subseteq\R^{d}$ extends to a $\Lambda$-Lipschitz
function $\widetilde Z$ on $\R^{d}$, e.g.\
$\widetilde Z(h)=\inf_{g\in\mathcal{K}}\{Z(g)+\Lambda\norm{h-g}\}$.
If $0\in\mathcal{K}$ and $Z(0)=0$, then $|\widetilde Z(h)|\le\Lambda\norm{h}$.
\end{lemma}

\begin{proof}
For $h\in\mathcal K$, the choice $g=h$ gives $\widetilde Z(h)\le Z(h)$. On
the other hand, by Lipschitzness we have $Z(g)+\Lambda\norm{h-g}\ge Z(h)$ for
every $g\in\mathcal K$. Hence $\widetilde Z=Z$ on $\mathcal K$. The infimum is
also finite, since for a fixed $g_{0}\in\mathcal K$, the same
Lipschitz inequality and the triangle inequality give
$Z(g)+\Lambda\norm{h-g}\ge
Z(g_{0})-\Lambda\norm{h-g_{0}}$. It remains to prove the Lipschitz bound. For
any $h,h'$ and every $g\in\mathcal K$, we have
\[
Z(g)+\Lambda\norm{h-g}
\le Z(g)+\Lambda\norm{h'-g}+\Lambda\norm{h-h'}.
\]
By taking the infimum over $g$, we obtain
$\widetilde Z(h)\le\widetilde Z(h')+\Lambda\norm{h-h'}$, and exchanging
$h,h'$ proves the claimed Lipschitz bound. The last claim of the
lemma follows from $\widetilde Z(0)=Z(0)=0$ and this bound with $h'=0$.
\end{proof}

\begin{lemma}[{Gaussian matrix series; \cite[Theorem 4.1.1]{t15}}]
\label{II:lem:matrix-gaussian}
Let $B_{1},\dots,B_{N}\in\Sn$ be fixed and $\xi$ standard Gaussian in
$\R^{N}$. Then, with
$\sigma^{2}:=\opnorm{\sum_{i}B_{i}^{2}}$,
\[
\Pr[\lVert\textstyle\sum_{i}\xi_{i}B_{i}\rVert_{\mathrm{op}}
\ge t]\ \le\ 2n\,e^{-t^{2}/2\sigma^{2}}
\qquad(t\ge0).
\]
\end{lemma}

The preceding lemma on Gaussian matrix series is the only noncommutative
concentration theorem that we use in the log-det analysis. In every application, we
first express the random matrix as a self-adjoint series $\sum_i\xi_iB_i$ and
then bound the deterministic variance operator $\sum_iB_i^2$ in Loewner order.
All remaining noncommutative steps are finite-dimensional trace, congruence,
and resolvent calculations proved in the text.

We also use Gaussian Lipschitz concentration, the self-concordant
distortion bound, cross-ratio isoperimetry, Cheeger's inequality for lazy
reversible chains, and Pinsker's inequality in the proof for spectrahedra exactly as stated in
Lemmas~\ref{I:lem:gaussian-concentration}, \ref{I:lem:distortion},
\ref{I:lem:isoperimetry-import}, \ref{I:lem:cheeger-import},
and~\ref{I:lem:pinsker}, which hold for any convex body. The cross ratio of
Definition~\ref{II:def:crossratio} is the one of
Definition~\ref{I:def:crossratio}, and, as fixed earlier in this section,
$0<c_{\mathrm{iso}}\le1$ and the conductance $\Phi$ follows the convention
stated after Lemma~\ref{I:lem:cheeger-import}.

\renewcommand{\nubar}{\bar{\nu}_{\mathrm{LS}}}

\section{Proof of Theorem~\ref{thm:intro-polytope}}\label{I:sec:polytope}

\subsection{The metric, the walk, and the main theorem}

In this section, we work with the unscaled \emph{Lee--Sidford metric},
defined as
\begin{equation*}
H(x)=A_{x}^{\top}W_{x}^{1-2/p}A_{x}\succ0 .
\end{equation*}
Let $\kappa$ be a sufficiently large fixed polynomial in $p$ and
$\log(nd)$ (its precise requirements are
Eq.~\eqref{I:eq:bootstrap-smallness}, Eq.~\eqref{I:eq:rho-choice}, and
$\kappa\ge2(1+p^{2})$ from Lemma~\ref{I:lem:ssc-ls}), and set
\begin{equation*}
G(x)=\kappa H(x)+L^{2}I .
\end{equation*}

\begin{definition}[LS walk with the exact metric]\label{I:def:walk}
At $x\in\intK$, propose
\begin{equation*}
Y\sim q_{x}:=N(x,\;\frac{r^{2}}{d}\,G(x)^{-1}),
\end{equation*}
where $r\in(0,1]$ is a universal constant, and accept with the exact
Metropolis--Hastings probability
\begin{equation*}
\alpha_{x}(y)=1\wedge R_{x}(y),\qquad
R_{x}(y):=\frac{e^{-f(y)}\,q_{y}(x)}{e^{-f(x)}\,q_{x}(y)} .
\end{equation*}
Proposals outside $K$ are rejected, and the chain is made
$\tfrac12$-lazy.
\end{definition}

We note that the walk is reversible with respect to $\pi\propto e^{-f}\ind{K}$.
Indeed, on $\intK\times\intK$, the density of an accepted move satisfies
\[
e^{-f(x)}q_x(y)(1\wedge R_x(y))
=\min\{e^{-f(x)}q_x(y),\,e^{-f(y)}q_y(x)\},
\]
which is symmetric in $x,y$. Moreover, rejection and laziness only add mass on
the diagonal. We call an initial law $\mu_{0}$ \emph{$w$-warm} if
$d\mu_{0}/d\pi\le w$.

\begin{theorem}[Polytopes; formal version of Theorem~\ref{thm:intro-polytope}]\label{I:thm:main-formal}
Let $K\subseteq B_{2}(0,R)$ be a bounded full-dimensional polytope in $\R^{d}$
given by $n$ linear inequalities, let $f$ be convex and $L$-Lipschitz on $K$,
and let $\pi\propto e^{-f}\ind{K}$. Let $H_{\mathrm{LS}}$ denote the
Lee--Sidford metric, fix a sufficiently small universal constant $r\in(0,1]$,
and choose $\kappa=\mathrm{poly}(p,\log(nd))$ so that Eq.~\eqref{I:eq:rho-choice}, Eq.~\eqref{I:eq:bootstrap-smallness}, and
$\kappa\ge2(1+p^{2})$ (Lemma~\ref{I:lem:ssc-ls}) hold. Consider the $\tfrac12$-lazy
Metropolis walk with exact Lewis weights, exact acceptance ratio, and proposal
covariance
\[
  \frac{r^{2}}{d}(\kappa H_{\mathrm{LS}}(x)+L^{2}I)^{-1}.
\]
Then, for $w\ge1$ and $0<\delta<1$, from a $w$-warm start the walk reaches
total-variation distance $\delta$ from $\pi$ in
\[
  O((\kappa d^{2}+dL^{2}R^{2})\log(w/\delta))
  =\Ot((d^{2}+dL^{2}R^{2})\log(w/\delta))
\]
steps, where $\kappa=\mathrm{poly}(p,\log(nd))$ is polylogarithmic in $n$ and
$d$.
\end{theorem}

\begin{proof}
Let $f_{0}=d\mu_{0}/d\pi$. Since $\mu_{0}$ is $w$-warm, we have
$0\le f_{0}\le w$ and $\int f_{0}\,d\pi=1$. Expanding the square and
using $f_0^2\le wf_0$ gives
\begin{equation*}
\norm{f_{0}-1}_{L^{2}(\pi)}^{2}=\int f_{0}^{2}\,d\pi-1\le w-1 .
\end{equation*}

We write the lazy kernel as $P=(I+P_{0})/2$, where $P_{0}$ is the
underlying non-lazy Metropolis kernel. Since $P_{0}$ is a self-adjoint
contraction on $L^2(\pi)$, its spectrum lies in $[-1,1]$, and hence the
spectrum of $P$ lies in $[0,1]$. Therefore, by
Lemma~\ref{I:lem:cheeger-import}, we have $\gap(P)\ge c_{\mathrm{Ch}}\Phi^{2}$,
and the spectral theorem gives
$\norm{P^tg}_2\le(1-\gap(P))^t\norm g_2$ for every mean-zero $g$.
By Cauchy--Schwarz, the preceding warm-start bound, and $1-u\le e^{-u}$, we
obtain
\begin{equation*}
\dtv(\mu_{0}P^{t},\pi)
\le\frac12\norm{P^{t}(f_{0}-1)}_{L^{2}(\pi)}
\le\frac12(1-\gap(P))^{t}\sqrt{w-1}
\le\frac12\,e^{-c_{\mathrm{Ch}}t\Phi^{2}}\sqrt{w-1} .
\end{equation*}

Solving for the total-variation error $\delta$ and plugging in
Eq.~\eqref{I:eq:conductance}, we obtain
\[
t_{\mathrm{mix}}(\delta)\ \le\
C\,\frac{d\,(\nubar+L^{2}R^{2})}{r^{2}}\,\log(w/\delta).
\]
Since $n^{2/p}\le e$, we have $\nubar\le e\kappa d$ by
Eq.~\eqref{I:eq:nubar}, and since $r$ is a universal constant, the bound above
on the mixing time is
$O((\kappa d^{2}+dL^{2}R^{2})\log(w/\delta))$.
Finally, by Eq.~\eqref{I:eq:rho-choice} and Eq.~\eqref{I:eq:M0-size},
together with $\kappa\ge2(1+p^{2})$ and Eq.~\eqref{I:eq:bootstrap-smallness},
we can take $\kappa$ to be a fixed polynomial in $p$ and $\log(nd)$. The same derivation of mixing from conductance appears
in~\cite[Lemma B.1 and Theorem 3.1]{kv24}.
\end{proof}

\subsection{The complete proposal ratio and the exact cancellation}
\label{I:sec:cancellation}

The goal of this section is to show the cancellation inside the complete
Gaussian proposal ratio, using only exact algebra and no estimates. To this
end, we normalize the coordinates at the current point and diagonalize the
endpoint metric exactly.

We fix a deterministic current point $x\in\intK$, let
$F_{x}=H_{\mathrm{orig}}(x)^{-1/2}$, and apply the linear change of
variables $z_{\mathrm{orig}}=F_{x}z$ to the entire inequality
representation, which transforms $x$, $y$, $h$, $\mA$, and $A_{x}$
accordingly. While $H$ and $G$ transform as metrics, the Lewis weights are
invariant under this substitution, because leverage scores, and
hence the fixed point Eq.~\eqref{I:eq:lewis-fixedpoint}, are unchanged by
$A\mapsto AT$ for invertible $T$. Throughout
Sections~\ref{I:sec:cancellation}, \ref{I:sec:stability}, \ref{I:sec:hessian} and \ref{I:sec:stein}, we suppress
tildes on the transformed quantities. In particular, every slack identity
below uses the transformed $\mA$ in the transformed coordinates. We write
\[
W:=W_{x},\qquad w:=w_{x}.
\]
With this choice of the change of variables, the baseline metric is
normalized:
\begin{equation*}
H(x)=I,\qquad
U:=W^{\beta}A_{x},\qquad
U^{\top}U=I .
\end{equation*}
Let $u_{i}^{\top}$ denote row $i$ of $U$. By the fixed point
Eq.~\eqref{I:eq:lewis-fixedpoint}, we have
\begin{equation*}
\norm{u_{i}}^{2}=w_{i},\qquad
P:=UU^{\top},\qquad
P_{ii}=w_{i},\qquad
\sum_{i}w_{i}=d .
\end{equation*}
After the same affine change, we factor the (now constant) regularizer as
\begin{equation*}
G(x)=\kappa(I+\creg),\qquad
\creg=\frac{L^{2}}{\kappa}F_{x}^{\top}F_{x}\succeq0,\qquad
Q:=(I+\creg)^{-1},\qquad
0\preceq Q\preceq I .
\end{equation*}
Since we use only $\creg\succeq0$ below, all estimates are uniform in
$L$ and in the orientation of the regularizer. With this factorization,
we can represent a proposal from the current point as
\begin{equation}
y=x+\eta h,\qquad
h\sim N(0,Q),\qquad
\eta=\frac{r}{\sqrt{\kappa d}}=\frac{\rho}{\sqrt{d}},\qquad
\rho:=\frac{r}{\sqrt{\kappa}} .
\label{I:eq:step}
\end{equation}
We use $\kappa$ only to make $\rho$ smaller than a sufficiently small
inverse polylogarithm, at the cost of only a polylogarithmic factor in the
LS symmetry parameter.

\begin{lemma}[Endpoint identities]\label{I:lem:endpoint}
Let $h\in\R^{d}$ be such that the segment $[x,y]$, $y=x+\eta h$, lies
in $\intK$, and put
\begin{equation}
\zeta:=A_{x}h,\qquad
a_{i}(h):=\beta\log\frac{w_{y,i}}{w_{i}}-\log(1+\eta\zeta_{i}),\qquad
b_{i}:=e^{2a_{i}}-1 .
\label{I:eq:a-def}
\end{equation}
Then
\begin{equation}
W_{y}^{\beta}A_{y}=\Diag(e^{a})\,U,
\qquad
H(y)=U^{\top}\Diag(e^{2a})\,U=I+U^{\top}\Diag(b)\,U .
\label{I:eq:endpoint-diag}
\end{equation}
Moreover, with
\begin{equation*}
c_{i}:=(u_{i}^{\top}h)^{2},
\qquad
\wb_{i}:=\E[c_{i}]=u_{i}^{\top}Qu_{i},
\end{equation*}
one has the exact quadratic identity
\begin{equation*}
h^{\top}(H(y)-I)h=c^{\top}b .
\end{equation*}
\end{lemma}

\begin{proof}
By the relations $y=x+\eta h$, $\mA=S_xA_x$, and $\zeta=A_xh$, we
have
\[
s_{y}=s_{x}+\eta\mA h
=S_{x}(\mathbf{1}+\eta A_{x}h)=S_{x}(\mathbf{1}+\eta\zeta).
\]
Hence, $A_{y}=\Diag(\mathbf{1}+\eta\zeta)^{-1}A_{x}$. Multiplying by
$W_{y}^{\beta}=\Diag((w_{y,i}/w_{i})^{\beta})W^{\beta}$ and
comparing with Eq.~\eqref{I:eq:a-def}, we obtain the first identity in
Eq.~\eqref{I:eq:endpoint-diag}. The second identity then follows from
$H(y)=(W_{y}^{\beta}A_{y})^{\top}(W_{y}^{\beta}A_{y})$ and
$U^{\top}U=I$. Finally, for the quadratic identity, we have
$h^{\top}U^{\top}\Diag(b)Uh=\sum_{i}b_{i}(u_{i}^{\top}h)^{2}=c^{\top}b$.
\end{proof}

We remark that the proof above makes no Taylor expansion of the endpoint
metric. The sign in the endpoint rescaling is explicit, and
Eq.~\eqref{I:eq:endpoint-diag} is an identity.

\begin{lemma}[Proposal-ratio identity]\label{I:lem:ledger}
In the setting of Lemma~\ref{I:lem:endpoint}, let
\begin{equation}
E:=Q^{1/2}(H(y)-I)Q^{1/2} .
\label{I:eq:E-def}
\end{equation}
Then the complete Gaussian proposal log-ratio satisfies, exactly,
\begin{equation}
\begin{aligned}
\Lambda_{q}(x,y)
&:=\frac{1}{2}\log\frac{\det G(y)}{\det G(x)}
-\frac{d}{2r^{2}}\,(y-x)^{\top}(G(y)-G(x))(y-x) \\
&=-\frac{1}{2}(c-\wb)^{\top}b
+\frac{1}{2}\{\log\det(I+E)-\tr[E]\}.
\end{aligned}
\label{I:eq:ledger}
\end{equation}
Equivalently,
$\log\tfrac{q_{y}(x)}{q_{x}(y)}=\Lambda_{q}(x,y)$.
\end{lemma}

\begin{proof}
In the normalized coordinates, we have $G(z)=\kappa H(z)+\kappa\creg$.
Thus, the regularizer part is constant in $z$, and
$G(y)-G(x)=\kappa(H(y)-I)$. Below, we verify the determinant, trace, and
quadratic terms separately.

\emph{(i) Determinant part.} We have
\begin{equation}
\log\frac{\det G(y)}{\det G(x)}
=\log\frac{\det(H(y)+\creg)}{\det(I+\creg)}
=\log\det(I+Q^{1/2}(H(y)-I)Q^{1/2})
=\log\det(I+E) ,
\label{I:eq:det-part}
\end{equation}
where the first step follows from cancelling the common factor $\kappa$,
the second step follows from $Q=(I+\creg)^{-1}$ and congruence with
$(I+\creg)^{-1/2}$ inside the determinant, and the last step follows from
Eq.~\eqref{I:eq:E-def}.

\emph{(ii) Trace of $E$.} We have
\begin{equation}
\tr[E]=\tr[Q\,U^{\top}\Diag(b)\,U]
=\sum_{i}b_{i}\,u_{i}^{\top}Qu_{i}=\wb^{\top}b ,
\label{I:eq:trE}
\end{equation}
where the first step follows from Eqs.~\eqref{I:eq:E-def}
and~\eqref{I:eq:endpoint-diag} and the cyclicity of the trace, the second
step follows from expanding over the rows of $U$, and the last step follows
from the definition of $\wb$.

\emph{(iii) Quadratic part.} We have
\begin{equation}
\frac{d}{r^{2}}\,(y-x)^{\top}(G(y)-G(x))(y-x)
=\frac{d\eta^{2}\kappa}{r^{2}}\,h^{\top}(H(y)-I)h
=c^{\top}b ,
\label{I:eq:quad-part}
\end{equation}
where the first step follows from $y-x=\eta h$ and the identity
$G(y)-G(x)=\kappa(H(y)-I)$, and the second step follows from
$d\eta^{2}\kappa/r^{2}=1$ (by Eq.~\eqref{I:eq:step}) and the quadratic
identity in Lemma~\ref{I:lem:endpoint}.

Putting Eqs.~\eqref{I:eq:det-part}, \eqref{I:eq:trE}, and
\eqref{I:eq:quad-part} together, and adding and subtracting
$\tfrac12\tr[E]$, we obtain
\[
\Lambda_{q}
=\frac12\log\det(I+E)-\frac12\,c^{\top}b
=\frac12\{\log\det(I+E)-\tr[E]\}
+\frac12(\wb^{\top}b-c^{\top}b),
\]
which is Eq.~\eqref{I:eq:ledger}.

For the last claim, note that the Gaussian densities satisfy
\[
q_{z}(\cdot)\propto\det G(z)^{1/2}
\exp(-\tfrac{d}{2r^{2}}\norm{\cdot-z}_{G(z)}^{2}).
\]
Taking the logarithm of the density ratio, we obtain
\[
\begin{aligned}
\log\frac{q_{y}(x)}{q_{x}(y)}
&=\frac12\log\frac{\det G(y)}{\det G(x)}
-\frac{d}{2r^{2}}
(\norm{x-y}_{G(y)}^{2}-\norm{y-x}_{G(x)}^{2})
=\Lambda_{q}(x,y),
\end{aligned}
\]
where we used that the normalizing constants and any Jacobians of the affine
substitution cancel in the ratio.
\end{proof}

\begin{remark}[Why the complete ratio]\label{I:rem:why-complete}
Eq.~\eqref{I:eq:ledger} contains the key cancellation. The potentially
difficult uncentered term $\wb^{\top}b$ cancels because it is
\emph{simultaneously} the linear
part of the determinant and, with $b$ held fixed, the Gaussian mean of the
random quadratic term. In contrast, an analysis that certifies a lower tail
for the determinant ratio and concentration for the quadratic form
separately, for instance through (regularized) average self-concordance,
does not use this cancellation, and our second-order Lewis-weight inputs do
not bound $\wb^{\top}b$ by itself. After the cancellation, the first term of
Eq.~\eqref{I:eq:ledger} is centered to leading order, and we reduce it
further to the exactly centered $Z=(c-\wb)^{\top}a$ in
Lemma~\ref{I:lem:remainder-quad}. Moreover, the bracket is a finite
determinant remainder (Lemma~\ref{I:lem:remainder-det}). In
Sections~\ref{I:sec:stability}, \ref{I:sec:hessian} and \ref{I:sec:stein},
we bound both terms at constant scale.
\end{remark}
\subsection{Uniform local stability}\label{I:sec:stability}

While the identities of Section~\ref{I:sec:cancellation} are exact, they
involve the endpoint weights $w_{y}$. In this section, we construct a
convex high-probability good set of steps $h$ on which all endpoint
quantities (weights, slacks, metric) are uniformly comparable to their
base-point values, with polylogarithmic constants. We then prove the two
finite remainder bounds that reduce the proposal-ratio identity
Eq.~\eqref{I:eq:ledger} to the centered fluctuation $Z=(c-\wb)^{\top}a$.
Throughout, we work in the normalized coordinates of
Section~\ref{I:sec:cancellation}.

Let
\[
k_{i}:=\frac{u_{i}}{\sqrt{w_{i}}},\qquad \norm{k_{i}}=1 .
\]
For a fixed small failure probability $\eone$, we choose the auxiliary
tail level
\begin{equation*}
\delta_{0}:=\eone^{2}d^{-3},
\end{equation*}
define the deterministic radii
\begin{equation*}
R_{0}:=\sqrt{d}+\sqrt{2\log(2/\delta_{0})},
\qquad
\tau_{0}:=\sqrt{2\log(4n/\delta_{0})},
\end{equation*}
and the convex good set
\begin{equation*}
\Ktube:=\{g\in\R^{d}:\ \norm{g}\le R_{0},\ \
\max_{i}|k_{i}^{\top}g|\le\tau_{0}\}.
\end{equation*}

\begin{lemma}\label{I:lem:tube-probability}
$\Pr[h\in\Ktube]\ge1-\delta_{0}$ for $h\sim N(0,Q)$.
\end{lemma}

\begin{proof}
We write $h=Q^{1/2}\xi$ with $\xi$ a standard Gaussian vector. Since
$Q\preceq I$, each $k_{i}^{\top}h$ is Gaussian with variance
$k_{i}^{\top}Qk_{i}\le1$, and
$\E[\norm{h}]\le(\E[\norm{h}^{2}])^{1/2}=(\tr[Q])^{1/2}\le\sqrt{d}$. The
map $\xi\mapsto\norm{Q^{1/2}\xi}$ is $1$-Lipschitz, so by Gaussian
Lipschitz concentration (Lemma~\ref{I:lem:gaussian-concentration}), we
have $\Pr[\norm{h}>R_{0}]\le e^{-\log(2/\delta_{0})}=\delta_{0}/2$.
Moreover, a union bound with the centered-Gaussian tail in the same
lemma gives
$\Pr[\max_{i}|k_{i}^{\top}h|>\tau_{0}]
\le2n\,e^{-\tau_{0}^{2}/2}=\delta_{0}/2$.
\end{proof}

We prove all estimates on the enlarged set $\Ktubep$, the closed Euclidean
one-neighborhood of $\Ktube$, since the cutoff of
Section~\ref{I:sec:stein} needs this slack. Let
\begin{equation*}
R_{1}:=R_{0}+1,\qquad
\tau_{1}:=\tau_{0}+1,\qquad
M_{0}:=\frac{R_{1}}{\sqrt{d}} .
\end{equation*}
Then every $g\in\Ktubep$ obeys $\norm{g}\le R_{1}$ and
$\max_{i}|k_{i}^{\top}g|\le\tau_{1}$ (since $\norm{k_{i}}=1$). Moreover,
for $\delta_{0}=\eone^{2}d^{-3}$, these constants satisfy
\begin{equation}
M_{0}=O_{\eone}(1),\qquad
\tau_{1}^{2}=O(1+\log(nd/\eone)),\qquad
M_{0}\ge1 .
\label{I:eq:M0-size}
\end{equation}

Since the bootstrap requires a first-derivative bound with the
\emph{row-dependent} scaling $\sqrt{w_{z,i}}$, we record such a bound in
the next lemma, which sharpens the aggregate bound
$\opnorm{\Rr_{z}}\le p/2$ of Lemma~\ref{I:lem:lewis-derivative} row by
row. We state part~(ii) of the lemma for a general coefficient, because
Section~\ref{I:sec:hessian} needs it for $I+\tfrac12N_{z}$ as well as
for $\Rr_{z}=I+\beta N_{z}$.

\begin{lemma}[Rowwise resolvent bound]\label{I:lem:rowwise}
Let $z\in\intK$, recall
$\Bh_{z}=B_{z}H(z)^{-1/2}$, and set $L_{p}:=\tfrac{p}{2}(1+\tfrac{p}{2})$.
Then for every row index $i$:
\begin{enumerate}[label=(\roman*),leftmargin=2.4em]
\item $\norm{e_{i}^{\top}\Rr_{z}\Bh_{z}}\le L_{p}\sqrt{w_{z,i}}$;
\item for every $\gamma\in[0,1]$,
$\norm{e_{i}^{\top}(I+\gamma N_{z})\Bh_{z}}\le L_{p}^{+}\sqrt{w_{z,i}}$
with $L_{p}^{+}:=1+2L_{p}+p$.
\end{enumerate}
Consequently, in arbitrary coordinates, for all $v\in\R^{d}$,
\begin{equation}
|e_{i}^{\top}(I+\beta N_{z})B_{z}v|
\le L_{p}\sqrt{w_{z,i}}\;\norm{v}_{H(z)} .
\label{I:eq:rowwise-metric}
\end{equation}
\end{lemma}

\begin{proof}
(i) Since $\Rr_{z}$ is a function of $\Lb_{z}$, it commutes with
$\Phi_{z}\Phi_{z}^{*}=I-\Lb_{z}$, and from
$\Rr_{z}(I-c_{p}\Lb_{z})=I$ we obtain the resolvent identity
\[
(1-c_{p})\,\Rr_{z}=I-c_{p}\,\Phi_{z}\Phi_{z}^{*}\Rr_{z} .
\]
We multiply by $\Bh_{z}$ on the right, take the $i$-th row, and use
$(1-c_{p})^{-1}=p/2$, together with
$\fnorm{(\Phi_z)_i}=\sqrt{w_{z,i}}$, $\opnorm{\Phi_z}\le1$,
$\opnorm{\Rr_z}\le p/2$, and $\opnorm{\Bh_z}\le1$, to obtain
\[
\norm{e_{i}^{\top}\Rr_{z}\Bh_{z}}
\le\frac{p}{2}(\norm{(\Bh_{z})_{i}}
+c_{p}\,\norm{e_{i}^{\top}\Phi_{z}}\;
\opnorm{\Phi_{z}^{*}\Rr_{z}\Bh_{z}})
\le\frac{p}{2}(w_{z,i}^{1/2+1/p}
+c_{p}\sqrt{w_{z,i}}\cdot\frac{p}{2})
\le L_{p}\sqrt{w_{z,i}},
\]
where the last inequality also uses $c_p\le1$, $w_{z,i}^{1/p}\le1$,
and the definition of $L_p$.

(ii) Since $\Lb_{z}$ commutes with $\Rr_{z}$, we have
$N_{z}=2\Lb_{z}\Rr_{z}=2(I-\Phi_{z}\Phi_{z}^{*})\Rr_{z}
=2\Rr_{z}-2\Phi_{z}\Phi_{z}^{*}\Rr_{z}$. By part~(i) and the same
operator-norm bounds, we obtain
\[
\norm{e_{i}^{\top}N_{z}\Bh_{z}}
\le2\norm{e_{i}^{\top}\Rr_{z}\Bh_{z}}
+2\norm{e_{i}^{\top}\Phi_{z}}\,\opnorm{\Phi_{z}^{*}\Rr_{z}\Bh_{z}}
\le(2L_{p}+p)\sqrt{w_{z,i}} .
\]

Consequently, we have
$\norm{e_{i}^{\top}(I+\gamma N_{z})\Bh_{z}}
\le\norm{(\Bh_{z})_{i}}+\gamma\,\norm{e_{i}^{\top}N_{z}\Bh_{z}}
\le(1+2L_{p}+p)\sqrt{w_{z,i}}$.
Finally, Eq.~\eqref{I:eq:rowwise-metric} follows from part~(i) via
$(I+\beta N_{z})B_{z}v=\Rr_{z}\Bh_{z}\,H(z)^{1/2}v$ and
Eq.~\eqref{I:eq:resolvent-identity}.
\end{proof}

By differentiating Eq.~\eqref{I:eq:a-def} in $h$, including the factor
$\eta$ from $y=x+\eta h$, we obtain the exact Jacobian. Specifically, for
$y=x+\eta g$ feasible, we have
\begin{equation}
D_{h}a(g)[v]=-\eta\,W_{y}^{-1/2}\,(I+\beta N_{y})\,B_{y}\,v .
\label{I:eq:jacobian}
\end{equation}
Here, the two contributions, $\beta D\log w_{y}[\eta v]$ from
Lemma~\ref{I:lem:lewis-derivative} and
$-\eta(A_{y}v)_{i}=-\eta(W_{y}^{-1/2}B_{y}v)_{i}$ from the slack factor,
combine into Eq.~\eqref{I:eq:jacobian} by
Eq.~\eqref{I:eq:resolvent-identity}.

To make the bootstrap quantitative, we enlarge $\kappa$, if necessary,
so that
\begin{equation}
\vartheta:=L_{p}M_{0}\rho\le\frac{1}{64},
\qquad
\eta\,\tau_{1}\le\frac{1}{64} .
\label{I:eq:bootstrap-smallness}
\end{equation}
By the size bounds in Eq.~\eqref{I:eq:M0-size}, both conditions still
require only a polylogarithmic $\kappa$.

\begin{lemma}[Stability on the good set]\label{I:lem:stability}
Let $g\in\Ktubep$ and $0\le t\le1$. Then the entire segment
$y_{t}=x+t\eta g$ lies in $\intK$, and:
\begin{align}
\norm{a(tg)}_{\infty}&\le2L_{p}M_{0}\rho\le\tfrac{1}{32},
&
\norm{D_{h}a(tg)}_{2\to\infty}&\le e^{1/32}L_{p}\,\eta,
\label{I:eq:stability-inf}\\
\norm{a(tg)}_{W}&\le e^{5/32}\,\tfrac{p}{2}\,M_{0}\rho,
&
\opnorm{W^{1/2}D_{h}a(tg)}&\le e^{5/32}\,\tfrac{p}{2}\,\eta .
\label{I:eq:stability-W}
\end{align}
Moreover
\begin{equation}
e^{-1/4}\,W\preceq W_{y_{t}}\preceq e^{1/4}\,W,
\qquad
e^{-2\norm{a}_{\infty}}I\preceq H(y_{t})\preceq
e^{2\norm{a}_{\infty}}I,
\label{I:eq:weight-metric-comparison}
\end{equation}
and the slack factors satisfy
\begin{equation}
|\zeta_{i}|=w_{i}^{1/p}\,|k_{i}^{\top}g|\le\tau_{1},
\qquad
|\eta\zeta_{i}|\le\eta\tau_{1}\le\tfrac{1}{64} .
\label{I:eq:slack-bound}
\end{equation}
\end{lemma}

\begin{proof}
We first show feasibility. At the base point, we have
$(A_{x}g)_{i}=w_{i}^{-\beta}u_{i}^{\top}g
=w_{i}^{1/p}\,k_{i}^{\top}g$, so
$|(A_{x}g)_{i}|\le\tau_{1}$, which is Eq.~\eqref{I:eq:slack-bound}. Hence
$1+t\eta(A_{x}g)_{i}\ge1-\eta\tau_{1}>0$ for all $t\in[0,1]$. Therefore,
the entire segment is feasible \emph{before} we use any endpoint weight
comparison, and $a(tg)$ is well-defined and smooth in $t$.

Now that the segment is feasible, we run a bootstrap argument with the upper Dini derivative. Let
$a_{\infty}(t):=\norm{a(tg)}_{\infty}$. By
Eq.~\eqref{I:eq:endpoint-diag}, we have
$H(y_{t})\preceq e^{2a_{\infty}(t)}I$. Combining
Eq.~\eqref{I:eq:jacobian} with Eq.~\eqref{I:eq:rowwise-metric} at $y_t$,
we get
\[
|\frac{d}{dt}a_{i}(tg)|
=\eta\,w_{y_{t},i}^{-1/2}\,
|e_{i}^{\top}(I+\beta N_{y_{t}})B_{y_{t}}g|
\le L_{p}\eta\,\norm{g}_{H(y_{t})}
\le L_{p}\eta\,e^{a_{\infty}(t)}\norm{g}.
\]
Hence, the upper Dini derivative satisfies
$D^{+}a_{\infty}(t)\le L_{p}\eta\,e^{a_{\infty}(t)}\norm{g}$. The scalar
comparison here is elementary. As the maximum of finitely
many absolute values of smooth functions, $a_\infty$ is absolutely continuous, and
at almost every $t$ the displayed bound implies
$(e^{-a_\infty(t)})'\ge-\vartheta$. Since $a_{\infty}(0)=0$ and
$L_{p}\eta\norm{g}\le L_{p}\eta R_{1}=L_{p}M_{0}\rho=\vartheta$,
integrating the derivative bound yields $e^{-a_\infty(t)}\ge1-\vartheta t$.
Taking negative logarithms and using $t\le1$ and $\vartheta\le1/64$, we
obtain
\begin{equation}
a_{\infty}(t)\le-\log(1-\vartheta t)
\le-\log(1-\vartheta)\le2\vartheta\le\tfrac{1}{32},
\label{I:eq:ainf-bound}
\end{equation}
which proves the first bound in Eq.~\eqref{I:eq:stability-inf}. The second
bound follows from the same rowwise display and $e^{a_{\infty}}\le e^{1/32}$.

To compare the weights, we use $\beta\ge\tfrac14$,
Eq.~\eqref{I:eq:a-def} and the bound $|\log(1+u)|\le2|u|$ for
$|u|\le1/64$, together with Eq.~\eqref{I:eq:ainf-bound} and $\eta\tau_1\le1/64$, to
obtain
\[
|\log\frac{w_{y_{t},i}}{w_{i}}|
\le\frac{1}{\beta}(a_{\infty}(t)
+|\log(1+t\eta(A_{x}g)_{i})|)
\le4(\frac{1}{32}+2\eta\tau_{1})\le\frac14,
\]
which is the first half of Eq.~\eqref{I:eq:weight-metric-comparison}. The
second half follows immediately from Eq.~\eqref{I:eq:endpoint-diag}.

Finally, we integrate the Jacobian in the base $W$-norm. The factors in
Eq.~\eqref{I:eq:jacobian} satisfy
\[
 \opnorm{I+\beta N_{y_t}}=\opnorm{\Rr_{y_t}}\le p/2,
 \qquad B_{y_t}^{\top}B_{y_t}\preceq H(y_t).
\]
The weight comparison also gives
$\opnorm{W^{1/2}W_{y_t}^{-1/2}}\le e^{1/8}$, and together with
Eq.~\eqref{I:eq:ainf-bound}, these estimates imply
\[
\norm{W^{1/2}D_{h}a(tg)[v]}
\le e^{1/8}\,\frac{p}{2}\,\eta\,\norm{v}_{H(y_{t})}
\le e^{1/8+1/32}\,\frac{p}{2}\,\eta\,\norm{v}
=e^{5/32}\,\frac{p}{2}\,\eta\,\norm{v}.
\]
This proves the operator bound in Eq.~\eqref{I:eq:stability-W}.
Integrating along $t\mapsto tg$ from $a(0)=0$, we obtain
$\norm{a(tg)}_{W}\le e^{5/32}\tfrac{p}{2}\,\eta\norm{g}t
\le e^{5/32}\tfrac{p}{2}M_{0}\rho$.
\end{proof}

We note that all constants in Lemma~\ref{I:lem:stability} are
dimension-free, since the slack ratio in Eq.~\eqref{I:eq:slack-bound} is
controlled by the scaled quantity $\eta\tau_{1}$ rather than by the
unscaled threshold $\tau_{1}$.

On the good set, the proposal-ratio identity Eq.~\eqref{I:eq:ledger}
reduces to the centered fluctuation
\begin{equation}
Z(h):=(c-\wb)^{\top}a
\label{I:eq:Z-def}
\end{equation}
up to two finite error terms, which we now bound. Recall from
Lemma~\ref{I:lem:stability} that on $\Ktubep$ we have
$c_{i}\le\tau_{1}^{2}w_{i}$ (since
$c_{i}=w_{i}(k_{i}^{\top}h)^{2}$), $\wb_{i}\le\norm{u_{i}}^{2}=w_{i}$,
and $\norm{a}_{\infty}\le\tfrac{1}{32}$.

\begin{lemma}[Quadratic remainder]\label{I:lem:remainder-quad}
On $\Ktubep$,
\begin{equation}
|(c-\wb)^{\top}(b-2a)|
\le C\,(\tau_{1}^{2}+1)\,\norm{a}_{W}^{2}
\le C_{p}\,M_{0}^{2}\rho^{2}\,(\tau_{1}^{2}+1) .
\label{I:eq:remainder-quad}
\end{equation}
\end{lemma}

\begin{proof}
Since $|c_{i}-\wb_{i}|\le(\tau_{1}^{2}+1)w_{i}$ and
$|b_{i}-2a_{i}|=|e^{2a_{i}}-1-2a_{i}|\le Ca_{i}^{2}$ for
$|a_{i}|\le\tfrac1{32}$, we obtain the claim by summing over $i$ and
applying Eq.~\eqref{I:eq:stability-W}.
\end{proof}

\begin{lemma}[Hadamard--Frobenius bound]\label{I:lem:frobenius}
Let $P_{Q}:=UQU^{\top}$. Then, on $\Ktubep$,
\begin{equation*}
\fnorm{E}^{2}=b^{\top}(P_{Q}\circ P_{Q})b
\le b^{\top}Wb\le C\norm{a}_{W}^{2}\le C_{p}M_{0}^{2}\rho^{2}.
\end{equation*}
\end{lemma}

\begin{proof}
By Eqs.~\eqref{I:eq:E-def} and~\eqref{I:eq:endpoint-diag}, the cyclicity
of the trace and an expansion over the rows of $U$, we have
\[
\fnorm{E}^{2}=\tr[QU^{\top}\Diag(b)UQU^{\top}\Diag(b)U]=\sum_{i,j}b_{i}b_{j}(u_{i}^{\top}Qu_{j})^{2}=b^{\top}(P_{Q}\circ P_{Q})b .
\]
Since $0\preceq P_{Q}\preceq P$, we use the following one-line proof of
the Schur product fact. If positive semidefinite matrices $C,D$ have Gram
representations $C_{ij}=\langle c_i,c_j\rangle$ and
$D_{ij}=\langle d_i,d_j\rangle$, then
$(C\circ D)_{ij}=\langle c_i\otimes d_i,c_j\otimes d_j\rangle$, and hence
$C\circ D\succeq0$. Applying this fact to the pairs $P,P-P_Q$ and
$P-P_Q,P_Q$, we obtain
\[
P\circ P-P_{Q}\circ P_{Q}=P\circ(P-P_{Q})+(P-P_{Q})\circ P_{Q}\succeq0,\qquad\text{so}\qquad P_{Q}\circ P_{Q}\preceq P\circ P .
\]
To prove $P\circ P\preceq\Diag(P)=W$, we fix any $x\in\R^{n}$ and
set $M:=\sum_{i}x_{i}u_{i}u_{i}^{\top}$. By the weighted Cauchy--Schwarz
inequality and Eq.~\eqref{I:eq:phi-contraction}, we have
\[
x^{\top}(P\circ P)x=\fnorm{M}^{2}=\sum_{i}x_{i}\,u_{i}^{\top}Mu_{i}\le\norm{x}_{W}(\sum_{i}\tfrac{(u_{i}^{\top}Mu_{i})^{2}}{w_{i}})^{1/2}\le\norm{x}_{W}\fnorm{M}.
\]
Hence $\fnorm{M}\le\norm{x}_{W}$ and
$x^{\top}(P\circ P)x\le x^{\top}Wx$.

Finally, since $|b_{i}|\le C|a_{i}|$ for $|a_{i}|\le\tfrac1{32}$, we have
$b^{\top}Wb\le C\norm{a}_{W}^{2}$, and Eq.~\eqref{I:eq:stability-W}
completes the proof.
\end{proof}

\begin{lemma}[Determinant remainder]\label{I:lem:remainder-det}
On $\Ktubep$, every eigenvalue of $I+E$ belongs to an interval satisfying
\[
[e^{-2\norm{a}_{\infty}},\,e^{2\norm{a}_{\infty}}]
\subseteq[e^{-1/16},e^{1/16}].
\]
Moreover,
\begin{equation}
|\log\det(I+E)-\tr[E]|
\le C\,\fnorm{E}^{2}\le C_{p}\,M_{0}^{2}\rho^{2}.
\label{I:eq:remainder-det}
\end{equation}
\end{lemma}

\begin{proof}
Let $a_{0}=e^{-2\norm{a}_{\infty}}$ and
$b_{0}=e^{2\norm{a}_{\infty}}$. By
Eq.~\eqref{I:eq:weight-metric-comparison}, we have
$a_{0}I\preceq H(y)\preceq b_{0}I$. Since $a_{0}\le1\le b_{0}$ and
$\creg\succeq0$, adding $a_0\creg\preceq\creg$ and
$\creg\preceq b_0\creg$ to these bounds gives
\[
a_{0}(I+\creg)\preceq H(y)+\creg
\preceq b_{0}(I+\creg).
\]
By congruence with $Q^{1/2}=(I+\creg)^{-1/2}$, we get
$a_{0}I\preceq I+E\preceq b_{0}I$ without requiring $H(y)$
and $\creg$ to commute. Hence, the eigenvalues $\lambda$ of $E$ obey
$|\lambda|<1/8$, and the fundamental theorem of calculus gives
\[
|\log(1+\lambda)-\lambda|
=|\int_0^\lambda\frac{-t}{1+t}\,dt|
\le\frac{\lambda^2}{2(1-|\lambda|)}\le C\lambda^2.
\]
Summing, we obtain $|\log\det(I+E)-\tr[E]|\le C\fnorm{E}^{2}$, and
Lemma~\ref{I:lem:frobenius} completes the proof.
\end{proof}

Thus, a single finite quadratic remainder controls \emph{all} orders of
the determinant. Our argument never expands $\log\det$ beyond
second order and never uses the convexity of $\log\det G$.
\subsection{A row-dependent contraction of Lewis Hessians}\label{I:sec:hessian}

In Section~\ref{I:sec:stein}, we use a Stein argument to convert the mean of
the centered fluctuation $Z=(c-\wb)^{\top}a$ into a sum of second derivatives
of $\log w$ evaluated at \emph{row-dependent} directions of magnitude
$\sqrt{w_{i}}$. In this section, we prove a deterministic lemma that bounds
this sum by $C_{p}d$, with no polynomial factor of $n$ (only through
$C_{p}$). We note that this lemma is the only genuinely new piece of
Lewis-weight calculus in this paper.

Throughout this section, we work at an arbitrary fixed point of $\intK$,
normalized as in Section~\ref{I:sec:cancellation} so that $U^{\top}U=I$, and
all derivatives are taken with respect to the spatial point in these
coordinates. In Section~\ref{I:sec:stein}, we will apply the lemma at the
random endpoint $y$, in $y$-normalized coordinates. We abbreviate
$\ell:=\log w$, use $u_{i}^{\top}$ to denote the rows of $U$, and set
$B:=W^{1/p}U=W^{1/2}A$. Letting $\Phi,\Lb,\Rr,N$ denote the objects of
Eq.~\eqref{I:eq:frame-objects} at this point, we have
$\opnorm{\Phi}\le1$, $\opnorm{B}\le1$, $\opnorm{\Rr}\le p/2$,
$\opnorm{N}\le p$, $\norm{u_{i}}^{2}=w_{i}$, and
$\norm{B_{i}}=w_{i}^{1/2+1/p}$.

The contraction relies on three ingredients: an exact second-derivative
identity (Lemma~\ref{I:lem:second-derivative}), a frame-derivative computation
(Lemma~\ref{I:lem:horizontal-frame}), and an aggregate bound on the derivative
of the resolvent kernel (Lemma~\ref{I:lem:DN-aggregate}). For a fixed
direction $r$, we define
\begin{equation}
\ell_{r}:=D\ell[r]=-W^{-1/2}NBr,
\qquad
s_{r}:=Ar=W^{-1/2}Br,
\qquad
t_{r}:=\tfrac12\ell_{r}-s_{r} .
\label{I:eq:first-derivative-objects}
\end{equation}

\begin{lemma}[Exact Hessian identity]\label{I:lem:second-derivative}
\begin{equation}
D^{2}\ell[r,r]
=-\tfrac12\,\ell_{r}^{\circ2}
-W^{-1/2}\,DN[r]\,Br
-W^{-1/2}\,N\,\Diag(t_{r})\,Br .
\label{I:eq:second-derivative}
\end{equation}
\end{lemma}

\begin{proof}
We differentiate $\ell_{v}=-W^{-1/2}NBv$ at fixed $v=r$ and apply the product
rule to its three point-dependent factors. From the identities
$DW[r]=W\Diag(\ell_{r})$ and $DA[r]=-\Diag(s_{r})A$, we obtain
\[
D(W^{-1/2})[r]
=-\tfrac12\Diag(\ell_{r})W^{-1/2}
\]
and, by the product rule,
\[
\begin{aligned}
DB[r]=D(W^{1/2}A)[r]
&=\Diag(\tfrac12\ell_{r}-s_{r})B \\
&=\Diag(t_{r})B.
\end{aligned}
\]
Combining these identities, we obtain
\[
D^{2}\ell[r,r]
=\tfrac12\Diag(\ell_{r})\,W^{-1/2}NBr
-W^{-1/2}DN[r]Br
-W^{-1/2}N\Diag(t_{r})Br,
\]
where the first term equals
$\tfrac12\Diag(\ell_{r})(-\ell_{r})=-\tfrac12\ell_{r}^{\circ2}$.
\end{proof}

Note that the operators $\Phi,\Lb,N$ depend on the point only through the
Gram matrix of the rows of the orthonormalized frame
$V_{z}=W_{z}^{\beta}A_{z}H(z)^{-1/2}$, and hence they do not depend on the
choice of square root. However, to differentiate them, we fix a smooth choice, and the computation is cleanest in the \emph{horizontal}
gauge.

\begin{lemma}[Horizontal frame derivative]\label{I:lem:horizontal-frame}
Fix the symmetric square root used to define the frame $U(z)$, where
\[
U(z):=W_{z}^{\beta}A_{z}H(z)^{-1/2}.
\]
At the normalized base point, where $H=I$, we have
\begin{equation*}
DU[r]=P^{\perp}\,\Diag(z_{r})\,U,
\qquad
P^{\perp}:=I-UU^{\top},
\qquad
W^{1/2}z_{r}:=-(I+\beta N)Br .
\end{equation*}
Consequently
$\fnorm{DU[r]}\le\norm{z_{r}}_{W}=\norm{(I+\beta N)Br}
\le\tfrac{p}{2}\norm{r}$.
\end{lemma}

\begin{proof}
Let $z_{r}:=\beta\ell_{r}-s_{r}$. By
Eq.~\eqref{I:eq:first-derivative-objects}, it satisfies
$W^{1/2}z_{r}=-\beta NBr-Br=-(I+\beta N)Br$. As in the proof of
Lemma~\ref{I:lem:second-derivative}, we have
$D(W^{\beta}A)[r]=\Diag(z_{r})\,W^{\beta}A$, which equals
$\Diag(z_{r})U$ at the base point. Therefore, by the product rule at $H=I$,
we have
\[
DU[r]=\Diag(z_{r})\,U+U\cdot H^{1/2}D(H^{-1/2})[r]
|_{H=I}
=\Diag(z_{r})\,U+U\,D(H^{-1/2})[r] .
\]

At $H=I$, differentiating $H=(W^{\beta}A)^{\top}(W^{\beta}A)$
gives $DH[r]=2\,U^{\top}\Diag(z_{r})U$. For the symmetric square root, we
have $D(H^{-1/2})[r]=-\tfrac12DH[r]=-U^{\top}\Diag(z_{r})U$, since the
derivative $X$ of $H^{-1/2}$ solves
$XH^{-1/2}+H^{-1/2}X=D(H^{-1})[r]=-DH[r]$ at $H=I$. Substituting this
derivative and using $P^\perp=I-UU^\top$, we obtain
\[
DU[r]=\Diag(z_{r})U-UU^{\top}\Diag(z_{r})U
=P^{\perp}\Diag(z_{r})U .
\]

It remains to prove the norm bound. We have
$\fnorm{DU[r]}\le\fnorm{\Diag(z_{r})U}
=(\sum_{i}z_{r,i}^{2}\norm{u_{i}}^{2})^{1/2}
=\norm{z_{r}}_{W}$. Moreover, by Eq.~\eqref{I:eq:resolvent-identity}, we
have $\norm{(I+\beta N)Br}\le\opnorm{\Rr}\,\opnorm{B}\,\norm{r}
\le\tfrac{p}{2}\norm{r}$.
\end{proof}

\begin{lemma}[Aggregate resolvent-derivative bound]\label{I:lem:DN-aggregate}
For every $r$, $\fnorm{DN[r]\,B}\le C_{p}\norm{r}$, and hence, summing
over an orthonormal basis $(e_{a})_{a=1}^{d}$,
\begin{equation}
\sum_{a=1}^{d}\fnorm{DN[e_{a}]\,B}^{2}\ \le\ C_{p}\,d .
\label{I:eq:DN-aggregate}
\end{equation}
\end{lemma}

\begin{proof}
Since $\Rr=(I-c_{p}\Lb)^{-1}$, we have $N=2\Lb\Rr=\tfrac{2}{c_{p}}(\Rr-I)$,
which gives $DN[r]=\tfrac{2}{c_{p}}D\Rr[r]=2\,\Rr\,D\Lb[r]\,\Rr$. Writing
$\Phi_{i}=u_{i}u_{i}^{\top}/\norm{u_{i}}\in\Sd$, we use the following derivative for $u\ne0$:
\[
D(\frac{uu^{\top}}{\norm u})[v]
=\frac{vu^{\top}+uv^{\top}}{\norm u}
-\frac{uu^{\top}}{\norm u^{3}}\langle u,v\rangle .
\]
The three displayed summands have Frobenius norms at most $\norm v$,
$\norm v$, and $\norm v$, respectively, and hence the derivative has norm at
most $3\norm v$. Let $\Phi_{r}:=D\Phi[r]$. By this componentwise derivative
bound and Lemma~\ref{I:lem:horizontal-frame}, we obtain
\[
\fnorm{\Phi_{r}}^{2}\le9\sum_{i}\norm{Du_{i}[r]}^{2}
=9\,\fnorm{DU[r]}^{2}\le9\,\tfrac{p^{2}}{4}\,\norm{r}^{2}.
\]

The identity $D\Lb[r]=-(\Phi_{r}\Phi^{*}+\Phi\Phi_{r}^{*})$ then
gives $\fnorm{D\Lb[r]}\le2\fnorm{\Phi_{r}}\,\opnorm{\Phi}\le3p\norm{r}$.
Putting everything together, we obtain
\[
\fnorm{DN[r]B}
\le2\,\opnorm{\Rr}\,\fnorm{D\Lb[r]}\,\opnorm{\Rr B}
\le2\cdot\frac{p}{2}\cdot3p\norm{r}\cdot\frac{p}{2}
=C_{p}\norm{r},
\]
where the first step follows from submultiplicativity, the second step
follows from the bounds on $\Rr$, $D\Lb[r]$, and $\Rr B$, and the last step
absorbs the fixed polynomial in $p$ into $C_p$.
\end{proof}

\begin{lemma}[Row-dependent Hessian contraction]
\label{I:lem:hessian-contraction}
If directions $r_{i}\in\R^{d}$ satisfy
$\norm{r_{i}}\le \mu\sqrt{w_{i}}$ for all $i$, then
\begin{equation}
\sum_{i=1}^{n}|D^{2}\ell_{i}[r_{i},r_{i}]|
\ \le\ C_{p}\,\mu^{2}\,d .
\label{I:eq:hessian-contraction}
\end{equation}
\end{lemma}

\begin{proof}
We write $r_{i}=\sqrt{w_{i}}\,q_{i}$ with $\norm{q_{i}}\le \mu$ and take the
$i$-th component of Eq.~\eqref{I:eq:second-derivative} with direction
$r_{i}$. Since each of the three terms is bilinear in $r_{i}$, it
contributes a factor $w_{i}$ that combines with the prefactor
$w_{i}^{-1/2}$. We bound the three terms in turn.

For the first terms, the substitution gives
$(\ell_{r_{i}})_{i}^{2}=(e_{i}^{\top}NBq_{i})^{2}$. By $\norm{q_i}\le\mu$,
Frobenius-norm submultiplicativity, $\opnorm N\le p$, $w_i\le1$, and
$\sum_iw_i=d$, we obtain
\begin{equation}
\sum_{i}\tfrac12(\ell_{r_{i}})_{i}^{2}
\le\tfrac12 \mu^{2}\,\fnorm{NB}^{2}
\le\tfrac12 \mu^{2}\,\opnorm{N}^{2}\fnorm{B}^{2}
\le\tfrac12 \mu^{2}p^{2}\sum_{i}w_{i}^{1+2/p}
\le C_{p}\mu^{2}d .
\label{I:eq:term1}
\end{equation}

For the middle terms, the $i$-th component is
$\sqrt{w_{i}}\,e_{i}^{\top}DN[q_{i}]Bq_{i}$. By Cauchy--Schwarz with weights
$\sqrt{w_{i}}$ and $\sum_{i}w_{i}=d$, we have
\[
\sum_{i}\sqrt{w_{i}}\,|e_{i}^{\top}DN[q_{i}]Bq_{i}|
\le\sqrt{d}\,(\sum_{i}
|e_{i}^{\top}DN[q_{i}]Bq_{i}|^{2})^{1/2}.
\]
To bound the right-hand side, we expand
\[
DN[q_{i}]Bq_{i}
=\sum_{a,b}(q_{i})_{a}(q_{i})_{b}\,DN[e_{a}]Be_{b},
\qquad
T^{(i)}_{ab}:=e_{i}^{\top}DN[e_{a}]Be_{b}.
\]
By Cauchy--Schwarz and $\norm{q_i}\le\mu$, we then have
\[
|e_{i}^{\top}DN[q_{i}]Bq_{i}|
\le\norm{q_{i}}^{2}\fnorm{T^{(i)}}
\le \mu^{2}\fnorm{T^{(i)}}.
\]
Moreover, the definition of $T^{(i)}$ and Eq.~\eqref{I:eq:DN-aggregate}
give
\[
\sum_{i}\fnorm{T^{(i)}}^{2}
=\sum_{a}\sum_{i,b}(e_{i}^{\top}DN[e_{a}]Be_{b})^{2}
=\sum_{a}\fnorm{DN[e_{a}]B}^{2}\le C_{p}d.
\]
Putting these bounds together and absorbing the square root of $C_p$ into
$C_p$, we obtain
\begin{equation}
\sum_{i}\sqrt{w_{i}}\,|e_{i}^{\top}DN[q_{i}]Bq_{i}|
\le\sqrt{d}\cdot \mu^{2}\sqrt{C_{p}d}=C_{p}\mu^{2}d .
\label{I:eq:term2}
\end{equation}

Finally, for the last terms, we set
\[
\Upsilon_{0}:=I+\tfrac12N,\qquad
L_{0}:=-W^{-1/2}\Upsilon_{0}B,\qquad
T_{i}:=L_{0}^{\top}\Diag(e_{i}^{\top}N)\,B
=\sum_{j}N_{ij}\,(L_{0})_{j}^{\top}(B)_{j},
\]
where $(L_{0})_{j},(B)_{j}$ denote rows and
$(L_{0})_{j}^{\top}(B)_{j}$ their rank-one outer products. Since
$t_{q}=\tfrac12\ell_{q}-s_{q}=-W^{-1/2}(I+\tfrac12N)Bq=L_{0}q$ by
Eq.~\eqref{I:eq:first-derivative-objects}, the $i$-th component of the last
term of Eq.~\eqref{I:eq:second-derivative} with direction $r_{i}$ is
$\sqrt{w_{i}}\;q_{i}^{\top}T_{i}q_{i}$. To bound these components, we view
$((L_{0})_{j}^{\top}(B)_{j})_{j\le n}$ as a vector with entries in the
Hilbert space $(\R^{d\times d},\fnorm{\cdot})$, with $N$ acting on the index
$j$. By the Hilbert-space operator-norm bound, the row identities for $L_0$
and $B$, and Lemma~\ref{I:lem:rowwise}(ii), we obtain
\[
\begin{aligned}
\sum_{i}\fnorm{T_{i}}^{2}
&\le\opnorm{N}^{2}
\sum_{j}\norm{(L_{0})_{j}}^{2}\norm{(B)_{j}}^{2} \\
&=\opnorm{N}^{2}\sum_{j}w_{j}^{2/p}
\norm{e_{j}^{\top}\Upsilon_{0}B}^{2} \\
&\le p^{2}(L_{p}^{+})^{2}\sum_{j}w_{j}^{1+2/p}
\le C_{p}\,d,
\end{aligned}
\]
where we apply Lemma~\ref{I:lem:rowwise}(ii) with $\gamma=\tfrac12$ at the
normalized point $\Bh=B$. Hence, by the weighted Cauchy--Schwarz inequality
and $\sum_iw_i=d$, we obtain
\begin{equation}
\sum_{i}\sqrt{w_{i}}\,|q_{i}^{\top}T_{i}q_{i}|
\le \mu^{2}\sum_{i}\sqrt{w_{i}}\,\fnorm{T_{i}}
\le \mu^{2}\sqrt{\textstyle\sum_{i}w_{i}}
\sqrt{\textstyle\sum_{i}\fnorm{T_{i}}^{2}}
\le C_{p}\mu^{2}d .
\label{I:eq:term3}
\end{equation}

Combining Eq.~\eqref{I:eq:term1}, \eqref{I:eq:term2} and \eqref{I:eq:term3}
with Lemma~\ref{I:lem:second-derivative}, we obtain
Eq.~\eqref{I:eq:hessian-contraction}.
\end{proof}

\begin{remark}
We state the lemma for the row-dependent hypothesis
$\norm{r_{i}}\le \mu\sqrt{w_{i}}$ because the Stein identity in
Section~\ref{I:sec:stein} produces precisely the directions $Qu_{i}$, whose norms are
at most $\sqrt{w_{i}}$. Thus, the Stein directions satisfy exactly the row-dependent
hypothesis of the lemma.
\end{remark}
\subsection{A Stein bound for the centered fluctuation, and the acceptance
lower tail}\label{I:sec:stein}

In this section, we control the centered fluctuation
\[
Z(h)=(c-\wb)^{\top}a,\qquad h\sim N(0,Q),
\]
and then assemble the uniform lower tail for the complete Metropolis ratio
(Lemma~\ref{I:lem:acceptance}). We bound the mean of $Z$ by Gaussian
integration by parts and reduce it through
Lemma~\ref{I:lem:hessian-contraction}, while a dimension-free Lipschitz
bound controls the deviations of $Z$. Both steps take place on the good set
of Section~\ref{I:sec:stability}. Since the Stein identity is stated on
the full space, we reconcile the two through a cutoff and a global extension.

Since the slack term is the composition of $-\log(1+\cdot)$ with a linear
map, we can compute the relevant second derivative exactly by
differentiating Eq.~\eqref{I:eq:a-def} twice in $h$, which gives
\begin{equation}
D^{2}_{h}a_{i}[v,v]
=\eta^{2}(\beta\,D^{2}\ell_{y,i}[v,v]+(A_{y}v)_{i}^{2}),
\qquad y=x+\eta h .
\label{I:eq:a-hessian}
\end{equation}

To apply Lemma~\ref{I:lem:hessian-contraction} at $y$, we pass to
$y$-normalized coordinates by letting $V=W_{y}^{\beta}A_{y}H(y)^{-1/2}$ with
rows $v_{i}$, so that $V^{\top}V=I$ and $\norm{v_{i}}^{2}=w_{y,i}$. By
Eq.~\eqref{I:eq:endpoint-diag}, we have
\begin{equation}
u_{i}=e^{-a_{i}}\,H(y)^{1/2}v_{i},
\qquad\text{hence}\qquad
H(y)^{1/2}Qu_{i}=e^{-a_{i}}\,H(y)^{1/2}QH(y)^{1/2}\,v_{i} .
\label{I:eq:frame-transport}
\end{equation}
On $\Ktubep$, Lemma~\ref{I:lem:stability} and $Q\preceq I$ give
the bounds
$\opnorm{H(y)^{1/2}QH(y)^{1/2}}\le e^{2\norm{a}_{\infty}}\le e^{1/16}$
and $e^{-a_{i}}\le e^{1/32}$, which together imply
\begin{equation}
\norm{H(y)^{1/2}Qu_{i}}\le C\sqrt{w_{y,i}}
\qquad\text{for all }i .
\label{I:eq:direction-bound}
\end{equation}
Since the vectors $H(y)^{1/2}Qu_{i}$ are the $y$-normalized versions of the
Stein directions $Qu_{i}$, Eq.~\eqref{I:eq:direction-bound} is the
row-dependent hypothesis of Lemma~\ref{I:lem:hessian-contraction} with
$\mu=O(1)$. For the slack part of Eq.~\eqref{I:eq:a-hessian}, we derive a
sharper bound directly by setting
$Q_{y}:=H(y)^{1/2}QH(y)^{1/2}\preceq CI$ and using
$\norm{v_i}^2=w_{y,i}$, $2-2\beta=1+2/p$, $w_{y,i}\le1$, and
$\sum_iw_{y,i}=d$, which gives
\begin{equation}
(A_{y}Qu_{i})_{i}
=e^{-a_{i}}\,w_{y,i}^{-\beta}\,v_{i}^{\top}Q_{y}v_{i}
\qquad\Longrightarrow\qquad
\sum_{i}(A_{y}Qu_{i})_{i}^{2}
\le C\sum_{i}w_{y,i}^{2-2\beta}
= C\sum_{i}w_{y,i}^{1+2/p}\le Cd .
\label{I:eq:slack-hessian}
\end{equation}

\begin{lemma}[Mean bound on the good set]\label{I:lem:mean}
Let $\chi$ be any $[0,1]$-valued cutoff supported in $\Ktubep$. Then
\begin{equation*}
|\sum_{i}\E[\chi(h)\,(Qu_{i})^{\top}D^{2}_{h}a_{i}(h)\,(Qu_{i})]|
\ \le\ C_{p}\,\eta^{2}d\ =\ C_{p}\,\rho^{2}.
\end{equation*}
\end{lemma}

\begin{proof}
By the coordinate covariance of second derivatives,
$D^{2}\ell_{y,i}$ evaluated at the direction $Qu_{i}$ in base coordinates
equals the $y$-normalized Hessian evaluated at $H(y)^{1/2}Qu_{i}$.
Combining Eq.~\eqref{I:eq:a-hessian}, Eq.~\eqref{I:eq:frame-transport} and
this covariance with Lemma~\ref{I:lem:hessian-contraction} and
Eqs.~\eqref{I:eq:direction-bound} and~\eqref{I:eq:slack-hessian}, we obtain,
pointwise on $\Ktubep$,
\[
\begin{aligned}
&\sum_{i}|(Qu_{i})^{\top}D^{2}_{h}a_{i}\,(Qu_{i})| \\
&\quad\le\eta^{2}(
\beta\sum_{i}
|D^{2}\ell_{y,i}
[H(y)^{1/2}Qu_{i},H(y)^{1/2}Qu_{i}]|_{y\text{-norm.}}
+\sum_{i}(A_{y}Qu_{i})_{i}^{2}) \\
&\quad\le\eta^{2}(C_{p}d+Cd).
\end{aligned}
\]
It remains to multiply by $\chi\le1$ and take expectations.
\end{proof}

\begin{lemma}[Dimension-free Lipschitz bound]\label{I:lem:lipschitz}
On the convex set $\Ktube$,
\begin{equation}
\Lip(Z|_{\Ktube})\ \le\
C_{p}\,M_{0}\rho\,(\tau_{1}^{2}+1),
\label{I:eq:lipschitz}
\end{equation}
with no power of $d$.
\end{lemma}

\begin{proof}
Differentiating Eq.~\eqref{I:eq:Z-def}, we obtain
\begin{equation*}
\nabla Z=2\,U^{\top}\Diag(a)\,Uh+(D_{h}a)^{\top}(c-\wb) .
\end{equation*}
For the first term, we write
$U^{\top}\Diag(a)Uh=\sum_{i}u_{i}\,a_{i}\sqrt{w_{i}}\,(k_{i}^{\top}h)$
and use $\opnorm{U^{\top}}\le1$ and $|k_{i}^{\top}h|\le\tau_{1}$ on the good
set to obtain
\[
\norm{U^{\top}\Diag(a)Uh}
\le(\sum_{i}a_{i}^{2}w_{i}(k_{i}^{\top}h)^{2})^{1/2}
\le\tau_{1}\norm{a}_{W} .
\]

For the second term, the bound
$|c_{i}-\wb_{i}|\le(\tau_{1}^{2}+1)w_{i}$ implies
$\norm{W^{-1/2}(c-\wb)}\le(\tau_{1}^{2}+1)\sqrt{d}$. Inserting
$W^{1/2}W^{-1/2}$, applying Eq.~\eqref{I:eq:stability-W}, and using
$\eta\sqrt d=\rho$, we obtain
\[
\norm{(D_{h}a)^{\top}(c-\wb)}
\le\opnorm{W^{1/2}D_{h}a}\,\norm{W^{-1/2}(c-\wb)}
\le e^{5/32}\tfrac{p}{2}\,\eta\,(\tau_{1}^{2}+1)\sqrt{d}
=C_{p}\,\rho\,(\tau_{1}^{2}+1).
\]

Combining these bounds with
$\norm{a}_{W}\le C_{p}M_{0}\rho$ from Eq.~\eqref{I:eq:stability-W} and
$M_{0}\ge1$ yields Eq.~\eqref{I:eq:lipschitz}, since the convexity of
$\Ktube$ converts the gradient bound into a Lipschitz bound.
\end{proof}

Since the function $a(h)$ is naturally defined only while $x+\eta h$
remains in the stable feasible region, we cannot apply
Lemma~\ref{I:lem:stein} to $a_{i}$ directly. To circumvent this domain
issue, we apply Lemma~\ref{II:lem:cutoff} to the closed convex set $\Ktube$
and obtain a function $\chi:\R^{d}\to[0,1]$ that equals $1$ on $\Ktube$, is
supported in $\Ktubep$, and satisfies $\norm{\nabla\chi}\le12$ and
$\opnorm{D^{2}\chi}\le C_{0}$ almost everywhere for a universal $C_{0}$.

Each product $a_{i}\chi$ is $C^{1,1}$ on the open set in which $\Ktubep$
is compactly contained, and by Lemma~\ref{I:lem:stability} the segment
stays strictly feasible on this open set. Since the cutoff vanishes outside
$\Ktubep$, each product extends by zero to a global $W^{2,\infty}$ function
of compact support. Thus, we may apply Lemma~\ref{I:lem:stein} to
$F=a_{i}\chi$ and sum over $i$, which gives
\begin{equation}
\E[Z\chi]
=\sum_{i}\E[(Qu_{i})^{\top}D^{2}(a_{i}\chi)(h)\,(Qu_{i})].
\label{I:eq:stein-applied}
\end{equation}
To bound the right-hand side, we use the expansion
$D^{2}(a_{i}\chi)=\chi D^{2}a_{i}
+(\nabla a_{i})(\nabla\chi)^{\top}+(\nabla\chi)(\nabla a_{i})^{\top}
+a_{i}D^{2}\chi$, whose leading term is bounded by Lemma~\ref{I:lem:mean}.
The remaining three cutoff terms are supported on $\Ktubep\setminus\Ktube$,
an event of probability at most $\delta_{0}$. On this set, the two gradient
terms together have deterministic size at most
\[
2\sum_i
|(Qu_i)^\top\nabla a_i|
|(Qu_i)^\top\nabla\chi|
\le C_pL_p\eta\sum_iw_i=C_pL_p\rho\sqrt d,
\]
where we use $\norm{Qu_{i}}\le\sqrt{w_{i}}$, $\sum_iw_i=d$,
$\norm{\nabla\chi}\le C$, and the rowwise bounds
$|(\nabla a_{i})^{\top}Qu_{i}|\le C L_{p}\eta\sqrt{w_{i}}$, which follow from
Eq.~\eqref{I:eq:jacobian}, \eqref{I:eq:rowwise-metric} and
\eqref{I:eq:direction-bound}. Similarly, on the same set, the Hessian-of-cutoff term has deterministic size at most
\[
\sum_i |a_i|\,\opnorm{D^2\chi}\norm{Qu_i}^{2}
\le C\sum_i|a_i|w_i
\le C\sqrt d\,\norm a_W
\le C_pM_0\rho\sqrt d.
\]
Integrating these sizes over this set and using
$\delta_0\le\sqrt{\delta_0}$ gives the common upper bound
\begin{equation}
C_{p}\,M_{0}\rho\,(1+\sqrt{d})\,\sqrt{\delta_{0}} .
\label{I:eq:cutoff-error}
\end{equation}
Finally, we apply Lemma~\ref{I:lem:mcshane} to $Z|_{\Ktube}$ and denote by
$\widetilde{Z}:\R^{d}\to\R$ the extension with the same Lipschitz constant.
The means of $\widetilde{Z}$ and $Z\chi$ agree up to an error supported off
$\Ktube$. Since the extension preserves values on $\Ktube$, and since
$0\in\Ktube$ and $a(0)=0$ give $Z(0)=0$, the Lipschitz constant $\Lambda$
of the extension yields the global growth bound
$|\widetilde Z(h)|\le\Lambda\norm h$. Using this growth bound and the
pointwise bound $|Z|\le(\tau_{1}^{2}+1)\sqrt{d}\,\norm{a}_{W}$ on
$\Ktubep$, we bound the second moments by
$\E[\widetilde Z(h)^2]\le\Lambda^2\E[\norm h^2]\le\Lambda^2d$ and
$\E[(Z\chi)^2]\le C_pM_0^2\rho^2(\tau_1^2+1)^2d$. Cauchy--Schwarz on
the probability-$\delta_{0}$ complement then gives
\begin{equation}
|\E[\widetilde{Z}]-\E[Z\chi]|
\le C_{p}\,M_{0}\rho\,(\tau_{1}^{2}+1)\,\sqrt{d}\,\sqrt{\delta_{0}} .
\label{I:eq:extension-error}
\end{equation}
Hence, with the choice $\delta_{0}=\eone^{2}d^{-3}$, both the cutoff error
Eq.~\eqref{I:eq:cutoff-error} and the extension error
Eq.~\eqref{I:eq:extension-error} are at most
$C_{p}M_{0}\rho\,(\tau_{1}^{2}+1)\,\eone/d$.

\begin{lemma}[Concentration of the fluctuation]\label{I:lem:Z-bound}
With probability at least $1-2\eone$ over $h\sim N(0,Q)$,
\begin{equation}
|Z|\ \le\
C_{p}[M_{0}^{2}\rho^{2}
+M_{0}\rho\,(\tau_{1}^{2}+1)(\sqrt{\log(2/\eone)}+\eone)].
\label{I:eq:Z-bound}
\end{equation}
\end{lemma}

\begin{proof}
We write $h=Q^{1/2}\xi$. Since $Q\preceq I$, Lemma~\ref{I:lem:lipschitz}
implies that the function $\widetilde Z(Q^{1/2}\cdot)$ is Lipschitz in
$\xi$ with constant
$\Lip(\widetilde Z)\le C_{p}M_{0}\rho(\tau_{1}^{2}+1)$. Hence, by
Lemma~\ref{I:lem:gaussian-concentration}, we have
$|\widetilde Z-\E[\widetilde Z]|
\le\Lip(\widetilde Z)\sqrt{2\log(2/\eone)}$ with probability
$\ge1-\eone$. By Lemma~\ref{I:lem:mean},
Eq.~\eqref{I:eq:stein-applied}, Eq.~\eqref{I:eq:cutoff-error} and
Eq.~\eqref{I:eq:extension-error}, the mean satisfies
$|\E[\widetilde Z]|\le C_{p}\rho^{2}+C_{p}M_{0}\rho(\tau_{1}^{2}+1)\eone/d$.
On the event $h\in\Ktube$, which has probability
$\ge1-\delta_{0}\ge1-\eone$, we have $Z=\widetilde Z$. Combining the
preceding bounds on the intersection of these events proves the claim.
\end{proof}

For fixed $\eone$, it suffices to take
\begin{equation}
\rho\ \le\ \frac{c}{\Gamma_{p}\,M_{0}\,(1+\log(nd/\eone))},
\qquad
\kappa=\frac{r^{2}}{\rho^{2}},
\label{I:eq:rho-choice}
\end{equation}
where $\Gamma_{p}$ is a sufficiently large polynomial in $p$. With this
choice, Eq.~\eqref{I:eq:bootstrap-smallness} holds and $\kappa$ is
polylogarithmic in $n$ and $d$. Moreover, by Eq.~\eqref{I:eq:M0-size}, the
same choice makes the right-hand sides of Eq.~\eqref{I:eq:remainder-quad},
\eqref{I:eq:remainder-det}, and \eqref{I:eq:Z-bound} smaller than any
desired universal constant.

\begin{lemma}[Uniform lower tail for the complete ratio]
\label{I:lem:acceptance}
There are universal constants $C<\infty$ and $\epsilon<\tfrac1{10}$
such that, uniformly at every deterministic $x\in\intK$, the walk of
Definition~\ref{I:def:walk} satisfies
\begin{equation}
\Pr_{Y\sim q_{x}}[Y\in K\ \text{ and }\
\log R_{x}(Y)\ge-C]\ \ge\ 1-\epsilon .
\label{I:eq:acceptance}
\end{equation}
The constants are independent of $x$, $n$, $d$, $L$, $R$, and the
orientation of the regularizer.
\end{lemma}

\begin{proof}
Since the ratio $R_{x}$ is invariant under the affine substitution, we work
in the normalized coordinates. We intersect three events, namely the event
$\{h\in\Ktube\}$ that the step lies in the good set, the concentration event
of Lemma~\ref{I:lem:Z-bound}, and the Gaussian norm event
$\{\norm{\xi}\le\sqrt{d}+\sqrt{2\log(1/\eone)}\}$. By
Lemma~\ref{I:lem:tube-probability}, the first event has probability
$\ge1-\delta_{0}$, and by Lemma~\ref{I:lem:Z-bound}, the concentration event
has probability $\ge1-2\eone$. Since $\E[\norm{\xi}]\le\sqrt d$,
Lemma~\ref{I:lem:gaussian-concentration} implies that the norm event has
probability $\ge1-\eone$. Because
$\delta_{0}=\eone^{2}d^{-3}\le\eone^{2}$, the failure probabilities sum to at
most $\epsilon:=\eone^{2}+3\eone<\tfrac1{10}$ for a small universal
$\eone$.

\emph{Feasibility.} On the norm event, we have
$\norm{Y-x}_{H(x)}=\eta\norm{h}\le\eta\norm{\xi}\le\rho M_{0}<1$.
By the base-point identity
$(A_{x})_{i}H(x)^{-1}(A_{x})_{i}^{\top}=w_{i}^{2/p}\le1$ (see
Eq.~\eqref{I:eq:dikin-containment} below) and Cauchy--Schwarz, the open unit
$H(x)$-ellipsoid lies in $\intK$, and hence $Y$ and the whole segment
$[x,Y]$ are interior.

\emph{Proposal ratio.} By Lemma~\ref{I:lem:ledger}, we have
\[
\Lambda_{q}(x,Y)
=-Z-\tfrac12(c-\wb)^{\top}(b-2a)
+\tfrac12\{\log\det(I+E)-\tr[E]\}.
\]
On the intersected event, we apply Lemma~\ref{I:lem:Z-bound},
Lemma~\ref{I:lem:remainder-quad}, and Lemma~\ref{I:lem:remainder-det} with
the choice Eq.~\eqref{I:eq:rho-choice} to bound each of the three terms by a
universal constant, which gives $\Lambda_{q}(x,Y)\ge-C/2$.

\emph{Target ratio.} If $L=0$, then by Lipschitzness $f$ is constant on the
convex body $K$, and hence $f(Y)-f(x)=0$. If $L>0$, then
$G(x)\succeq L^{2}I$ and the proposal has covariance
$\tfrac{r^{2}}{d}G(x)^{-1}$. In this case, on the norm event, we use
Lipschitzness, the metric lower bound, and the Gaussian proposal
representation to obtain
\begin{equation*}
f(Y)-f(x)\le L\norm{Y-x}_{2}
\le L\cdot\frac{1}{L}\,\norm{Y-x}_{G(x)}
=\frac{r}{\sqrt{d}}\,\norm{\xi}=O(r) .
\end{equation*}

Putting things together, we obtain $\log R_{x}(Y)=-(f(Y)-f(x))+\Lambda_{q}(x,Y)
\ge-O(r)-C/2\ge-C$ for $r\le1$, after enlarging the universal constant $C$
if necessary.
\end{proof}

\begin{remark}
Lemma~\ref{I:lem:acceptance} holds at \emph{every} interior
point, with no dependence on the distance to $\partial K$. Moreover, the
analysis is affine-invariant, and Sections~\ref{I:sec:cancellation},
\ref{I:sec:stability}, and \ref{I:sec:hessian} use no property of the
regularizer other than $\creg\succeq0$. Consequently, we do not need any
boundary-layer or $s$-conductance argument in Section~\ref{I:sec:mixing}.
\end{remark}
\subsection{From acceptance to mixing}\label{I:sec:mixing}

In this section, we return to the original coordinates and convert
Lemma~\ref{I:lem:acceptance} into a conductance bound in five explicit
steps, namely the local Gaussian overlap, a kernel-overlap lemma,
symmetry, cross-ratio isoperimetry, and Cheeger's inequality. All five
steps work directly with the \emph{complete} kernel, so that no
forward/reverse asymmetry ever reappears.

\begin{lemma}[SSC of the regularized metric]\label{I:lem:ssc-G}
$G=\kappa H+L^{2}I$ with $\kappa\ge2(1+p^{2})$ satisfies, for all
$x\in\intK$ and $v\in\R^{d}$,
\begin{equation*}
\fnorm{G(x)^{-1/2}\,DG(x)[v]\,G(x)^{-1/2}}\le2\norm{v}_{G(x)} .
\end{equation*}
\end{lemma}

\begin{proof}
Let $H_{\kappa}:=\kappa H$, which is SSC by Lemma~\ref{I:lem:ssc-ls}. We
define $X=H_{\kappa}^{-1/2}DH_{\kappa}[v]H_{\kappa}^{-1/2}$, which satisfies
$\fnorm{X}\le2\norm{v}_{H_{\kappa}}$. The matrix
$C=H_{\kappa}^{1/2}G^{-1/2}$ satisfies $\opnorm{C}\le1$ because
$H_{\kappa}\preceq G$. Moreover, $DG[v]=DH_{\kappa}[v]$. Thus, by
the definitions of $C$ and $X$, the SSC of $H_\kappa$, and
$H_\kappa\preceq G$, we obtain
\[
\fnorm{G^{-1/2}DG[v]G^{-1/2}}
=\fnorm{C^{\top}XC}\le\fnorm{X}
\le2\norm{v}_{H_{\kappa}}\le2\norm{v}_{G} .
\]
\end{proof}

Since the Frobenius norm dominates the spectral norm, $G$ satisfies
the hypothesis of Lemma~\ref{I:lem:distortion}. By that lemma, if
\begin{equation}
\Delta:=\norm{x-z}_{G(x)}\le\frac{c_{\mathrm{loc}}\,r}{\sqrt{d}}
\qquad(x,z\in\intK),
\label{I:eq:close-pair}
\end{equation}
then, for $\Delta\le\tfrac12$ (ensured by $r\le1$), we have
\begin{equation}
(1-\Delta)^{2}G(x)\preceq G(z)\preceq(1-\Delta)^{-2}G(x) .
\label{I:eq:distortion}
\end{equation}

\begin{lemma}[Pinsker's inequality]\label{I:lem:pinsker}\label{II:lem:pinsker}
For probability densities $p,q$,
\begin{equation*}
\dtv(p,q)^2\le\tfrac12D_{\mathrm{KL}}(p\Vert q).
\end{equation*}
\end{lemma}

\begin{proof}
If the relative entropy is infinite, the claim is immediate. Otherwise, we
define $A=\{p\ge q\}$,
$a=\int_Ap$, and $b=\int_Aq$, so that
$a-b=\dtv(p,q)$. Applying Jensen's inequality on $A$ and $A^c$ (or,
equivalently, the two-set log-sum inequality), we obtain
\[
D_{\mathrm{KL}}(p\Vert q)
\ge a\log\frac ab+(1-a)\log\frac{1-a}{1-b}=:g_b(a).
\]

For fixed $b\in(0,1)$, we have $g_b(b)=g_b'(b)=0$ and
$g_b''(s)=1/[s(1-s)]\ge4$, and hence
$g_b(a)\ge2(a-b)^2$. The endpoint cases follow by continuity.
\end{proof}

\begin{lemma}[Proposal overlap]\label{I:lem:proposal-overlap}
There is a universal $c_{\mathrm{loc}}>0$ such that
Eq.~\eqref{I:eq:close-pair} implies
\[
 \dtv(q_x,q_z)\le\delta_q
\]
for a fixed constant $\delta_q<1-2\epsilon$, where $\epsilon$ is from
Lemma~\ref{I:lem:acceptance}.
\end{lemma}

\begin{proof}
We consider the Gaussian proposals
$q_{x}=N(x,\tfrac{r^{2}}{d}G(x)^{-1})$ and
$q_{z}=N(z,\tfrac{r^{2}}{d}G(z)^{-1})$. Taking the expectation of the log
density ratio, we obtain the general identity
\[
D_{\mathrm{KL}}(N(m_0,\Sigma_0)\Vert N(m_1,\Sigma_1))
=\frac12[\tr[\Sigma_1^{-1}\Sigma_0]
+\norm{m_1-m_0}_{\Sigma_1^{-1}}^2-d
+\log\frac{\det\Sigma_1}{\det\Sigma_0}].
\]
Substituting the two proposals into this identity gives
\begin{equation*}
D_{\mathrm{KL}}(q_{x}\Vert q_{z})
=\frac12[\sum_{j=1}^{d}(\lambda_{j}-1-\log\lambda_{j})
+\frac{d}{r^{2}}\norm{x-z}_{G(z)}^{2}],
\end{equation*}
where $\lambda_{j}$ are the eigenvalues of
$G(x)^{-1/2}G(z)G(x)^{-1/2}$. By Eq.~\eqref{I:eq:distortion}, we have
$\lambda_{j}\in[(1-\Delta)^{2},(1-\Delta)^{-2}]$. For
$\Delta\le1/2$, this interval implies
$|\lambda_j-1|\le C\Delta$ and $\lambda_j\ge1/4$. Moreover, we have
\[
0\le\lambda-1-\log\lambda
=\int_1^\lambda\frac{t-1}{t}\,dt
\le C(\lambda-1)^2.
\]

Consequently, we have
$\sum_{j}(\lambda_{j}-1-\log\lambda_{j})\le Cd\Delta^{2}$ and
$\norm{x-z}_{G(z)}^{2}\le4\Delta^{2}$. Using $r\le1$ and
Eq.~\eqref{I:eq:close-pair}, we obtain
\[
D_{\mathrm{KL}}(q_{x}\Vert q_{z})
\le C(d\Delta^{2}+\frac{d\Delta^{2}}{r^{2}})
\le C'c_{\mathrm{loc}}^{2} .
\]

Choosing $c_{\mathrm{loc}}$ sufficiently small and applying
Lemma~\ref{I:lem:pinsker} gives the claimed bound.
\end{proof}

In the next lemma, we use the complete-ratio estimate of
Lemma~\ref{I:lem:acceptance} directly, so the determinant and quadratic
terms need no separate treatment at the coupling stage either.

\begin{lemma}[Kernel overlap]\label{I:lem:kernel-overlap}
Suppose Eq.~\eqref{I:eq:acceptance} holds at both $x$ and $z$, and
$\dtv(q_{x},q_{z})\le\delta_{q}$. Then the accepted (moving) parts of
the two transition kernels have common mass at least
\begin{equation}
\gamma_{0}:=e^{-C}\,(1-\delta_{q}-2\epsilon)>0,
\label{I:eq:gamma0}
\end{equation}
and consequently, for the $\tfrac12$-lazy kernels $P_{x},P_{z}$,
\begin{equation}
\dtv(P_{x},P_{z})\le1-\frac{\gamma_{0}}{2} .
\label{I:eq:kernel-overlap}
\end{equation}
\end{lemma}

\begin{proof}
Let $E_{x}:=\{u\in K:R_{x}(u)\ge e^{-C}\}$, and define $E_{z}$ in the
same way. Since the moving part of the non-lazy kernel at $x$ has density
$u\mapsto q_{x}(u)\,(1\wedge R_{x}(u))$ on $K$, the common
mass of the two moving parts is at least
\[
\int_{E_{x}\cap E_{z}}e^{-C}\,
\min\{q_{x}(u),q_{z}(u)\}\,du .
\]
The measure with density $\min(q_{x},q_{z})$ has total mass at least
$1-\delta_{q}$. Moreover, by Eq.~\eqref{I:eq:acceptance}, the mass of this
measure on $E_{x}^{c}$ is at most $q_{x}(E_{x}^{c})\le\epsilon$, and the
same bound holds on $E_{z}^{c}$. These bounds, together with the preceding
display, prove Eq.~\eqref{I:eq:gamma0}. Since laziness
retains half of the common moving mass,
Eq.~\eqref{I:eq:kernel-overlap} follows.
\end{proof}

Combining Lemmas~\ref{I:lem:proposal-overlap} and \ref{I:lem:kernel-overlap},
we obtain that the lazy kernels have overlap at least $\gamma_{0}/2$
whenever
\begin{equation}
\norm{x-z}_{G(x)}\le\frac{c_{\mathrm{loc}}\,r}{\sqrt{d}} .
\label{I:eq:coupling-radius}
\end{equation}
Next, we bound the symmetry parameter of the metric.

\begin{definition}[Symmetry;~\cite{llv20}]
For a metric $M$, we write
$\mathcal{D}_{M}(x,t):=\{z:\norm{z-x}_{M(x)}\le t\}$. The metric is
\emph{$\bar\nu$-symmetric} if
\[
\mathcal{D}_{M}(x,1)\subseteq K\cap(2x-K)
\subseteq\mathcal{D}_{M}(x,\sqrt{\bar\nu})
\qquad(x\in\intK).
\]
\end{definition}

\begin{lemma}[Symmetry parameter]\label{I:lem:symmetry}
The scaled LS metric $H_{\kappa}=\kappa H$ is $\nubar$-symmetric with
\begin{equation}
\nubar:=\kappa\,n^{2/p}d^{1-2/p}=\Ot(d),
\label{I:eq:nubar}
\end{equation}
and $G=\kappa H+L^{2}I$ is $\nu_{G}$-symmetric with
\begin{equation*}
\nu_{G}\le\nubar+4L^{2}R^{2}.
\end{equation*}
\end{lemma}

\begin{proof}
We first prove the containment of the unit ellipsoid. By the Lewis fixed
point Eq.~\eqref{I:eq:lewis-fixedpoint}, the $i$-th row of $A_{x}$ satisfies
\begin{equation}
(A_{x})_{i,\cdot}\,H(x)^{-1}\,(A_{x})_{i,\cdot}^{\top}
=w_{x,i}^{2/p}\le1 ,
\label{I:eq:dikin-containment}
\end{equation}
which in normalized coordinates reads
$\norm{w_{i}^{-\beta}u_{i}}^{2}=w_{i}^{1-2\beta}=w_{i}^{2/p}$.
By the Cauchy--Schwarz inequality, $\norm{z-x}_{H(x)}\le1$ implies
$|(A_{x}(z-x))_{i}|\le1$ for all $i$, and hence both $z$ and $2x-z$
have nonnegative slacks $s_{x}\circ(\mathbf{1}\pm A_{x}(z-x))$. Thus,
using $\kappa\ge1$, we obtain
$\mathcal{D}_{H_{\kappa}}(x,1)\subseteq\mathcal{D}_{H}(x,1)\subseteq
K\cap(2x-K)$.

For the containment in the outer ellipsoid, let
$z\in K\cap(2x-K)$. Then
every normalized slack displacement satisfies
$|(A_{x}(z-x))_{i}|\le1$. Using H\"older's inequality in the form
$\sum_{i}1\cdot w_{i}^{1-2/p}\le n^{2/p}(\sum_{i}w_{i})^{1-2/p}$,
we have
\begin{equation*}
\norm{z-x}_{H(x)}^{2}
=\sum_{i}w_{x,i}^{1-2/p}(A_{x}(z-x))_{i}^{2}
\le\sum_{i}w_{x,i}^{1-2/p}
\le n^{2/p}\,d^{1-2/p} .
\end{equation*}
Scaling by $\kappa$, we obtain Eq.~\eqref{I:eq:nubar}. Moreover, since
$n^{2/p}\le e$ and $\kappa=\Ot(1)$, we have $\nubar=\Ot(d)$.

It remains to treat the regularized metric. Since
$G\succeq H_{\kappa}$, we have $\mathcal{D}_{G}(x,1)\subseteq
\mathcal{D}_{H_{\kappa}}(x,1)\subseteq K\cap(2x-K)$. For
$z\in K\cap(2x-K)$, Eq.~\eqref{I:eq:nubar} and the diameter bound
$\operatorname{diam}(K)\le2R$ give
\[
\norm{z-x}_{G(x)}^{2}
=\norm{z-x}_{H_{\kappa}(x)}^{2}+L^{2}\norm{z-x}_{2}^{2}
\le\nubar+4L^{2}R^{2}.
\]
This proves the bound for the regularized metric.
\end{proof}

The preceding computation is self-contained, and it
agrees with \cite[Lemma~4.3 of arXiv v2]{llv20}. We now
compare the local norm with the cross ratio.

\begin{definition}[Chord cross ratio]\label{I:def:crossratio}
For distinct $x,z\in\intK$, let their line meet $\partial K$ in the
order $p_{\partial},x,z,q_{\partial}$ and put
$\alpha=\norm{x-p_{\partial}}_{2}$, $\beta_{c}=\norm{z-x}_{2}$,
$\gamma=\norm{q_{\partial}-z}_{2}$. The cross ratio is
\[
\sigma(x,z):=\frac{\beta_{c}\,(\alpha+\beta_{c}+\gamma)}
{\alpha\,\gamma} .
\]
\end{definition}

\begin{lemma}[Cross ratio dominates the local norm]\label{I:lem:crossratio}
For all distinct $x,z\in\intK$,
\begin{equation*}
\sigma(x,z)\ \ge\ \frac{\norm{x-z}_{G(x)}}{\sqrt{\nu_{G}}} .
\end{equation*}
\end{lemma}

\begin{proof}
Let $e=(z-x)/\norm{z-x}_{2}$. The radius of the symmetric body
$K\cap(2x-K)$ at $x$ in direction $\pm e$ is
$t=\min(\alpha,\beta_{c}+\gamma)$, and hence $x+te$ lies in the closure of
$K\cap(2x-K)$. Thus, by the $\nu_{G}$-symmetry, we have
$t\,\norm{e}_{G(x)}\le\sqrt{\nu_{G}}$. On the other hand, by checking the
two cases $t=\alpha$ and $t=\beta_c+\gamma$, we obtain
\[
\sigma(x,z)
=\frac{\beta_{c}(\alpha+\beta_{c}+\gamma)}{\alpha\gamma}
\ \ge\ \frac{\beta_{c}}{t}.
\]

Combining the two bounds, we have
$\sigma(x,z)\ge\beta_{c}\norm{e}_{G(x)}/\sqrt{\nu_{G}}
=\norm{z-x}_{G(x)}/\sqrt{\nu_{G}}$.
\end{proof}

We are now ready to bound the conductance.

\begin{lemma}[Conductance]\label{I:lem:conductance}
The walk of Definition~\ref{I:def:walk} has conductance
\begin{equation}
\Phi:=\inf_{0<\pi(A)\le1/2}
\frac{\mathcal{Q}(A,A^{c})}{\pi(A)}
\ \ge\ \frac{c_{\mathrm{iso}}\gamma_{0}}{128}\,\lsep
\ \ge\ \frac{c\,r}{\sqrt{d\,(\nubar+L^{2}R^{2})}},
\qquad
\lsep:=\frac{c_{\mathrm{loc}}\,r}{\sqrt{d\,\nu_{G}}}\wedge1,
\label{I:eq:conductance}
\end{equation}
where $\mathcal{Q}(A,A^{c}):=\int_{A}P_{x}(A^{c})\,\pi(dx)$ is the
stationary flow.
\end{lemma}

\begin{proof}
We fix a measurable set $A$ with $0<\pi(A)\le\tfrac12$, put
$b_{0}:=\gamma_{0}/8$, and define the deep sets
\[
A_{1}=\{x\in A:P_{x}(A^{c})\le b_{0}\},
\qquad
A_{2}=\{z\in A^{c}:P_{z}(A)\le b_{0}\}.
\]
Since $K$ is full-dimensional, its boundary has Lebesgue measure zero, and
hence $\pi$-measure zero. Thus, we may intersect $A_1,A_2$ with $\intK$ and
put $\partial K$ into $A_3$ below without changing the stationary masses
and flows. In particular, every pair to which we apply the local and
cross-ratio lemmas lies in $\intK$.

For $x\in A_{1}$ and $z\in A_{2}$, we have
$\dtv(P_{x},P_{z})\ge P_{z}(A^{c})-P_{x}(A^{c})
\ge1-2b_{0}>1-\gamma_{0}/2$. By
Lemma~\ref{I:lem:kernel-overlap}, such a pair therefore cannot satisfy
Eq.~\eqref{I:eq:coupling-radius}, which means that
$\norm{x-z}_{G(x)}>c_{\mathrm{loc}}r/\sqrt{d}$. Applying
Lemmas~\ref{I:lem:crossratio} and \ref{I:lem:symmetry}, we then obtain the
uniform separation
\begin{equation}
\sigma(x,z)\ \ge\ \lsep
\qquad\text{for all }x\in A_{1},\,z\in A_{2}.
\label{I:eq:separation}
\end{equation}
Let $A_{3}:=K\setminus(A_{1}\cup A_{2})$,
$B_{1}:=A\setminus A_{1}$, and $B_{2}:=A^{c}\setminus A_{2}$, so that
$A_{3}=B_{1}\,\dot\cup\,B_{2}$. By the definitions of $B_1,B_2$ and
reversibility, we have
\begin{equation}
\mathcal{Q}(A,A^{c})\ge b_{0}\,\pi(B_{1}),
\qquad
\mathcal{Q}(A,A^{c})=\mathcal{Q}(A^{c},A)\ge b_{0}\,\pi(B_{2}).
\label{I:eq:flow-lb}
\end{equation}

\emph{Case 1: $\pi(A_{1})<\pi(A)/2$ or $\pi(A_{2})<\pi(A^{c})/2$.}
Then $\pi(B_{1})\ge\pi(A)/2$ or $\pi(B_{2})\ge\pi(A^{c})/2\ge\pi(A)/2$.
In either case, Eq.~\eqref{I:eq:flow-lb} gives
$\mathcal{Q}(A,A^{c})/\pi(A)\ge b_{0}/2$.

\emph{Case 2: $\pi(A_{1})\ge\pi(A)/2$ and
$\pi(A_{2})\ge\pi(A^{c})/2\ge\tfrac14$.}
Since $\pi$ is log-concave and Eq.~\eqref{I:eq:separation} provides the
separation $\sigma_{0}=\lsep\le1$, we can apply
Lemma~\ref{I:lem:isoperimetry-import} to the partition
$(A_{1},A_{2},A_{3})$ and obtain
\[
\pi(A_{3})\ \ge\ c_{\mathrm{iso}}\,\lsep\,
\pi(A_{1})\,\pi(A_{2})
\ \ge\ \frac{c_{\mathrm{iso}}\,\lsep}{8}\,\pi(A).
\]
On the other hand, by adding the two bounds in Eq.~\eqref{I:eq:flow-lb} and
using $A_{3}=B_{1}\dot\cup B_{2}$, we obtain
$2\,\mathcal{Q}(A,A^{c})\ge b_{0}\,\pi(A_{3})$. Together with the
preceding display, this gives
$\mathcal{Q}(A,A^{c})/\pi(A)\ge b_{0}c_{\mathrm{iso}}\lsep/16$.

In both cases, we have
$\Phi\ge\tfrac{c_{\mathrm{iso}}\gamma_{0}}{128}\lsep$. Substituting
$\nu_{G}\le\nubar+4L^{2}R^{2}$ from Lemma~\ref{I:lem:symmetry} and absorbing
constants, we obtain Eq.~\eqref{I:eq:conductance}.
\end{proof}

\renewcommand{\nubar}{\bar{\nu}}
\renewcommand{\Hld}{H}

\section{Proof of Theorem~\ref{thm:intro-spectral}}\label{II:sec:spectral}

\subsection{Main theorem}

\begin{theorem}[Spectrahedra; formal version of Theorem~\ref{thm:intro-spectral}]\label{II:thm:main-formal}
Let $K:=\{x\in\R^{d}:A_{0}+\sum_{i=1}^{d}x_{i}A_{i}\succeq0\}$ be bounded,
full-dimensional, strictly feasible after minimal-face reduction, and
contained in $B_{2}(0,R)$,
with $n\times n$ matrix blocks, let $f$ be convex and $L$-Lipschitz on $K$, and
let $\pi\propto e^{-f}\ind{K}$. For the log-det metric, let $\psistar$ be the maximum of one and the
supremum, over $x\in\intK$, of the spectral norm of the whitened matrix
leverage operator at $x$. Fix a sufficiently small universal $r\in(0,1]$
and choose $\kappa$ as in Eq.~\eqref{II:eq:rho-choice}, so
$\kappa=\Ot(\psistar)$. Then, for $w\ge1$ and $0<\delta<1$, the
$\tfrac12$-lazy Metropolis walk with the exact metric reaches total-variation distance
$\delta$ from a $w$-warm start in
\[
  \Ot((\psistar nd+dL^{2}R^{2})\log(w/\delta))
\]
steps.
\end{theorem}

\begin{proof}
Since $f_{0}=d\mu_{0}/d\pi\le w$, we have
$\norm{f_{0}-1}_{L^{2}(\pi)}^{2}=\int f_{0}^{2}d\pi-1\le w-1$. Moreover, since
the chain is reversible and $\tfrac12$-lazy, its spectrum is nonnegative,
and the Cheeger inequality in Lemma~\ref{II:lem:cheeger} gives
$\gap(P)\ge c_{\mathrm{Ch}}\Phi^{2}$. Therefore, we obtain
\[
\dtv(\mu_{0}P^{t},\pi)
\le\tfrac12\norm{P^{t}(f_{0}-1)}_{L^{2}(\pi)}
\le\tfrac12(1-\gap P)^{t}\sqrt{w-1}
\le\tfrac12\,e^{-c_{\mathrm{Ch}}t\Phi^{2}}\sqrt w ,
\]
where the first step follows from the Cauchy--Schwarz inequality, the second step follows from spectral contraction and the warm-start bound on the initial $L^2(\pi)$ norm, and the
last step follows from $1-u\le e^{-u}$,
$\gap(P)\ge c_{\mathrm{Ch}}\Phi^{2}$, and $\sqrt{w-1}\le\sqrt w$.

Solving for the number of steps needed to reach error $\delta$ and using
Lemmas~\ref{II:lem:conductance} and~\ref{II:lem:symmetry}, we obtain
\[
t_{\mathrm{mix}}(\delta)\ \le\
C\,\frac{d\,\nu_{G}}{r^{2}}\log(w/\delta)
\ =\ C\,\frac{d(\kappa n+4L^{2}R^{2})}{r^{2}}
\log(w/\delta),
\]
where the equality follows from $\nu_G=\kappa n+4L^2R^2$. Finally, since
$r$ is universal and $\kappa=\Ot(\psistar)$ by
Eq.~\eqref{II:eq:rho-choice}, this bound is
$\Ot((\psistar nd+dL^{2}R^{2})\log(w/\delta))$.
\end{proof}

\subsection{Normalization, endpoint formulas, and the proposal ratio}
\label{II:sec:ledger}

All the computations in this section are exact, with no estimates. We fix
a deterministic current point $x\in\intK$ and replace each $A_{i}$
($i\ge1$) by its slack congruence
$S(x)^{-1/2}A_{i}S(x)^{-1/2}$. This gives an equivalent representation of
the same constraint set, since
$S(x+u)=S(x)^{1/2}(I+\sum_{i}u_{i}\,
S(x)^{-1/2}A_{i}S(x)^{-1/2})S(x)^{1/2}$. We then apply the linear
change of variables in $\R^{d}$ that orthonormalizes these congruenced
directions in the Frobenius inner product. In the new coordinates,
which are used throughout
Sections~\ref{II:sec:ledger}, \ref{II:sec:tube}, \ref{II:sec:stein}, and \ref{II:sec:fluct},
we have
\begin{equation}
S(x)=I,\qquad
\Ab(u):=\sum_{a=1}^{d}u_{a}\Ab_{a},\qquad
\fnorm{\Ab(u)}=\norm{u},\qquad
\Hld(x)=I_{d},
\label{II:eq:normalized}
\end{equation}
where $\{\Ab_{a}\}$ is the orthonormal frame of
Definition~\ref{II:def:leverage}, and the third identity follows from frame
orthonormality. We note that the Metropolis ratio $R_{x}$ is invariant under
this substitution. Indeed, the congruence is an identity on $K$. Under
the linear change of variables, $\Hld$ and $G$ transform as metrics, and both
Gaussian densities in the ratio acquire the same Jacobian.

Moreover, after the same change of variables, the regularizer becomes a
constant positive semidefinite matrix, and we factor
\begin{equation}
G(x)=\kappa(I+\creg),\qquad \creg\succeq0,\qquad
Q:=(I+\creg)^{-1},\qquad 0\preceq Q\preceq I .
\label{II:eq:Q-def}
\end{equation}
Since we use only $\creg\succeq0$ below, all our estimates are uniform in
$L$ and in the orientation of the regularizer. In these coordinates, a
proposal has the representation
\begin{equation}
Y=x+\eta h,\qquad h\sim N(0,Q),\qquad
\eta=\frac{r}{\sqrt{\kappa d}}=\frac{\rho}{\sqrt d},\qquad
\rho:=\frac{r}{\sqrt\kappa}.
\label{II:eq:step}
\end{equation}

For a step $g\in\R^{d}$, let
\begin{equation*}
\Delta=\Delta(g):=\eta\,\Ab(g),\qquad
\Rres=\Rres(g):=(I+\Delta)^{-1}\ \ \text{(defined when
$\opnorm{\Delta}<1$)} .
\end{equation*}
Note that the matrices $\Rres$ and $\Delta$ commute. Using
$S(x+\eta g)=I+\Delta(g)$ in the normalized coordinates, we obtain the
following lemma about the endpoint of the step.

\begin{lemma}[Endpoint identities]\label{II:lem:endpoint}
Let $g$ satisfy $\opnorm{\Delta(g)}<1$. Then the entire segment
$[x,\,x+\eta g]$ lies in $\intK$, and with $y=x+\eta g$,
\begin{equation}
\Hld(y)_{ab}=\tr[\Rres\Ab_{a}\Rres\Ab_{b}],
\qquad
\Eb:=\Hld(y)-I,\qquad
\Eb_{ab}
=-\tr[\Delta \Rres\Ab_{a}\Rres\Ab_{b}]
-\tr[\Delta \Rres\Ab_{b}\Ab_{a}].
\label{II:eq:endpoint}
\end{equation}
Moreover, as functions of $h$ (with
$\partial_{c}:=\partial/\partial h_{c}$ and all matrices evaluated at
the same $h$),
\begin{align}
\partial_{c}\Rres&=-\eta\,\Rres\Ab_{c}\Rres,
\label{II:eq:dR}\\
\partial_{c}\Eb_{ab}
&=-\eta[\tr[\Rres\Ab_{c}\Rres\Ab_{a}\Rres\Ab_{b}]
+\tr[\Rres\Ab_{a}\Rres\Ab_{c}\Rres\Ab_{b}]],
\label{II:eq:dE}\\
\partial_{e}\partial_{c}\Eb_{ab}
&=\eta^{2}[
F_{ecab}+F_{ceab}+F_{caeb}
+F_{eacb}+F_{aecb}+F_{aceb}
], \label{II:eq:ddE}\\
F_{wxyz}&:= \tr[\Rres\Ab_{w}\Rres\Ab_{x}\Rres\Ab_{y}\Rres\Ab_{z}],
\nonumber
\end{align}
where each of the four letters $a,c,b,e$ appears exactly once in
every word, and $F$ is invariant under cyclic shifts of its four
letters.
\end{lemma}

\begin{proof}
Feasibility follows since, when $\opnorm{\Delta}<1$, we have
$S(x+t\eta g)=I+t\Delta\succ0$ for $t\in[0,1]$. The formula for the metric
is Eq.~\eqref{II:eq:metric-def} after the congruence by
$S(y)^{-1/2}$, in the form
$\Hld(y)_{ab}=\tr[S(y)^{-1}\Ab_{a}S(y)^{-1}\Ab_{b}]$ with
$S(y)^{-1}=\Rres$. For the increment, we insert
$\Rres-I=-\Delta \Rres$ twice in the identity
$\tr[\Rres\Ab_{a}\Rres\Ab_{b}]-\tr[\Ab_{a}\Ab_{b}]
=\tr[(\Rres-I)\Ab_{a}\Rres\Ab_{b}]+\tr[\Ab_{a}(\Rres-I)\Ab_{b}]$, and we use
the cyclicity of the trace on the second term.

For the derivatives, Eq.~\eqref{II:eq:dR} follows from
$\partial_{c}(I+\Delta)^{-1}=-(I+\Delta)^{-1}
(\partial_{c}\Delta)(I+\Delta)^{-1}$ and
$\partial_{c}\Delta=\eta\Ab_{c}$. Differentiating
$\Eb_{ab}=\tr[\Rres\Ab_{a}\Rres\Ab_{b}]-\delta_{ab}$ by the product rule
across the two $\Rres$-factors, we obtain Eq.~\eqref{II:eq:dE}. Similarly,
Eq.~\eqref{II:eq:ddE} follows by differentiating each three-letter trace
in Eq.~\eqref{II:eq:dE} across its three $\Rres$-factors, that is, by
inserting $-\eta\Ab_{e}$, with its two flanking $\Rres$'s, at each of the
three $\Rres$-slots of each of the two words. This gives six four-letter
words in total, each carrying the sign $(-\eta)(-\eta)=\eta^{2}$. Finally,
every resulting word contains each of the differentiated directions $\Ab_c$, $\Ab_e$
and the original directions $\Ab_a$, $\Ab_b$ exactly once, and by
cyclic invariance, we can start the trace at any one of these four letters.
This proves the last claim.
\end{proof}

\begin{lemma}[Proposal-ratio identity]\label{II:lem:ledger}
Let $g$ satisfy $\opnorm{\Delta(g)}<1$ and $y=x+\eta g$. With
$E:=Q^{1/2}\Eb Q^{1/2}$, the complete Gaussian proposal log-ratio
satisfies, exactly,
\begin{equation}
\begin{aligned}
\log\frac{q_{y}(x)}{q_{x}(y)}
&=\frac12\{\log\det(I+E)-\tr[E]\}
-\frac12\,Z(g),\\
Z(g)&:=g^{\top}\Eb g-\tr[Q\Eb]\\
&=\sum_{a,b}(g_{a}g_{b}-Q_{ab})\Eb_{ab}.
\end{aligned}
\label{II:eq:ledger}
\end{equation}
\end{lemma}

\begin{proof}
Since $G(z)=\kappa\Hld(z)+\kappa\creg$ in the normalized coordinates, we
have $G(y)-G(x)=\kappa\Eb$. Plugging in
$q_{z}(\cdot)\propto\det G(z)^{1/2}
\exp(-\tfrac{d}{2r^{2}}\norm{\cdot-z}_{G(z)}^{2})$, we obtain
\[
\log\frac{q_{y}(x)}{q_{x}(y)}
=\frac12\log\frac{\det G(y)}{\det G(x)}
-\frac{d}{2r^{2}}\,(y-x)^{\top}(G(y)-G(x))(y-x).
\]
For the determinant, Eq.~\eqref{II:eq:Q-def} gives
$\det G(y)/\det G(x)=\det(I+\creg+\Eb)/\det(I+\creg)
=\det(I+Q^{1/2}\Eb Q^{1/2})=\det(I+E)$. Since $y-x=\eta g$ and
$d\eta^{2}\kappa/r^{2}=1$ by Eq.~\eqref{II:eq:step}, the quadratic part
equals $g^{\top}\Eb g$. Finally, since $\tr[E]=\tr[Q\Eb]$, adding and
subtracting $\tfrac12\tr[E]$ gives Eq.~\eqref{II:eq:ledger}.
\end{proof}

We remark that the uncentered term $\tr[Q\Eb]$, which is both the
linear part of the determinant and the mean of the random quadratic form,
cancels identically in the proposal-ratio identity. It remains to control the
determinant remainder and the centered fluctuation $Z$, which we do in Section~\ref{II:sec:tube} and in Sections~\ref{II:sec:stein} and \ref{II:sec:fluct}, respectively.
\subsection{The good set, stability, and the determinant remainder}
\label{II:sec:tube}

Throughout, we fix $x$ and work in the normalized coordinates of
Section~\ref{II:sec:ledger}. Let $\eone$ be a small universal failure parameter
and set
\begin{equation*}
\delta_{0}:=\eone^{2}d^{-3},\qquad
R_{0}:=\sqrt d+\sqrt{2\log(2/\delta_{0})},\qquad
\tau_{0}:=\sqrt{2\,\psi(x)\,\log(8n/\delta_{0})} .
\end{equation*}
We define the convex \emph{good set} and the enlarged set
\begin{equation*}
\Ktube:=\{g\in\R^{d}:\norm{g}\le R_{0},\
\opnorm{\Ab(g)}\le\tau_{0}\},
\qquad
\Ktubep:=\Ktube+B_{2}(0,1),
\end{equation*}
and let $R_{1}:=R_{0}+1$, $\tau_{1}:=\tau_{0}+1$, and
$M_{0}:=R_{1}/\sqrt d$. These quantities satisfy $M_{0}=O_{\eone}(1)$ and
\begin{equation}
\tau_{1}^{2}\ \le\ 16\,\psi^{+}(x)\,(1+\log(nd/\eone)),
\qquad \psi^{+}(x):=\max(1,\psi(x)) .
\label{II:eq:tau-size}
\end{equation}
Indeed, $(a+1)^2\le2a^2+2$ and, for $n,d\ge1$ and
$0<\eone\le1$,
\[
\log(8n/\delta_0)=\log(8nd^3/\eone^2)
\le3(1+\log(nd/\eone)).
\]
Substituting the definition of $\tau_0$, we obtain
$\tau_1^2\le12\psi(x)(1+\log(nd/\eone))+2$, which is bounded by
the right-hand side of Eq.~\eqref{II:eq:tau-size} since $\psi^+\ge1$.

\begin{lemma}[Probability of the good set]\label{II:lem:tube-prob}
For $h\sim N(0,Q)$ with $Q\preceq I$,
$\Pr[h\in\Ktube]\ge1-\delta_{0}$.
\end{lemma}

\begin{proof}
For the norm, we write $h=Q^{1/2}\xi$ for a standard Gaussian $\xi$. Then
the map $\xi\mapsto\norm{Q^{1/2}\xi}$ is $1$-Lipschitz with mean at most
$(\E[\norm{Q^{1/2}\xi}^{2}])^{1/2}=(\tr[Q])^{1/2}\le\sqrt d$.
Thus, by Lemma~\ref{II:lem:gaussian-concentration} with
$t=\sqrt{2\log(2/\delta_{0})}$, we have
$\Pr[\norm{h}>R_{0}]\le e^{-t^{2}/2}=\delta_{0}/2$.

To apply Lemma~\ref{II:lem:matrix-gaussian} to the spectral part, we let
$\widetilde B_{i}:=\sum_{a}(Q^{1/2})_{ai}\Ab_{a}$. Then
$\Ab(h)=\sum_{i}\xi_{i}\widetilde B_{i}$ is a Gaussian matrix series whose
variance statistic is
\[
\sum_{i}\widetilde B_{i}^{2}
=\sum_{a,b}Q_{ab}\,\Ab_{a}\Ab_{b}
=V^{*}\,(Q\otimes I_{n})\,V
\ \preceq\ V^{*}V=\Psib(x),
\]
where $V:\R^{n}\to\R^{d}\otimes\R^{n}$ stacks
$u\mapsto(\Ab_{a}u)_{a}$. The inequality follows from
$Q\otimes I_n\preceq I_{dn}$ by congruence with $V$. Hence,
$\sigma^{2}\le\psi(x)$, and by Lemma~\ref{II:lem:matrix-gaussian} with
$t=\tau_{0}$, we obtain
$\Pr[\opnorm{\Ab(h)}>\tau_{0}]
\le2n\,e^{-\tau_{0}^{2}/2\psi}
=2n\cdot\delta_{0}/(8n)\le\delta_{0}/2$. We complete the proof by a
union bound with the norm estimate.
\end{proof}

\begin{lemma}[Enlarged set]\label{II:lem:collar}
Every $g\in\Ktubep$ satisfies $\norm{g}\le R_{1}$ and
$\opnorm{\Ab(g)}\le\tau_{1}$.
\end{lemma}

\begin{proof}
Decomposing $g=g_{0}+e$ with $g_{0}\in\Ktube$ and $\norm{e}\le1$, we obtain
$\opnorm{\Ab(e)}\le\fnorm{\Ab(e)}=\norm{e}\le1$ by the orthonormality of the
frame in Eq.~\eqref{II:eq:normalized}.
\end{proof}

We will choose $\rho$ (equivalently $\kappa=r^{2}/\rho^{2}$) so that,
in particular,
\begin{equation}
\eta\,\tau_{1}\ \le\ \tfrac{1}{64}
\qquad\text{(recall }\eta=\rho/\sqrt d\text{)}.
\label{II:eq:smallness}
\end{equation}
Our final choice in Eq.~\eqref{II:eq:rho-choice} satisfies this bound with
room to spare. On the enlarged set, this condition makes every endpoint quantity
uniformly comparable to its base value. Since the metric has a closed form,
this stability is immediate and requires no bootstrap argument.

\begin{lemma}[Stability on the enlarged set]\label{II:lem:stability}
Assume Eq.~\eqref{II:eq:smallness} and let $g\in\Ktubep$,
$t\in[0,1]$, $\Delta=\Delta(tg)$, $\Rres=\Rres(tg)$. Then:
\begin{enumerate}[label=(\roman*),leftmargin=2.4em,itemsep=2pt]
\item $\opnorm{\Delta}\le\eta\tau_{1}\le\tfrac1{64}$, the segment
$[x,x+\eta g]$ lies in $\intK$, and $\Rres\succ0$ with
$\tfrac{64}{65}I\preceq \Rres\preceq\tfrac{64}{63}I$, so
$\opnorm{\Rres}\le\tfrac{64}{63}$ and $\opnorm{\Rres}^{4}\le\tfrac{13}{12}$.
\item The eigenvalues of $\Hld(x+t\eta g)$ lie in
$[(\tfrac{64}{65})^{2},(\tfrac{64}{63})^{2}]
\subseteq[e^{-1/16},e^{1/16}]$.
\item $\opnorm{\Eb}\le3\,\eta\tau_{1}$ and
$\fnorm{\Eb}\le3\,\eta\tau_{1}\sqrt d=3\rho\tau_{1}$, where
$\Eb=\Hld(x+\eta g)-I$.
\end{enumerate}
\end{lemma}

\begin{proof}
(i) By Lemma~\ref{II:lem:collar}, we have
$\opnorm{\Delta}=t\eta\opnorm{\Ab(g)}\le\eta\tau_{1}$. Then feasibility and
the bounds on $\Rres=(I+\Delta)^{-1}$ follow from
$\opnorm{\Delta}\le\tfrac1{64}$.

(ii) For any $u$, Eq.~\eqref{II:eq:endpoint} and
Eq.~\eqref{II:eq:normalized} give
$u^{\top}\Hld(y_{t})u=\fnorm{\Rres^{1/2}\Ab(u)\Rres^{1/2}}^{2}
\in[\lambda_{\min}(\Rres)^{2},\lambda_{\max}(\Rres)^{2}]\cdot\norm{u}^{2}$.

(iii) By Eq.~\eqref{II:eq:endpoint},
$|u^{\top}\Eb v|\le
|\tr[\Delta \Rres\Ab(u)\Rres\Ab(v)]|+|\tr[\Delta \Rres\Ab(v)\Ab(u)]|$.
The first term satisfies
$|\tr[\Delta\cdot \Rres\Ab(u)\cdot \Rres\Ab(v)]|
\le\opnorm{\Delta}\,\snorm{\Rres\Ab(u)\Rres\Ab(v)}
\le\opnorm{\Delta}\opnorm{\Rres}^{2}\fnorm{\Ab(u)}\fnorm{\Ab(v)}$, and we
bound the second term in the same way, with one fewer $\Rres$. Using
$\fnorm{\Ab(u)}=\norm{u}$, these bounds give
$\opnorm{\Eb}\le\opnorm{\Delta}(\opnorm{\Rres}^{2}+\opnorm{\Rres})\le
3\opnorm{\Delta}$. For the Frobenius bound, we apply the frame projection inequality
$\sum_{b}\langle M,\Ab_{b}\rangle_{F}^{2}\le\fnorm{M}^{2}$ to
each of the two terms of $\Eb_{ab}$. This inequality holds for any
$M\in\Sn$ since $\{\Ab_{b}\}$ is orthonormal. For a nonsymmetric matrix, we
use its symmetric part instead. The inner products depend only on the symmetric part, whose
Frobenius norm is no larger. For the first term, we fix $a$. By cyclicity,
$\tr[\Delta \Rres\Ab_{a}\Rres\Ab_{b}]
=\langle\Ab_{b},\Delta \Rres\Ab_{a}\Rres\rangle_F$ with
$M_{a}=\Delta \Rres\Ab_{a}\Rres$ satisfying
$\fnorm{M_{a}}\le\opnorm{\Delta}\opnorm{\Rres}^{2}$, and thus
$\sum_{ab}(\tr[\Delta \Rres\Ab_{a}\Rres\Ab_{b}])^{2}
\le\opnorm{\Delta}^{2}\opnorm{\Rres}^{4}d$. For the second term, we fix $b$
and sum over $a$ with $M_{b}'=\Delta \Rres\Ab_{b}$ satisfying
$\fnorm{M_{b}'}\le\opnorm{\Delta}\opnorm{\Rres}$, which gives
$\sum_{ab}(\tr[\Delta \Rres\Ab_{b}\Ab_{a}])^{2}
\le\opnorm{\Delta}^{2}\opnorm{\Rres}^{2}d$. Hence,
$\fnorm{\Eb}^{2}\le2\opnorm{\Delta}^{2}
(\opnorm{\Rres}^{4}+\opnorm{\Rres}^{2})\,d\le9\,\opnorm{\Delta}^{2}d$.
\end{proof}

\begin{lemma}[Determinant remainder]\label{II:lem:det-remainder}
Assume Eq.~\eqref{II:eq:smallness} and $g\in\Ktubep$. With
$E=Q^{1/2}\Eb Q^{1/2}$ as in Lemma~\ref{II:lem:ledger}, every eigenvalue of
$I+E$ lies in $[e^{-1/16},e^{1/16}]$, and
\begin{equation*}
|\log\det(I+E)-\tr[E]|\ \le\ \fnorm{E}^{2}
\ \le\ \fnorm{\Eb}^{2}\ \le\ 9\,\rho^{2}\tau_{1}^{2}.
\end{equation*}
\end{lemma}

\begin{proof}
Let
\[
a:=\min\{1,\lambda_{\min}(\Hld(y))\},\qquad
b:=\max\{1,\lambda_{\max}(\Hld(y))\}.
\]
By definition, we have $aI\preceq\Hld(y)\preceq bI$, and since
$a\le1\le b$ and $\creg\succeq0$, we also have
$a\creg\preceq\creg\preceq b\creg$. Adding the corresponding inequalities
gives
\[
a(I+\creg)\preceq\creg+\Hld(y)\preceq b(I+\creg).
\]
By congruence with $Q^{1/2}=(I+\creg)^{-1/2}$, we obtain
$aI\preceq I+E\preceq bI$. We note that this argument does not require
$\creg$ and $\Hld(y)$ to commute. By Lemma~\ref{II:lem:stability}(ii),
$[a,b]\subseteq[e^{-1/16},e^{1/16}]$. On that interval,
$|\log(1+\mu)-\mu|\le\mu^{2}$ for the eigenvalues $\mu$ of $E$, since
$|\log(1+\mu)-\mu|\le\tfrac{\mu^{2}}{2(1-|\mu|)}\le\mu^{2}$
for $|\mu|\le\tfrac12$. Summing over the eigenvalues gives the first
inequality. Finally, by two-sided Frobenius submultiplicativity, we have
$\fnorm{E}=\fnorm{Q^{1/2}\Eb Q^{1/2}}
\le\opnorm{Q^{1/2}}^{2}\fnorm{\Eb}\le\fnorm{\Eb}$, and we conclude by
Lemma~\ref{II:lem:stability}(iii).
\end{proof}
\subsection{The Stein mean bound}\label{II:sec:stein}

In this section, we bound the mean of the centered fluctuation
$Z=\sum_{ab}(h_{a}h_{b}-Q_{ab})\Eb_{ab}(h)$ from
Lemma~\ref{II:lem:ledger}. Throughout, we assume
Eq.~\eqref{II:eq:smallness}, let $\chi$ be the cutoff of
Lemma~\ref{II:lem:cutoff} for the good set $\Ktube$, and write
$\psi=\psi(x)$, $\Psib=\Psib(x)$.

The main step of our argument is the following deterministic estimate,
which will bound every second-derivative term produced by the Stein
identity.

\begin{lemma}[Paired four-trace bound]\label{II:lem:four-trace}
Let $\Rres\in\Sn$ with $0\prec \Rres$ and $\opnorm{\Rres}\le\tfrac{64}{63}$, and
let $\{\Bh_{i}\}_{i=1}^{N}\subset\Sn$ satisfy
$\sum_{i}\Bh_{i}^{2}\preceq\Psib$. For each pair $(i,j)$ let
$w(i,j)=(w_{1},w_{2},w_{3},w_{4})$ be any word consisting of two
copies of $\Bh_{i}$ and two copies of $\Bh_{j}$, and set
$F_{ij}:=\tr[\Rres w_{1}\Rres w_{2}\Rres w_{3}\Rres w_{4}]$. Then
\begin{equation*}
\sum_{i,j=1}^{N}|F_{ij}|
\ \le\ \opnorm{\Rres}^{4}\,\tr[\Psib^{2}]
\ \le\ \tfrac{13}{12}\,\psi\,d .
\end{equation*}
\end{lemma}

\begin{proof}
Let $\Nt_{i}:=\Rres^{1/2}\Bh_{i}\Rres^{1/2}\in\Sn$. By inserting $\Rres=\Rres^{1/2}\Rres^{1/2}$ around each letter and using cyclicity, we obtain $F_{ij}=\tr[\Nt_{w_{1}}\Nt_{w_{2}}\Nt_{w_{3}}\Nt_{w_{4}}]$. Up to cyclic rotation,
the multiset $\{\Nt_{i},\Nt_{i},\Nt_{j},\Nt_{j}\}$ has only two arrangements:
\[
\begin{aligned}
\text{adjacent:}\quad
&\tr[\Nt_{i}^{2}\Nt_{j}^{2}]\ge0,\\
\text{alternating:}\quad
&|\tr[\Nt_{i}\Nt_{j}\Nt_{i}\Nt_{j}]|
 \le\fnorm{\Nt_{i}\Nt_{j}}^{2}
 =\tr[\Nt_{i}^{2}\Nt_{j}^{2}].
\end{aligned}
\]
The adjacent case holds since
$\tr[\Nt_i^2\Nt_j^2]=\fnorm{\Nt_i\Nt_j}^2$, and the alternating case follows
from the Frobenius Cauchy--Schwarz inequality applied to $\Nt_i\Nt_j$ and its
transpose.

Hence, $|F_{ij}|\le\tr[\Nt_{i}^{2}\Nt_{j}^{2}]$ in either case. Letting
$\widehat\Psi:=\sum_i\Nt_i^2
=\Rres^{1/2}(\sum_i\Bh_i\Rres\Bh_i)\Rres^{1/2}$, we have
\[
\sum_{ij}\tr[\Nt_{i}^{2}\Nt_{j}^{2}]
=\tr[(\textstyle\sum_{i}\Nt_{i}^{2})^{2}]
=\tr[\widehat\Psi^{2}].
\]
Lemma~\ref{II:lem:psd-facts}(ii) gives $\sum_{i}\Bh_{i}\Rres\Bh_{i}\preceq\opnorm{\Rres}\sum_{i}\Bh_{i}^{2}
\preceq\opnorm{\Rres}\,\Psib$, and hence
$0\preceq\widehat\Psi\preceq\opnorm{\Rres}\,\Rres^{1/2}\Psib \Rres^{1/2}$.
Applying Lemma~\ref{II:lem:psd-facts}(i) and then (iii), we obtain
\[
\tr[\widehat\Psi^{2}]
\le\opnorm{\Rres}^{2}\,\tr[(\Rres^{1/2}\Psib \Rres^{1/2})^{2}]
=\opnorm{\Rres}^{2}\,\fnorm{\Rres^{1/2}\Psib \Rres^{1/2}}^{2}
\le\opnorm{\Rres}^{4}\,\fnorm{\Psib}^{2}
=\opnorm{\Rres}^{4}\,\tr[\Psib^{2}],
\]
where the first inequality follows from the Loewner monotonicity of
$\tr[X^2]$ for positive-semidefinite matrices, the second from two-sided
Frobenius submultiplicativity, and the equalities from
$\fnorm{M}^{2}=\tr[M^{2}]$ for symmetric $M$.

Finally, the second inequality of the lemma follows from the numerical inequality
$(\tfrac{64}{63})^{4}\le\tfrac{13}{12}$ and the bound
$\tr[\Psib^{2}]\le\opnorm{\Psib}\tr[\Psib]=\psi d$, which holds by
Lemma~\ref{II:lem:psd-facts}(ii) and Lemma~\ref{II:lem:leverage-basic}.
\end{proof}

\begin{lemma}[Stein integrand bound]\label{II:lem:stein-integrand}
For every $h$ in the enlarged set $\Ktubep$,
\begin{equation}
|\sum_{a,b,c,e}Q_{ac}Q_{be}\,
\partial_{e}\partial_{c}\Eb_{ab}(h)|
\ \le\ 7\,\eta^{2}\,\psi\,d\ =\ 7\,\rho^{2}\,\psi .
\label{II:eq:stein-integrand}
\end{equation}
\end{lemma}

\begin{proof}
We diagonalize $Q=\sum_{i}\lambda_{i}q_{i}q_{i}^{\top}$ with
$\lambda_{i}\in[0,1]$ and $\{q_{i}\}$ orthonormal, and set
$\Bh_{i}:=\sqrt{\lambda_{i}}\,\Ab(q_{i})$. Then
\[
\sum_{i}\Bh_{i}^{2}
=\sum_{i}\lambda_{i}\Ab(q_{i})^{2}
\preceq\sum_{i}\Ab(q_{i})^{2}=\Psib,
\]
where we used $0\le\lambda_i\le1$ and frame independence
(Lemma~\ref{II:lem:leverage-basic}) applied to the orthonormal basis
$\{q_i\}$.

Now consider Eq.~\eqref{II:eq:ddE}. Since each of its six words contains the
letters $a,c$ once each and $b,e$ once each, the contraction by
$Q_{ac}Q_{be}$ replaces the pair $\{a,c\}$ by two copies of $\Bh_{i}$
(summed over $i$) and the pair $\{b,e\}$ by two copies of $\Bh_{j}$ (summed
over $j$), preserving each letter's slot. Therefore, each of the six words
contributes $\sum_{ij}F_{ij}$ for some fixed word arrangement of
$(\Bh_{i},\Bh_{i},\Bh_{j},\Bh_{j})$. By Lemma~\ref{II:lem:four-trace}
with $\opnorm{\Rres}\le\tfrac{64}{63}$ from Lemma~\ref{II:lem:stability}(i),
each contribution is at most $\tfrac{13}{12}\psi d$. Multiplying
by the prefactor $\eta^{2}$, summing over the six words and using $6\cdot\tfrac{13}{12}\le7$ gives
Eq.~\eqref{II:eq:stein-integrand}.
\end{proof}

To bound the mean, we apply Lemma~\ref{II:lem:stein} to $F=\Eb_{ab}\chi$
and sum over $(a,b)$. The lemma applies because $\Eb_{ab}$ is smooth
on an open neighborhood of the compact set $\Ktubep$ (the resolvent exists
there by Lemma~\ref{II:lem:stability}(i)), and because $\chi$ is
$W^{2,\infty}$ and supported in $\Ktubep$, so that the product extends by
zero. We obtain
\begin{equation}
\E[Z\chi]
=\underbrace{\E[\chi\cdot\textstyle\sum_{abce}Q_{ac}Q_{be}
\partial_{e}\partial_{c}\Eb_{ab}]}_{\text{main}}
+\underbrace{2\,\E[\textstyle\sum_{be}Q_{be}\partial_{e}\chi
\cdot G_{b}]
+\E[\langle Q\Eb Q,\;D^{2}\chi\rangle_{F}]}_{\text{cutoff terms}},
\label{II:eq:stein-applied}
\end{equation}
where $G_{b}:=\sum_{ac}Q_{ac}\partial_{c}\Eb_{ab}$, and we used the
symmetry of $\Eb$ and of the two cross terms.

\begin{lemma}[First-derivative contraction]\label{II:lem:G-vector}
On $\Ktubep$, with $\Bh_{i}$ as above and
$Y:=\sum_{i}\Bh_{i}\Rres\Bh_{i}$,
\begin{equation*}
\begin{aligned}
G_{b}&=-2\eta\,\langle\Ab_{b},\,\Rres Y\Rres\rangle_{F},\\
0&\preceq Y\preceq\opnorm{\Rres}\Psib,
&\tr[Y]&\le\opnorm{\Rres}\,d,\\
(\sum_{b}G_{b}^{2})^{1/2}
&\le3\,\eta\sqrt{\psi d}
=3\rho\sqrt{\psi}.
\end{aligned}
\end{equation*}
\end{lemma}

\begin{proof}
For the formula, we contract Eq.~\eqref{II:eq:dE} with $Q_{ac}$. Since the letters $a$
and $c$ occupy the two slots preceding $\Ab_{b}$ in both three-letter words,
both words equal $\sum_{i}\tr[\Rres\Bh_{i}\Rres\Bh_{i}\Rres\Ab_{b}]$ after
the substitution above. By a cyclic shift, we have
\[
\tr[\Rres\Bh_{i}\Rres\Bh_{i}\Rres\Ab_{b}]
=\langle\Ab_{b},\Rres\Bh_{i}\Rres\Bh_{i}\Rres\rangle_{F},
\]
and summing over $i$ yields the stated formula. The bounds on $Y$
follow from Lemma~\ref{II:lem:psd-facts}(ii) together with
$\sum_{i}\Bh_{i}^{2}\preceq\Psib$ and $\tr[\Psib]=d$. For the last bound,
the frame projection inequality gives
\[
\sum_{b}G_{b}^{2}
\le4\eta^{2}\fnorm{\Rres Y\Rres}^{2}
\le4\eta^{2}\opnorm{\Rres}^{4}\fnorm{Y}^{2}.
\]
Since $Y\succeq0$, by Lemma~\ref{II:lem:psd-facts} we have
\[
\fnorm{Y}^{2}=\tr[Y^{2}]
\le\opnorm{Y}\tr[Y]
\le\opnorm{\Rres}^{2}\psi d,
\]
where the last inequality follows from the preceding bounds on $Y$. Putting
things together, we obtain
$(\sum_{b}G_{b}^{2})^{1/2}\le2\eta\opnorm{\Rres}^{3}\sqrt{\psi d}
\le3\eta\sqrt{\psi d}$.
\end{proof}

\begin{lemma}[Mean bound]\label{II:lem:mean}
Assume Eq.~\eqref{II:eq:smallness}. Then
\begin{equation*}
|\E[Z\chi]|
\ \le\ 7\rho^{2}\psi
\ +\ \delta_{0}\,(72\rho\sqrt\psi+3C_{0}\rho\tau_{1}\sqrt d)
\ \le\ 7\rho^{2}\psi+\eone .
\end{equation*}
\end{lemma}

\begin{proof}
By Lemma~\ref{II:lem:stein-integrand} and $0\le\chi\le1$, we bound the main
term of Eq.~\eqref{II:eq:stein-applied} pointwise by $7\rho^{2}\psi$. The
integrands of the two cutoff terms are supported on
$\Ktubep\setminus\Ktube$, an event of probability at most $\delta_{0}$ by
Lemma~\ref{II:lem:tube-prob}, and they are bounded uniformly there.
Specifically, by Lemmas~\ref{II:lem:cutoff} and \ref{II:lem:G-vector}, we
have $|2\sum_{be}Q_{be}\partial_{e}\chi\,G_{b}|
\le2\norm{\nabla\chi}\,(\sum_{b}G_{b}^{2})^{1/2}
\le72\rho\sqrt\psi$, and by Lemma~\ref{II:lem:stability}(iii),
$|\langle Q\Eb Q,D^{2}\chi\rangle_{F}|
\le\opnorm{D^{2}\chi}\,\snorm{Q\Eb Q}
\le C_{0}\sqrt d\,\fnorm{\Eb}\le3C_{0}\rho\tau_{1}\sqrt d$. It remains to multiply
these uniform bounds by $\delta_{0}$. By $\sqrt\psi\le\tau_{1}$ and the
smallness assumption, we have both $\rho\sqrt\psi\le\sqrt d/64$ and
$\rho\tau_{1}\sqrt d\le d/64$. Thus, with $\delta_{0}=\eone^{2}d^{-3}$, the
two contributions are at most
\[
\frac{72}{64}\,\eone^{2}d^{-5/2}
\quad\text{and}\quad
\frac{3C_{0}}{64}\,\eone^{2}d^{-2},
\]
respectively. By the universal choice of $\eone$ in
Section~\ref{II:sec:fluct}, their sum is at most
$\eone$.
\end{proof}
\subsection{Fluctuation bound, parameter choice, and the acceptance
theorem}\label{II:sec:fluct}

\begin{lemma}[Gradient of $Z$]\label{II:lem:gradZ}
Assume Eq.~\eqref{II:eq:smallness}. For $h$ in a neighborhood of $\Ktube$,
\begin{equation}
\nabla Z(h)=2\,\Eb(h)\,h+W(h),
\qquad
W_{c}(h)=-2\eta\,\langle\Ab_{c},\;
\Rres[\Ab(h)\,\Rres\,\Ab(h)-Y]\Rres\rangle_{F},
\label{II:eq:gradZ}
\end{equation}
with $Y=\sum_{i}\Bh_{i}\Rres\Bh_{i}$ as in Lemma~\ref{II:lem:G-vector}, and on
$\Ktube$,
\begin{equation}
\norm{\nabla Z}\ \le\ 9\rho\,\tau_{1}M_{0}+3\rho\sqrt\psi
\ \le\ 12\,\rho\,\tau_{1}M_{0}\ =:\ \Lambda_{Z}.
\label{II:eq:lipschitz}
\end{equation}
Consequently $Z$ is $\Lambda_{Z}$-Lipschitz on the convex set
$\Ktube$.
\end{lemma}

\begin{proof}
We differentiate $Z=\sum_{ab}(h_{a}h_{b}-Q_{ab})\Eb_{ab}$ and obtain two
contributions. The product-rule term from
$h_{a}h_{b}$ is $2(\Eb h)_{c}$ by the symmetry of $\Eb$, and the term from $\Eb_{ab}$ is
$W_{c}=\sum_{ab}(h_{a}h_{b}-Q_{ab})\partial_{c}\Eb_{ab}$. For the
$h$-part, we contract Eq.~\eqref{II:eq:dE} with $h_{a}h_{b}$. As in
Lemma~\ref{II:lem:G-vector} with $\Bh_{i}$ replaced by $\Ab(h)$, both
three-letter words become $\tr[\Rres\Ab(h)\Rres\Ab(h)\Rres\Ab_{c}]$ after a
cyclic shift, which gives
$-2\eta\langle\Ab_{c},\Rres\Ab(h)\Rres\Ab(h)\Rres\rangle_{F}$. By
Lemma~\ref{II:lem:G-vector}, the $Q$-part is
$-G_{c}=+2\eta\langle\Ab_{c},\Rres Y\Rres\rangle_{F}$. Here we used
$\sum_{ab}Q_{ab}\partial_{c}\Eb_{ab}=G_{c}$, which holds because
$\partial_{c}\Eb_{ab}$ is totally symmetric in $(a,b,c)$ by
Eq.~\eqref{II:eq:dE}. Combining the two parts, we obtain
Eq.~\eqref{II:eq:gradZ}.

For the norm bound, the frame projection inequality,
Lemma~\ref{II:lem:psd-facts}(iii), the estimates
$\opnorm{\Ab(h)}\le\tau_{0}$ and
$\fnorm{\Ab(h)}=\norm{h}\le R_{0}$, and the bounds on $Y$ from
Lemma~\ref{II:lem:G-vector} give
\[
\begin{aligned}
\norm{W}
&\le2\eta\opnorm{\Rres}^{2}
 (\fnorm{\Ab(h)\Rres\Ab(h)}+\fnorm{Y})\\
&\le2\eta\opnorm{\Rres}^{3}
 (\tau_{0}R_{0}+\sqrt{\psi d})\\
&\le3\eta(\tau_{1}R_{1}+\sqrt{\psi d}),
\end{aligned}
\]
where the first step also uses the triangle inequality, the second step
uses Frobenius submultiplicativity, and the last step uses
$\opnorm{\Rres}\le64/63$, $\tau_0\le\tau_1$, and $R_0\le R_1$.

Moreover, Lemma~\ref{II:lem:stability}(iii) gives
$\norm{2\Eb h}\le2\opnorm{\Eb}\norm{h}\le6\eta\tau_{1}R_{1}$. Since
$\eta R_{1}=\rho M_{0}$ and $\eta\sqrt{\psi d}=\rho\sqrt\psi$, these two
estimates give the first bound in Eq.~\eqref{II:eq:lipschitz}. For the
second bound, we use $\sqrt\psi\le\tau_{0}\le\tau_{1}$ (from
$\tau_{0}^{2}=2\psi\log(8n/\delta_{0})\ge\psi$) and $M_{0}\ge1$. Finally,
the Lipschitz claim follows, since on a convex set a bound on the gradient
norm is also a bound on the Lipschitz constant.
\end{proof}

\begin{lemma}[Uniform size on the enlarged set]\label{II:lem:supZ}
On $\Ktubep$,
$|Z|\le3\rho\tau_{1}\sqrt d\,(M_{0}^{2}+1)\le6\rho\tau_{1}\sqrt d\,
M_{0}^{2}=:S_{Z}$.
\end{lemma}

\begin{proof}
By Lemma~\ref{II:lem:stability}(iii) and the operator-norm and
trace-duality bounds, we have
\[
\begin{aligned}
|h^{\top}\Eb h|
&\le\opnorm{\Eb}\norm{h}^{2}
 \le3\eta\tau_{1}R_{1}^{2}
 =3\rho\tau_{1}\sqrt d\,M_{0}^{2},\\
|\tr[Q\Eb]|
&\le\opnorm{Q}\snorm{\Eb}
 \le\sqrt d\,\fnorm{\Eb}
 \le3\rho\tau_{1}\sqrt d,
\end{aligned}
\]
where the first line follows from Lemma~\ref{II:lem:stability}(iii) and
$\eta R_1^2=\rho\sqrt d\,M_0^2$, and the second line follows from
$\snorm{\Eb}\le\sqrt d\,\fnorm{\Eb}$, $\opnorm Q\le1$, and again
Lemma~\ref{II:lem:stability}(iii).
\end{proof}

\begin{lemma}[Fluctuation bound]\label{II:lem:Z-bound}
Assume Eq.~\eqref{II:eq:smallness}. With probability at least
$1-\delta_{0}-\eone$ over $h\sim N(0,Q)$,
\begin{equation}
|Z|\ \le\ 7\rho^{2}\psi
+\Lambda_{Z}\sqrt{2\log(2/\eone)}
+\eone+40\,\rho\tau_{1}M_{0}^{2}\sqrt d\,\sqrt{\delta_{0}} .
\label{II:eq:Z-bound}
\end{equation}
\end{lemma}

\begin{proof}
Let $\widetilde Z$ be the McShane extension
(Lemma~\ref{II:lem:mcshane}) of $Z|_{\Ktube}$ with constant $\Lambda_{Z}$.
We first bound the mean of the extension. Since $\widetilde Z=Z\chi=Z$ on
$\Ktube$ and $\Pr[h\notin\Ktube]\le\delta_{0}$, the triangle inequality
and the Cauchy--Schwarz inequality give
\[
|\E[\widetilde Z]-\E[Z\chi]|
\le\E[(|\widetilde Z|+|Z\chi|)\ind{h\notin\Ktube}]
\le(\E[(|\widetilde Z|+|Z\chi|)^{2}])^{1/2}
\sqrt{\delta_{0}} .
\]

On all of $\R^{d}$, we have
$|\widetilde Z(h)|\le S_{Z}+\Lambda_{Z}(\norm h+R_{0})$ by Lipschitz
continuity from any point of $\Ktube$, and $|Z\chi|\le S_{Z}$ by
Lemma~\ref{II:lem:supZ}, since the cutoff is supported in $\Ktubep$.
Therefore, Minkowski's inequality and $\E[\norm h^{2}]\le d$ give
$|\E[\widetilde Z]-\E[Z\chi]|
\le(2S_{Z}+\Lambda_{Z}R_{0}+\Lambda_{Z}\sqrt d)\sqrt{\delta_{0}}
\le36\rho\tau_{1}M_{0}^{2}\sqrt d\sqrt{\delta_{0}}$, where we used
$R_0\le R_1=M_0\sqrt d$ and $M_0\ge1$. Relaxing $36$ to $40$ and
combining this with Lemma~\ref{II:lem:mean}, we obtain
$|\E[\widetilde Z]|\le7\rho^{2}\psi+\eone
+40\rho\tau_{1}M_{0}^{2}\sqrt d\sqrt{\delta_{0}}$.

To bound the deviation from the mean, observe that the map
$\xi\mapsto\widetilde Z(Q^{1/2}\xi)$ is $\Lambda_{Z}$-Lipschitz, since $\opnorm{Q^{1/2}}\le1$.
Therefore, by the two-sided form of Lemma~\ref{II:lem:gaussian-concentration}, we have
$\Pr[|\widetilde Z-\E[\widetilde Z]|
>\Lambda_{Z}\sqrt{2\log(2/\eone)}]\le\eone$. On the intersection of the
Gaussian-deviation event
\[
\mathcal C_Z:=\{|\widetilde Z-\E[\widetilde Z]|
\le\Lambda_{Z}\sqrt{2\log(2/\eone)}\}
\]
with the event $\mathcal T:=\{h\in\Ktube\}$, we have $Z=\widetilde Z$, and
Eq.~\eqref{II:eq:Z-bound} holds. Since $\Pr[\mathcal C_Z^c]\le\eone$,
$\Pr[\mathcal T^c]\le\delta_0$, and both events are defined under the same
Gaussian law, the union bound gives the stated probability.
\end{proof}

We now fix once and for all a sufficiently small universal $r\in(0,1]$ and
the universal failure level
\[
\eone:=\min\{\frac1{40},\frac{32}{72+3C_{0}}\},
\]
where $C_{0}$ is the cutoff constant in Lemma~\ref{II:lem:cutoff}. We then
set
\begin{equation}
\rho\ :=\ \frac{1}{\Gamma_{0}\,M_{0}\,\sqrt{\psistar}\,
(1+\log(nd/\eone))},
\qquad
\kappa:=\frac{r^{2}}{\rho^{2}}
=\Gamma_{0}^{2}\,r^{2}M_{0}^{2}\,\psistar\,
(1+\log(nd/\eone))^{2},
\label{II:eq:rho-choice}
\end{equation}
for a sufficiently large universal constant $\Gamma_{0}$, where we enlarge
$\Gamma_0$ if necessary so that $\kappa\ge1$. Then $\kappa=\Ot(\psistar)$,
and Eq.~\eqref{II:eq:tau-size} together with $\psi\le\psistar$ gives
the following:
\begin{itemize}[leftmargin=2em,itemsep=1pt]
\item $\eta\tau_{1}=\rho\tau_{1}/\sqrt d
\le4/(\Gamma_{0}\sqrt d)\le\tfrac1{64}$, and hence
Eq.~\eqref{II:eq:smallness} holds.
\item $7\rho^{2}\psi\le7/\Gamma_{0}^{2}$ and
$9\rho^{2}\tau_{1}^{2}\le405/\Gamma_{0}^{2}$
(Lemma~\ref{II:lem:det-remainder}).
\item $\Lambda_{Z}\sqrt{2\log(2/\eone)}
=12\rho\tau_{1}M_{0}\sqrt{2\log(2/\eone)}
\le C/\Gamma_{0}$ and
$40\rho\tau_{1}M_{0}^{2}\sqrt d\sqrt{\delta_{0}}
\le C\eone/\Gamma_{0}$.
\end{itemize}
Since $\eone$ is fixed, we can take $\Gamma_{0}$ universal and large enough
that each quantity displayed above is at most $\tfrac{1}{10}$, and we
assume this choice henceforth.

\begin{lemma}[Uniform lower tail for the complete ratio]
\label{II:lem:acceptance}
With the choice Eq.~\eqref{II:eq:rho-choice} there are universal constants
$C_{\mathrm{acc}}<\infty$ and $\epsilon\le\tfrac1{10}$ such that,
uniformly at every deterministic $x\in\intK$,
\begin{equation}
\Pr_{Y\sim q_{x}}[Y\in\intK\ \text{ and }\
\log R_{x}(Y)\ge-C_{\mathrm{acc}}]\ \ge\ 1-\epsilon .
\label{II:eq:acceptance}
\end{equation}
The constants are independent of $x,n,d,L,R$ and of the orientation
of the regularizer.
\end{lemma}

\begin{proof}
We work in normalized coordinates, which leave the ratio invariant
(Section~\ref{II:sec:ledger}), and write the proposal as
$Y=x+\eta h$, $h=Q^{1/2}\xi$. We intersect three events. The event
$\mathcal T=\{h\in\Ktube\}$ that the step lies in the good set fails with
probability at most $\delta_{0}$ by Lemma~\ref{II:lem:tube-prob}, and
the Gaussian-deviation event $\mathcal C_Z$ from the proof of
Lemma~\ref{II:lem:Z-bound} fails with probability at most $\eone$. The
third is the norm event $\{\norm\xi\le\sqrt d+\sqrt{2\log(1/\eone)}\}$. By the
one-sided form of Lemma~\ref{II:lem:gaussian-concentration}, the norm event
has probability $\ge1-\eone$, since $\xi\mapsto\norm\xi$ is $1$-Lipschitz
and $\E[\norm\xi]\le\sqrt d$. Moreover, since
$\delta_{0}=\eone^{2}d^{-3}\le\eone^{2}$, the total failure probability is
at most $\epsilon:=\eone^{2}+2\eone\le\tfrac1{10}$.

For feasibility, note that on the good set we have
$\opnorm{\Delta(h)}\le\eta\tau_{0}<1$, and hence, by
Lemma~\ref{II:lem:endpoint}, $Y\in\intK$ and the whole segment is interior.

Next, we bound the proposal ratio. By Lemma~\ref{II:lem:ledger}, the
proposal log-ratio is
$\log\tfrac{q_{Y}(x)}{q_{x}(Y)}
=\tfrac12\{\log\det(I+E)-\tr[E]\}-\tfrac12 Z$. On the intersection of the
three events, Lemma~\ref{II:lem:det-remainder} and
Lemma~\ref{II:lem:Z-bound} with the parameter choice give
$|\log\det(I+E)-\tr[E]|\le\tfrac1{10}$ and $|Z|\le\tfrac{4}{10}$. Hence,
the proposal log-ratio is at least $-\tfrac12$.

It remains to bound the target ratio. If $L=0$, then $f$ is constant on the
convex body $K$, and the target log-ratio is zero. If $L>0$, then
$G(x)\succeq L^{2}I$ in the original coordinates, and hence
$\norm{Y-x}_{2}\le L^{-1}\norm{Y-x}_{G(x)}
=L^{-1}\tfrac{r}{\sqrt d}\norm\xi$. On the norm event, by Lipschitzness
and the proposal representation, we have
\[
f(Y)-f(x)\le L\norm{Y-x}_{2}
\le\frac{r}{\sqrt d}(\sqrt d+\sqrt{2\log(1/\eone)})
\le C_{\eone}r,
\qquad C_{\eone}:=1+\sqrt{2\log(1/\eone)} .
\]

Putting things together, we obtain
$\log R_{x}(Y)\ge-C_{\eone}r-\tfrac12\ge-C_{\mathrm{acc}}$ with
$C_{\mathrm{acc}}:=C_{\eone}r+1$, which is a universal constant because
$\eone$ is fixed.
\end{proof}

\begin{remark}
Eq.~\eqref{II:eq:acceptance} holds at every interior point with no
dependence on the distance to $\partial K$, since our analysis is
affine-invariant and uses only $\creg\succeq0$. Hence, we will not need any
boundary-layer or $s$-conductance argument.
\end{remark}
\subsection{From acceptance to mixing}\label{II:sec:mixing}

We now return to the original coordinates and follow the
geometric half of the polytope argument, which relies on the
cross-ratio and Cheeger lemmas, Lemmas~\ref{II:lem:isoperimetry}
and~\ref{II:lem:cheeger}.

\begin{lemma}[Operator self-concordance]\label{II:lem:op-sc}
For all $z\in\intK$ and $u\in\R^{d}$,
\begin{equation}
\opnorm{\,\Hld(z)^{-1/2}\,D\Hld(z)[u]\,\Hld(z)^{-1/2}}
\ \le\ 2\,\norm{u}_{\Hld(z)},
\label{II:eq:op-sc}
\end{equation}
and the same holds for $G=\kappa\Hld+L^{2}I$ in the $G$-norms.
Consequently, Lemma~\ref{II:lem:distortion} applies to $G$: if
$\Delta:=\norm{x-z}_{G(x)}<1$ then
$(1-\Delta)^{2}G(x)\preceq G(z)\preceq(1-\Delta)^{-2}G(x)$.
\end{lemma}

\begin{proof}
We normalize at $z$ as in Section~\ref{II:sec:ledger}, which gives
$\Hld(z)=I$ and makes directions $v$ correspond to
$\Ab(v)$ with $\fnorm{\Ab(v)}=\norm v$. Since
$\Hld=\nabla^{2}\phi$ with $\phi=-\log\det S$, in these coordinates we have
$v^{\top}D\Hld(z)[u]\,v=D^{3}\phi(z)[u,v,v]
=-2\tr[\Ab(u)\Ab(v)^{2}]$. This trace satisfies
\[
|\tr[\Ab(u)\Ab(v)^{2}]|
\le\opnorm{\Ab(u)}\,\tr[\Ab(v)^{2}]
\le\fnorm{\Ab(u)}\,\fnorm{\Ab(v)}^{2}
=\norm{u}\,\norm{v}^{2},
\]
where the first step follows from trace duality,
$|\tr[TM]|\le\opnorm{T}\tr[M]$ for symmetric $T$ and $M\succeq0$, applied
with $M=\Ab(v)^{2}$ and with $T=\Ab(u)$ (which may be indefinite), the
second step follows from $\opnorm{\Ab(u)}\le\fnorm{\Ab(u)}$ and
$\tr[\Ab(v)^2]=\fnorm{\Ab(v)}^2$, and the last step follows from the
isometry of the normalized coordinates.

Taking the supremum over unit $v$, we obtain Eq.~\eqref{II:eq:op-sc}. For
$G$, scaling by $\kappa\ge1$ preserves the inequality, since
the left side is scale-invariant and the right side grows. We set
$H_{\kappa}:=\kappa\Hld$, $\Xi:=H_{\kappa}^{1/2}G^{-1/2}$, and
$X:=H_{\kappa}^{-1/2}DH_{\kappa}[u]H_{\kappa}^{-1/2}$, where
$\opnorm{\Xi}\le1$ holds since $H_{\kappa}\preceq G$. Then we have
\[
\opnorm{G^{-1/2}DG[u]G^{-1/2}}=\opnorm{\Xi^{\top}X\Xi}
\le\opnorm X\le2\norm{u}_{H_{\kappa}}\le2\norm{u}_{G},
\]
where the first step follows from $DG=DH_{\kappa}$ and the definitions of
$\Xi$ and $X$, the second step follows from $\opnorm{\Xi}\le1$, the third
step follows from the self-concordance of $H_\kappa$, and the last step
follows from $H_\kappa\preceq G$.

The distortion conclusion follows from the elementary ODE lemma,
Lemma~\ref{II:lem:distortion}, whose domain hypothesis holds because the
segment $[x,z]$ lies in the convex set $\intK$.
\end{proof}

\begin{lemma}[Proposal overlap]\label{II:lem:proposal-overlap}
There is a universal $c_{\mathrm{loc}}>0$ such that, with $r\le1$,
\begin{equation}
\norm{x-z}_{G(x)}\le\frac{c_{\mathrm{loc}}\,r}{\sqrt d}
\quad\Longrightarrow\quad
\dtv(q_{x},q_{z})\le\delta_{q}:=\tfrac{1}{2}.
\label{II:eq:proposal-overlap}
\end{equation}
\end{lemma}

\begin{proof}
Let $\Delta=\norm{x-z}_{G(x)}\le\tfrac12$. For the two Gaussian proposals
$q_{x}=N(x,\tfrac{r^{2}}{d}G(x)^{-1})$ and
$q_{z}=N(z,\tfrac{r^{2}}{d}G(z)^{-1})$, the KL divergence is
\[
D_{\mathrm{KL}}(q_{x}\Vert q_{z})
=\frac12[\sum_{j=1}^{d}(\lambda_{j}-1-\log\lambda_{j})
+\frac{d}{r^{2}}\norm{x-z}_{G(z)}^{2}],
\]
where $\lambda_{j}$ are the eigenvalues of
$G(x)^{-1/2}G(z)G(x)^{-1/2}$. This formula follows from integrating the
log-density ratio: the centered quadratic term gives the trace of the
relative covariance, the mean displacement gives
$d\norm{x-z}_{G(z)}^2/r^2$, and the normalizers give the log
determinant. By Lemma~\ref{II:lem:op-sc}, the eigenvalues satisfy
$\lambda_{j}\in[(1-\Delta)^{2},(1-\Delta)^{-2}]$. This gives
$\lambda_{j}-1-\log\lambda_{j}\le C\Delta^{2}$ for each eigenvalue and
$\norm{x-z}_{G(z)}^{2}\le(1-\Delta)^{-2}\Delta^{2}\le4\Delta^{2}$.
Since $\Delta\le c_{\mathrm{loc}}r/\sqrt d$ and $r\le1$, we have
$D_{\mathrm{KL}}\le C(d\Delta^{2}+d\Delta^{2}/r^{2})
\le C'c_{\mathrm{loc}}^{2}$. Applying Pinsker's inequality
$\dtv\le\sqrt{D_{\mathrm{KL}}/2}$ (Lemma~\ref{II:lem:pinsker}), we obtain
Eq.~\eqref{II:eq:proposal-overlap} for sufficiently small $c_{\mathrm{loc}}$.
\end{proof}

\begin{lemma}[Kernel overlap]\label{II:lem:kernel-overlap}
Suppose Eq.~\eqref{II:eq:acceptance} holds at $x$ and $z$ and
$\dtv(q_{x},q_{z})\le\delta_{q}=\tfrac12$. Then the
$\tfrac12$-lazy kernels satisfy
\begin{equation*}
\dtv(P_{x},P_{z})\ \le\ 1-\frac{\gamma_{0}}{2},
\qquad
\gamma_{0}:=e^{-C_{\mathrm{acc}}}
(1-\delta_{q}-2\epsilon)
\ \ge\ e^{-C_{\mathrm{acc}}}\cdot\tfrac{3}{10}\ >\ 0 .
\end{equation*}
\end{lemma}

\begin{proof}
Let $E_{x}:=\{u\in\intK:R_{x}(u)\ge e^{-C_{\mathrm{acc}}}\}$ and define
$E_{z}$ likewise. By Lemma~\ref{II:lem:acceptance}, we have
$q_{x}(E_{x}^{c})\le\epsilon$ and $q_{z}(E_{z}^{c})\le\epsilon$. Note that
the moving part of the non-lazy kernel at $x$ has density
$u\mapsto q_{x}(u)(1\wedge R_{x}(u))$ on $K$. Since both acceptance ratios
are at least $e^{-C_{\mathrm{acc}}}$ on $E_x\cap E_z$, and since
$\min\{q_x,q_z\}$ has total mass at least $1-\delta_q$ and loses at most
$\epsilon$ on each of $E_x^c$ and $E_z^c$, the two moving parts have
common mass at least
\[
e^{-C_{\mathrm{acc}}}\int_{E_{x}\cap E_{z}}
\min\{q_{x},q_{z}\}\,du
\ \ge\ e^{-C_{\mathrm{acc}}}
(1-\delta_{q}-\epsilon-\epsilon)=\gamma_{0}.
\]

Since laziness retains half of the common moving mass and the total
variation distance satisfies $\dtv\le1-(\text{overlap})$, the first claim
follows. Finally, we have
$1-\delta_{q}-2\epsilon\ge1-\tfrac12-\tfrac{2}{10}=\tfrac{3}{10}$.
\end{proof}

\begin{definition}[Chord cross ratio]\label{II:def:crossratio}
For distinct $x,z\in\intK$, let their line meet $\partial K$ in the
order $p_{\partial},x,z,q_{\partial}$ and put
$\alpha=\norm{x-p_{\partial}}$, $\beta=\norm{z-x}$,
$\gamma=\norm{q_{\partial}-z}$ (Euclidean). Then
$\sigma(x,z):=\beta(\alpha+\beta+\gamma)/(\alpha\gamma)$.
\end{definition}

\begin{lemma}[Symmetry]\label{II:lem:symmetry}
Write $\mathcal D_{M}(x,t)=\{z:\norm{z-x}_{M(x)}\le t\}$. Then, for
every $x\in\intK$,
\[
\mathcal D_{G}(x,1)\ \subseteq\ K\cap(2x-K)\ \subseteq\
\mathcal D_{G}(x,\sqrt{\nu_{G}}),
\qquad
\nu_{G}:=\kappa\,n+4L^{2}R^{2} .
\]
\end{lemma}

\begin{proof}
We work at $x$ in normalized coordinates, in which
$\norm{z-x}_{\Hld(x)}=\fnorm{\Delta}$ for
$\Delta:=$ the normalized block displacement of $z$, and in which
$z\in K\cap(2x-K)$ is equivalent to $I\pm\Delta\succeq0$, that is, to
$\opnorm\Delta\le1$.

For the inner inclusion, if $\norm{z-x}_{G(x)}\le1$, then,
since $G\succeq\kappa\Hld$ and $\kappa\ge1$, we have
$\fnorm\Delta\le\kappa^{-1/2}\le1$ and hence $\opnorm\Delta\le1$.

For the outer inclusion, if $\opnorm\Delta\le1$, then
$\norm{z-x}_{\kappa\Hld(x)}^{2}=\kappa\fnorm\Delta^{2}
\le\kappa\,n\,\opnorm\Delta^{2}\le\kappa n$. Moreover,
$L^{2}\norm{z-x}_{2}^{2}\le4L^{2}R^{2}$ because $K\subseteq B_{2}(0,R)$
has diameter $2R$. Adding these two bounds proves the outer inclusion.
\end{proof}

\begin{lemma}[Cross ratio dominates the local norm]\label{II:lem:crossratio}
For all distinct $x,z\in\intK$,
$\sigma(x,z)\ge\norm{x-z}_{G(x)}/\sqrt{\nu_{G}}$.
\end{lemma}

\begin{proof}
Let $e=(z-x)/\beta$. Then the radius of the symmetric body
$K\cap(2x-K)$ at $x$ along $\pm e$ is $t=\min(\alpha,\beta+\gamma)$.
Thus $x+te$ lies in the closure of this body, and by
Lemma~\ref{II:lem:symmetry}, we have $t\,\norm{e}_{G(x)}\le\sqrt{\nu_{G}}$.
It remains to show that $\sigma(x,z)\ge\beta/t$. If $t=\alpha$, the cross ratio
is $(\alpha+\beta+\gamma)/\gamma\ge1$ times $\beta/\alpha$, and if
$t=\beta+\gamma$, we use
$(\alpha+\beta+\gamma)(\beta+\gamma)\ge\alpha\gamma$. Hence
$\sigma\ge\beta\norm{e}_{G(x)}/\sqrt{\nu_{G}}
=\norm{z-x}_{G(x)}/\sqrt{\nu_{G}}$.
\end{proof}

\begin{lemma}[Conductance]\label{II:lem:conductance}
The walk of Definition~\ref{II:def:walk} has conductance
\begin{equation}
\Phi:=\inf_{0<\pi(A)\le1/2}
\frac{\mathcal Q(A,A^c)}{\pi(A)}
\ \ge\ \frac{c_{\mathrm{iso}}\,\gamma_{0}}{128}\,\lsep,
\qquad
\lsep:=\min(1,\ \frac{c_{\mathrm{loc}}\,r}
{\sqrt{d\,\nu_{G}}}).
\label{II:eq:conductance}
\end{equation}
Here
$\mathcal Q(A,A^c):=\int_A P_x(A^c)\,\pi(dx)$ is the stationary flow.
\end{lemma}

\begin{proof}
We fix a measurable set $A$ with $0<\pi(A)\le\tfrac12$, put
$b_{0}:=\gamma_{0}/8$, and define the sets
\[
A_{1}=\{x\in A:P_{x}(A^{c})\le b_{0}\},\qquad
A_{2}=\{z\in A^{c}:P_{z}(A)\le b_{0}\},\qquad
A_{3}=K\setminus(A_{1}\cup A_{2}).
\]
Since $\pi(\partial K)=0$, from now on we intersect $A_1,A_2$ with
$\intK$ and absorb $\partial K$ into $A_3$. This changes neither mass nor
stationary flow and makes every pair below an interior pair.

For $x\in A_{1},z\in A_{2}$, we have
$\dtv(P_{x},P_{z})\ge P_{z}(A^{c})-P_{x}(A^{c})\ge1-2b_{0}
>1-\gamma_{0}/2$. Hence, by Lemmas~\ref{II:lem:kernel-overlap} and
\ref{II:lem:proposal-overlap}, the pair must violate the closeness
condition, that is, $\norm{x-z}_{G(x)}>c_{\mathrm{loc}}r/\sqrt d$, and
Lemma~\ref{II:lem:crossratio} then gives $\sigma(x,z)\ge\lsep$. Let
$B_{1}=A\setminus A_{1}$ and $B_{2}=A^{c}\setminus A_{2}$, which satisfy
$A_{3}=B_{1}\dot\cup B_{2}$. For the stationary flow $\mathcal Q$,
reversibility gives $\mathcal Q(A,A^{c})\ge b_{0}\pi(B_{1})$ and
$\mathcal Q(A,A^{c})=\mathcal Q(A^{c},A)\ge b_{0}\pi(B_{2})$.

We now distinguish two cases. If $\pi(A_{1})<\pi(A)/2$ or
$\pi(A_{2})<\pi(A^{c})/2$, then
$\pi(B_{1})\ge\pi(A)/2$ or $\pi(B_{2})\ge\pi(A^{c})/2\ge\pi(A)/2$,
which gives the bound $\mathcal Q(A,A^{c})/\pi(A)\ge b_{0}/2$. Otherwise,
we have $\pi(A_{1})\ge\pi(A)/2$ and $\pi(A_{2})\ge\tfrac14$. Since $\pi$ is
log-concave ($f$ convex) on the convex body $K$, and since we showed above
that $\sigma(x,z)\ge\lsep$ for $x\in A_{1},z\in A_{2}$, all hypotheses of
the cross-ratio isoperimetry result in Lemma~\ref{II:lem:isoperimetry}
hold. Applying this result to $(A_{1},A_{2},A_{3})$ with separation
$\lsep\le1$, we obtain
$\pi(A_{3})\ge c_{\mathrm{iso}}\lsep\,\pi(A_{1})\pi(A_{2})
\ge\tfrac{c_{\mathrm{iso}}\lsep}{8}\pi(A)$. Together with
$2\mathcal Q(A,A^{c})\ge b_{0}\pi(A_{3})$, this yields
$\mathcal Q(A,A^{c})/\pi(A)\ge
b_{0}c_{\mathrm{iso}}\lsep/16$. Both cases are covered by
Eq.~\eqref{II:eq:conductance}.
\end{proof}
\subsection{The arrow spectrahedron}
\label{II:sec:tight}

The goal of this subsection is to show that, for a fixed universal proposal radius,
uniform constant acceptance can require $\kappa=\Omega(\psi(x))$.  Since the
\emph{arrow} instance below has $\psi(0)=\Theta(\psistar)=\Theta(d)$, this
requirement matches the choice $\kappa=\Ot(\psistar)$ of our theorem up to
polylogarithmic factors.  The arrow instance is given by
\begin{equation*}
n=d+1,\qquad
A_{0}=I_{n},\qquad
A_{a}=\tfrac1{\sqrt2}(e_{1}e_{a+1}^{\top}
+e_{a+1}e_{1}^{\top})\quad(a=1,\dots,d),
\end{equation*}
for which we have
\[
S(x)=\begin{pmatrix}1&x^{\top}/\sqrt2\\x/\sqrt2&I_d\end{pmatrix}.
\]
For $x\ne0$, it acts on the span of $e_1$ and $(0,x/\norm x)$ as the
$2\times2$ matrix with diagonal entries $1$ and off-diagonal entries
$\norm x/\sqrt2$, and as the identity on the orthogonal complement.  Thus
its eigenvalues are $1\pm\norm x/\sqrt2$ and $1$ (with multiplicity $d-1$),
which implies that $K=\{x:\norm x\le\sqrt2\}$,
$\det S(x)=1-\norm x^{2}/2$, and $S(x)\succeq0$ iff
$\norm x\le\sqrt2$.  Moreover, we have
\[
A_a^2=\tfrac12(e_1e_1^{\top}+e_{a+1}e_{a+1}^{\top}),
\qquad \langle A_a,A_b\rangle_F=\delta_{ab}.
\]
Hence the frame is already Frobenius-orthonormal at $x=0$ and
$\Hld(0)=I$, and the displayed squares sum to
\begin{equation*}
\Psib(0)=\sum_{a}A_{a}^{2}
=\tfrac{d}{2}\,e_{1}e_{1}^{\top}+\tfrac12\sum_{a=1}^{d}
e_{a+1}e_{a+1}^{\top},
\qquad
\psi(0)=\frac d2 .
\end{equation*}

\begin{lemma}[Acceptance at the center of the arrow]
\label{II:lem:arrow}
Consider the walk of Definition~\ref{II:def:walk}, with $\kappa>0$ and $r>0$
allowed, on the arrow instance with $f\equiv0$, started at $x=0$. Allow any scale $\kappa>0$, any $L\ge0$ (the
function $f\equiv0$ is $L$-Lipschitz for every $L$) and any radius $r>0$, and put
\[
\kappa_{L}:=\kappa+L^{2},\qquad \bar\kappa:=\frac{\kappa}{\kappa_{L}}\in(0,1],\qquad
\mathcal{J}:=\frac{d\,r^{2}\kappa}{\kappa_{L}^{2}} .
\]
\begin{enumerate}
\item[(a)] The proposal is $Y\sim N(0,\frac{r^{2}}{d\kappa_{L}}I)$. For $Y\in\intK$, the
log-ratio $\Lambda(Y):=\log\frac{q_{Y}(0)}{q_{0}(Y)}$ depends on $Y$ only
through $t:=\norm{Y}^{2}<2$:
\begin{equation}
\Lambda
=\frac d2\log\frac{2-(1-\bar\kappa)t}{2-t}
+\frac12\log\Bigl(1+\frac{4\bar\kappa t}{(2-t)(2-(1-\bar\kappa)t)}\Bigr)
-\frac{d\kappa}{2r^{2}}\cdot\frac{t^{2}(6-t)}{(2-t)^{2}} .
\label{II:eq:arrow-lambda}
\end{equation}
\item[(b)] For every $d$ and every $s_{0}\in(\frac13+\frac{2}{3d},1)$, the
probability that the proposed step is accepted is at most
\[
\exp\Bigl(-s_{0}\bigl(\tfrac{3s_{0}-1}{4}-\tfrac{1}{2d}\bigr)\mathcal{J}\Bigr)
+\exp\Bigl(-\tfrac d2\bigl(s_{0}-1-\log s_{0}\bigr)\Bigr).
\]
\item[(c)] In particular, for $d\ge8$, this probability is at most
$e^{-3\mathcal{J}/16}+e^{-d/54}$. For the unscaled walk ($\kappa=1$ and $L=0$, so that
$\mathcal{J}=dr^{2}$), it is at most
\[
e^{-dr^{2}/10}+2e^{-d/54}
\]
for every $r>0$ and every $d\ge8$.
Consequently, constant acceptance requires
$r\lesssim d^{-1/2}\asymp\psi(0)^{-1/2}$, and hence the $\psi$-scaling of
Eq.~\eqref{II:eq:rho-choice} is necessary up to logarithmic factors.
\end{enumerate}
\end{lemma}

\begin{proof}
At $x=0$, the metric is $\Hld(0)=I$, which gives $G(0)=\kappa_{L}I$ and
$Y\sim N(0,\frac{r^{2}}{d\kappa_{L}}I)$.  Writing $Y=\sqrt{r^{2}/(d\kappa_{L})}\,\xi$ with
$\xi\sim N(0,I_{d})$, we define $s:=\norm{\xi}^{2}/d$ and $\varepsilon:=r^{2}/\kappa_{L}$,
so that $t=\varepsilon s$ and $\mathcal{J}=d\bar\kappa\varepsilon$.  Since $f\equiv0$, a proposal
$Y\in\intK$ is accepted with probability $1\wedge e^{\Lambda(Y)}$, and a
proposal outside $\intK$ is rejected.  Hence the acceptance probability is
$\E[\ind{Y\in\intK}(1\wedge e^{\Lambda})]$.

\emph{Part (a).} Since $\phi(y)=-\log\det S(y)=-\log(1-t/2)$ with
$t=\norm y^{2}$, we differentiate first in $t$ and obtain
$\nabla\phi(y)=y/(1-t/2)$ and
\[
\Hld(y)=aI+a^{2}yy^{\top},\qquad a:=\frac{1}{1-t/2}=\frac{2}{2-t}\ \ge1 .
\]
Hence $G(y)=(\kappa a+L^{2})I+\kappa a^{2}yy^{\top}$, and by the matrix determinant lemma, we have
\[
\log\frac{\det G(y)}{\det G(0)}
=d\log\frac{\kappa a+L^{2}}{\kappa_{L}}+\log\Bigl(1+\frac{\kappa a^{2}t}{\kappa a+L^{2}}\Bigr)
=d\log\bigl(1+\bar\kappa(a-1)\bigr)+\log\Bigl(1+\frac{\bar\kappa a^{2}t}{1+\bar\kappa(a-1)}\Bigr).
\]
Since $G(y)-G(0)=\kappa(\Hld(y)-I)$, the identity term cancels from the
quadratic form, and using $a-1=at/2$, we obtain
\[
y^{\top}(G(y)-G(0))y=\kappa\bigl((a-1)t+a^{2}t^{2}\bigr)
=\kappa t^{2}\bigl(\tfrac a2+a^{2}\bigr).
\]
The Gaussian densities give
$\Lambda=\frac12\log\frac{\det G(Y)}{\det G(0)}
-\frac{d}{2r^{2}}Y^{\top}(G(Y)-G(0))Y$.  Substituting $a=2/(2-t)$, so that
$1+\bar\kappa(a-1)=\frac{2-(1-\bar\kappa)t}{2-t}$ and $\frac a2+a^{2}=\frac{6-t}{(2-t)^{2}}$,
we obtain Eq.~\eqref{II:eq:arrow-lambda}.

\emph{An upper bound on $\Lambda$.} We apply $\log(1+v)\le v$ to both
logarithms, with $\bar\kappa(a-1)=\bar\kappa at/2$ in the first and
$1+\bar\kappa(a-1)\ge1$ in the second.  Using
$\frac{d\kappa}{r^{2}}t^{2}=\frac{d\kappa\varepsilon^{2}}{r^{2}}s^{2}=\mathcal{J}s^{2}$,
$d\bar\kappa t=\mathcal{J}s$, and $\bar\kappa t=\mathcal{J}s/d$, we obtain, for $Y\in\intK$,
\[
\Lambda\ \le\ \frac{d\bar\kappa a t}{4}+\frac{\bar\kappa a^{2}t}{2}
-\frac{\mathcal{J}}2 s^{2}\Bigl(\frac a2+a^{2}\Bigr)
\ =\ \mathcal{J}s\Bigl(\frac a4+\frac{a^{2}}{2d}\Bigr)
-\frac{\mathcal{J}}2 s^{2}\Bigl(\frac a2+a^{2}\Bigr).
\]
If $s\ge s_{0}$, then $s^{2}\ge s_{0}s$, and hence
\[
\Lambda\ \le\ \mathcal{J}s\,k(a),\qquad
k(a):=\frac{(1-s_{0})a}{4}-\frac{(s_{0}-1/d)\,a^{2}}{2}.
\]
The hypothesis $s_{0}>\frac13+\frac2{3d}$ implies $s_{0}>\frac1d$ and
$s_{0}\ge\frac15+\frac4{5d}$.  Thus, for $a\ge1$, the derivative satisfies
\[
k'(a)=\frac{1-s_{0}}{4}-\Bigl(s_{0}-\frac1d\Bigr)a
\ \le\ \frac{1-5s_{0}}{4}+\frac1d\ \le\ 0 .
\]
This implies $k(a)\le k(1)=-\bigl(\frac{3s_{0}-1}{4}-\frac1{2d}\bigr)$, which is
negative by the hypothesis on $s_{0}$.  Therefore, since $s\ge s_{0}$, on the event
$\{s\ge s_{0}\}\cap\{Y\in\intK\}$ we have
\[
\Lambda\ \le\ \mathcal{J}s\,k(1)\ \le\ -s_{0}\Bigl(\frac{3s_{0}-1}{4}-\frac1{2d}\Bigr)\mathcal{J} .
\]

\emph{Part (b).} By Chernoff's bound with
$\E[e^{-\lambda\norm{\xi}^{2}}]=(1+2\lambda)^{-d/2}$ and
$\lambda=\frac12(\frac1{s_{0}}-1)$, we have
\[
\Pr[s<s_{0}]\ \le\ e^{\lambda s_{0}d}(1+2\lambda)^{-d/2}
=\exp\Bigl(-\frac d2\bigl(s_{0}-1-\log s_{0}\bigr)\Bigr).
\]
Splitting the expectation over $\{s<s_{0}\}$ and $\{s\ge s_{0}\}$, we obtain
\[
\E[\ind{Y\in\intK}(1\wedge e^{\Lambda})]
\ \le\ \Pr[s<s_{0}]
+\sup_{s\ge s_{0},\,Y\in\intK}e^{\Lambda},
\]
which is the bound in (b).

\emph{Part (c).} We take $s_{0}=\frac34$, which is admissible for $d\ge8$. Then
$s_{0}(\frac{3s_{0}-1}{4}-\frac1{2d})=\frac34(\frac5{16}-\frac1{2d})
\ge\frac3{16}$ and
$\frac12(s_{0}-1-\log s_{0})=\frac12(\log\frac43-\frac14)>\frac1{54}$.
For $\kappa=1$ and $L=0$, we have $\mathcal{J}=dr^{2}$ and $\frac3{16}\ge\frac1{10}$.
Finally, if the acceptance probability is at least $p\in(0,1]$ and
$d\ge54\log(2/p)$, then $e^{-3\mathcal{J}/16}\ge p/2$, or equivalently
$\mathcal{J}\le\frac{16}{3}\log\frac2p$.  For $\kappa=1$ and $L=0$, this bound
reads $dr^{2}\le\frac{16}{3}\log\frac2p$.
\end{proof}

\begin{remark}[Reading the bound]
\label{II:rem:arrow}
Lemma~\ref{II:lem:arrow} proves a necessity statement for one specific family.
At the center of the arrow instance, the acceptance probability is at most
$e^{-3\mathcal{J}/16}+e^{-d/54}$ with $\mathcal{J}=dr^{2}\kappa/(\kappa+L^{2})^{2}$, and hence constant
acceptance requires $\mathcal{J}=O(1)$. For the unscaled walk, this means
$r\lesssim\psi(0)^{-1/2}$, which, together with
Lemmas~\ref{II:lem:stein-integrand} and \ref{II:lem:Z-bound}, matches the
sufficient leverage scale up to logarithmic factors on that family.
However, the lemma proves neither a pointwise ``if and only if'' for all
spectrahedra nor a mixing-time lower bound.

For $L=0$, the bound reads $e^{-3dr^{2}/(16\kappa)}+e^{-d/54}$, and hence scaling
the center metric from $I$ to $\kappa I$ acts as the proposal radius
$r_{\mathrm{eff}}=r/\sqrt\kappa$. Thus, at a fixed universal $r$, constant
acceptance forces $dr^{2}/\kappa=O(1)$, that is,
$\kappa=\Omega(d)=\Omega(\psi(0))$ on this family.
The identity term can restore constant acceptance on this example
only if it is large. Indeed, by
Lemma~\ref{II:lem:arrow}(c), constant acceptance forces
$(\kappa+L^{2})^{2}\gtrsim dr^{2}\kappa$, for example $L^{2}\gtrsim\sqrt d$
at $\kappa=1$ and a fixed universal $r$.

The regularized log-det metric of~\cite{gkmsz24} is
$\check\Phi=\check\alpha^{-1}\Hld+\check\eta^{-1}I$ in their notation (a check marks a symbol
of~\cite{gkmsz24}), and it equals $\frac{d}{r^{2}}G$ with
$\check\alpha=r^{2}/(d\kappa)$ and $\check\eta=r^{2}/(dL^{2})$. With these
parameters, we have
$\mathcal{J}=d^{2}\check\alpha\check\eta^{2}/(\check\alpha+\check\eta)^{2}$. Hence, if $\check\alpha$ and $\check\eta$ are both of
order $1/d$, then $\mathcal{J}$ is of order $d$ and the acceptance at the center is
exponentially small in $d$. Lemma~\ref{II:lem:arrow-gkmsz} below states the
consequences for two lemmas of~\cite{gkmsz24}.

At the other extreme, for polytopes ($\psi\le1$ by
Lemma~\ref{II:lem:leverage-basic}) and for incoherent instances
($\psi=\Ot(1)$) alike, Theorem~\ref{II:thm:main-formal} reads
$\Ot(nd)$, namely dimension times symmetry. In the diagonal case, this recovers the
Kannan--Narayanan bound \cite{kn12}. Moreover, since $\psi\le d$ always, the bound of
the theorem never exceeds the generic barrier-scaling bound $\Ot(d^{2}n)$
(cf.~\cite{n16,kv24}).
\end{remark}

\begin{lemma}[Two lemmas and one assumption of~\cite{gkmsz24} on the arrow]
\label{II:lem:arrow-gkmsz}
On the arrow instance let $\check\Phi(x):=\check\alpha^{-1}\Hld(x)+\check\eta^{-1}I$ with
$\check\alpha,\check\eta>0$, and put $\mathcal{J}:=d^{2}\check\alpha\check\eta^{2}/(\check\alpha+\check\eta)^{2}$. Let
$\xi\sim N(0,I_{d})$ and $u:=\check\Phi(0)^{-1/2}\xi$.
\begin{enumerate}
\item[(i)] (The statistic of their Lemma~D.5 at $x=0$.) On the event
$\{u\in\intK\}$,
\[
\mathcal{V}:=\bigl|\norm{u}_{\check\Phi(u)}^{2}-\norm{u}_{\check\Phi(0)}^{2}\bigr|
=\frac{t^{2}(6-t)}{\check\alpha(2-t)^{2}}
\ \ge\ \frac32\,\mathcal{J}\Bigl(\frac{\norm{\xi}^{2}}{d}\Bigr)^{2},
\qquad t:=\norm{u}^{2}=\frac{\check\alpha\check\eta}{\check\alpha+\check\eta}\norm{\xi}^{2}.
\]
Suppose $\check\eta\le0.001/d$. If $\mathcal{J}\ge2/75$, then
$\Pr[u\in\intK,\ \mathcal{V}\le0.01]\le0.83^{d/2}$ and
$\Pr[u\notin\intK]\le e^{-995d}$. Hence the event of their Lemma~D.5 has
probability less than $0.999$ for every $d$, whichever way points
$u\notin\intK$ are counted, and this probability tends to $0$ as
$d\to\infty$. If instead $\mathcal{J}\le7\cdot10^{-4}$ and $d\ge16$, then
$\Pr[\mathcal{V}\le0.01]\ge0.999$.
\item[(ii)] (Their Lemma~D.8 at $x=0$.) For $y\in\intK$ put
$\check G_{y}(v):=\det(\check\Phi(y))^{1/2}\exp(-\frac12\norm{v-y}_{\check\Phi(y)}^{2})$. Let
$\check P_{x}$ be the one-step kernel of the walk that proposes
$u\sim N(x,\check\Phi(x)^{-1})$, rejects $u\notin\intK$, and otherwise accepts $u$
with probability $1\wedge\check G_{u}(x)/\check G_{x}(u)$. The kernel may be made
$\frac12$-lazy or not. This is the Metropolis filter of
Definition~\ref{II:def:walk} for the uniform target, with the exact metric
$\check\Phi$ at both endpoints, and it is the kernel that the proof of their Lemma~D.8
writes out, with $\check G_{u}$ and $\check G_{x}$ both built from $\check\Phi$. Their
Definition~C.4 defines the exact process $\check P_{x}$ differently: it puts an
approximate metric $\widehat{\check\Phi}(u)$ at the proposal and averages over
$\widehat{\check\Phi}(u)$. Part~(ii) does not cover that kernel.
Suppose $d\ge10^{4}$ and $\mathcal{J}\ge11$. Then the acceptance probability at $0$ is
at most $0.0096$, and $\dtv(\check P_{0},\check P_{z})>0.99$ for every
$z\in\intK\setminus\{0\}$, in particular for every such $z$ with
$\norm{z}_{\check\Phi(0)}\le0.001$. Let $u\sim N(0,\check\Phi(0)^{-1})$ and
$\check a_{2}:=\check G_{u}(0)/\check G_{0}(u)$, the ratio of the same name in their proof at
$x=0$ with $f\equiv0$. Then
$\Pr[u\in\intK,\ \check a_{2}\ge e^{-4.83}]\le0.0016$. If $\mathcal{J}\ge11.4$, this holds
with $e^{-5}$ in place of $e^{-4.83}$. If also $\check\eta\le0.001/d$, then the
quantity $\check p_{\min}$ of their proof, with $1\wedge\check a_{2}$ counted as at most
$1$ at points $u\notin\intK$, satisfies
$\check p_{\min}\le\E[1\wedge\check a_{2}]\le0.0096+e^{-995d}$.
\item[(iii)] (Their Assumption~B.1(iii) for the log-det barrier.) For
$y\in\intK$ with $t_{y}:=\norm{y}^{2}$,
\[
\norm{\Hld(y)^{-1/2}\nabla\log\det\Hld(y)}_{2}^{2}
=\frac{2t_{y}}{2+t_{y}}\Bigl(d+\frac{4}{2+t_{y}}\Bigr)^{2}.
\]
At $t_{y}=1$ this equals $\frac23(d+\frac43)^{2}$, which is of order $d^{2}$.
\item[(iv)] (Their Assumption~B.1(ii) for the log-det barrier.) For every
$c_{I}\ge0$ the function $x\mapsto\log\det(\Hld(x)+c_{I}I)$ is convex on
$\intK$.
\end{enumerate}
For $\check\alpha=\check\eta=c/d$, we have $\mathcal{J}=cd/4$. With $c=0.001$, the hypotheses of the
first claim of (i) hold for $d\ge107$, and those of (ii) hold for
$d\ge44{,}000$. More generally, with any $c\in(0,0.001]$, they hold for
$d\ge8/(75c)$ and for $d\ge\max(10^{4},44/c)$, respectively, and the condition
$\mathcal{J}\ge11.4$ of (ii) holds for $d\ge\max(10^{4},45.6/c)$.
\end{lemma}

\begin{proof}
The walk of (ii) is the walk of Definition~\ref{II:def:walk} with
$\kappa=1$, $r^{2}=d\check\alpha$, and $L^{2}=\check\alpha/\check\eta$, since then
$\frac{r^{2}}{d}G^{-1}=\check\alpha(\Hld+\frac{\check\alpha}{\check\eta}I)^{-1}=\check\Phi^{-1}$. For
these values, we have $\kappa_{L}=(\check\alpha+\check\eta)/\check\eta$, and the quantity $\mathcal{J}$ of
Lemma~\ref{II:lem:arrow} is $dr^{2}\kappa/\kappa_{L}^{2}=d^{2}\check\alpha\check\eta^{2}/(\check\alpha+\check\eta)^{2}$,
which is the $\mathcal{J}$ defined here. We also use the Chernoff bounds
$\Pr[\norm{\xi}^{2}\le bd]\le(be^{1-b})^{d/2}$ for $b\in(0,1)$ and
$\Pr[\norm{\xi}^{2}\ge b'd]\le(b'e^{1-b'})^{d/2}$ for $b'>1$, which we prove as in
Lemma~\ref{II:lem:arrow}(b).

\emph{Part (i).} We have $\check\Phi(0)=(\check\alpha^{-1}+\check\eta^{-1})I$, so that
$t=\varsigma\norm{\xi}^{2}$ with $\varsigma:=\check\alpha\check\eta/(\check\alpha+\check\eta)$, and
$\check\Phi(u)-\check\Phi(0)=\check\alpha^{-1}(\Hld(u)-I)$. By the computation in the proof of
Lemma~\ref{II:lem:arrow}(a), we have
$u^{\top}(\Hld(u)-I)u=t^{2}(\frac a2+a^{2})=t^{2}\frac{6-t}{(2-t)^{2}}$, which
gives the stated formula. Since the function $t\mapsto(6-t)/(2-t)^{2}$ has derivative
$(10-t)/(2-t)^{3}>0$ on $[0,2)$, it is at least its value $\frac32$ at
$t=0$. Moreover, $t^{2}/\check\alpha=\varsigma^{2}\norm{\xi}^{4}/\check\alpha
=\mathcal{J}(\norm{\xi}^{2}/d)^{2}$, which proves the lower bound on $\mathcal{V}$.

If $u\in\intK$ and $\mathcal{V}\le0.01$, then $(\norm{\xi}^{2}/d)^{2}\le1/(150\mathcal{J})$,
or equivalently $\norm{\xi}^{2}\le bd$ with $b:=(150\mathcal{J})^{-1/2}$. If $\mathcal{J}\ge2/75$, then
$b\le\frac12$, and since $be^{1-b}$ increases on $(0,1)$, we obtain
$\Pr[\norm{\xi}^{2}\le bd]\le(\frac12e^{1/2})^{d/2}\le0.83^{d/2}$. Next, since
$\varsigma\le\check\eta\le0.001/d$, the event $u\notin\intK$ forces
$\norm{\xi}^{2}\ge2/\varsigma\ge2000d$, which has probability at most
$(2000e^{-1999})^{d/2}\le e^{-995d}$. Moreover, the sum $0.83^{d/2}+e^{-995d}$ is at
most $0.92$ for every $d\ge1$ and tends to $0$.

For the converse, we have $t\le\varsigma\norm{\xi}^{2}\le0.001\,\norm{\xi}^{2}/d$,
and on $\{\norm{\xi}^{2}\le3d\}$, we get $t\le0.003$. Hence
$(6-t)/(2-t)^{2}\le6/1.997^{2}<1.51$, and
$\mathcal{V}=\mathcal{J}(\norm{\xi}^{2}/d)^{2}(6-t)/(2-t)^{2}\le9\cdot1.51\,\mathcal{J}<0.01$ when
$\mathcal{J}\le7\cdot10^{-4}$. Finally,
$\Pr[\norm{\xi}^{2}>3d]\le(3e^{-2})^{d/2}<0.001$ for $d\ge16$.

\emph{Part (ii).} The ratio $\check G_{u}(x)/\check G_{x}(u)$ equals $q_{u}(x)/q_{x}(u)$,
where $q_{y}$ is the density of $N(y,\check\Phi(y)^{-1})$. By the parameter map
above, this ratio is the Metropolis ratio of Definition~\ref{II:def:walk} with $f\equiv0$, and at $x=0$ it equals $e^{\Lambda(u)}$ for $u\in\intK$. Therefore, the
acceptance probability at $0$ is the quantity bounded in
Lemma~\ref{II:lem:arrow}(b), with $s=\norm{\xi}^{2}/d$ for
$u=\check\Phi(0)^{-1/2}\xi$. We take $s_{0}=0.95$ there and obtain, for $d\ge10^{4}$,
\[
s_{0}\Bigl(\frac{3s_{0}-1}{4}-\frac1{2d}\Bigr)\ge0.95\,(0.4625-0.00005)>0.4393,
\qquad
\frac d2\bigl(s_{0}-1-\log s_{0}\bigr)>0.0006466\,d\ge6.466 .
\]
By the proof of Lemma~\ref{II:lem:arrow}, we have
$\Lambda\le-0.4393\,\mathcal{J}$ on $\{s\ge0.95\}\cap\{u\in\intK\}$ and
$\Pr[s<0.95]\le e^{-6.466}<0.0016$. For $\mathcal{J}\ge11$, this gives
$\Lambda<-4.83$ on that event, and since $0.4393\cdot11.4>5$, we have $\Lambda<-5$
there when $\mathcal{J}\ge11.4$. Hence $\{u\in\intK,\ \check a_{2}\ge e^{-4.83}\}$ lies in
$\{s<0.95\}$, as does $\{u\in\intK,\ \check a_{2}\ge e^{-5}\}$ when $\mathcal{J}\ge11.4$,
which proves the bound $0.0016$. Moreover, the acceptance probability at $0$ is
$\E[\ind{u\in\intK}(1\wedge\check a_{2})]$, which is at most
$e^{-4.83}+\Pr[s<0.95]<0.0080+0.0016=0.0096$.

For $z\in\intK\setminus\{0\}$, the law $\check P_{z}$ puts no mass on $\{0\}$,
since from $z$ the chain either stays at $z\ne0$ or moves to a proposal,
which has a Gaussian law. In contrast, the law $\check P_{0}$ puts mass at least $1-0.0096$ on
$\{0\}$, lazy or not. Hence
$\dtv(\check P_{0},\check P_{z})\ge\check P_{0}(\{0\})-\check P_{z}(\{0\})>0.99$.

Finally, suppose that $\check\eta\le0.001/d$, in which case $\Pr[u\notin\intK]\le e^{-995d}$ as in part~(i).
Since the quantity $\check p_{\min}$ is the mean of a minimum that includes $1\wedge\check a_{2}$,
we have $\check p_{\min}\le\E[1\wedge\check a_{2}]$. By splitting over $\{u\in\intK\}$, we obtain
$\E[1\wedge\check a_{2}]\le0.0096+\Pr[u\notin\intK]\le0.0096+e^{-995d}$.

\emph{Part (iii).} By the matrix determinant lemma, we have
$\log\det\Hld(y)=d\log a+\log(1+at_{y})$ with $a=2/(2-t_{y})$ and
$1+at_{y}=(2+t_{y})/(2-t_{y})$. Since $\nabla t_{y}=2y$, we obtain
\[
\nabla\log\det\Hld(y)
=\Bigl(\frac{2d}{2-t_{y}}+\frac{2}{2+t_{y}}+\frac{2}{2-t_{y}}\Bigr)y
=\frac{2}{2-t_{y}}\Bigl(d+\frac{4}{2+t_{y}}\Bigr)y .
\]
By the Sherman--Morrison formula, we have
$\Hld(y)^{-1}=a^{-1}\bigl(I-\frac{a\,yy^{\top}}{1+at_{y}}\bigr)$, which gives
$y^{\top}\Hld(y)^{-1}y=\frac{t_{y}}{a(1+at_{y})}=\frac{t_{y}(2-t_{y})^{2}}{2(2+t_{y})}$. Multiplying
the square of the scalar factor in the gradient by this quantity, we obtain the formula.

\emph{Part (iv).} By the matrix determinant lemma, we have
$\log\det(\Hld(x)+c_{I}I)=(d-1)h_{1}(t_{x})+h_{2}(t_{x})$ with $t_{x}=\norm x^{2}$,
$h_{1}(t_{x})=\log(a+c_{I})$, and $h_{2}(t_{x})=\log(\hat a+c_{I})$, where
$a=2/(2-t_{x})$ and $\hat a=a+a^{2}t_{x}=2(2+t_{x})/(2-t_{x})^{2}$. The derivatives $a'=a^{2}/2$,
$a''=a^{3}/2$, $\hat a'=2(6+t_{x})/(2-t_{x})^{3}$, and $\hat a''=4(10+t_{x})/(2-t_{x})^{4}$ are all
positive, which gives $h_{1}'>0$, $h_{2}'>0$,
$(a+c_{I})^{2}h_{1}''=a''(a+c_{I})-a'^{2}=a^{3}(\frac a4+\frac{c_{I}}2)>0$,
and
$(\hat a+c_{I})^{2}h_{2}''\ge\hat a''\hat a-\hat a'^{2}=4(4+12t_{x}+t_{x}^{2})/(2-t_{x})^{6}>0$. Since a
nondecreasing convex function of $t_{x}$ composed with the convex map
$x\mapsto\norm x^{2}$ is convex, this proves part~(iv).
\end{proof}

\begin{remark}[Scope of Lemma~\ref{II:lem:arrow-gkmsz}]
\label{II:rem:arrow-gkmsz}
Their Lemma~D.5 displays only the hypotheses $x\in\intK$ and
$\check\eta\le0.001/d$, for the metric $\check\Phi=\check\alpha^{-1}\check H+\check\eta^{-1}I$. Here,
$\check H$ is the Hessian of a general barrier (by the title of their Appendix~D) and a
self-concordant matrix function as in their Section~B.2. Moreover, the proof sketch
of Lemma~D.5 cites no part of their Assumption~B.1. Part~(i) shows that the conclusion of Lemma~D.5 is false
under its displayed hypotheses, for the log-det barrier on the arrow instance
whenever $\mathcal{J}\ge2/75$, for example for $\check\alpha=\check\eta=0.001/d$ with $d\ge107$.
On the other hand, their Lemma~D.3 uses ``Assumption (iii)'' without displaying it, and hence their
Assumption~B.1 can also be read as standing throughout their Appendix~D. In
that reading, the arrow instance fails Assumption~B.1(iii), as part~(iii)
proves, and thus Lemma~D.5 does not apply to this instance. However, their spectrahedral
proof still applies it to every spectrahedron, the arrow included, and the
conclusion of Lemma~D.5, as applied to the arrow, is false.

Their Lemma~D.8 displays only the hypothesis $\norm{x-z}_{\check\Phi(x)}\le0.001$.
Its proof uses their Lemmas~D.1 to~D.7, and their Lemma~D.10 through
Lemma~D.7. These lemmas add the hypotheses $\check\alpha<0.001/d$ (Lemma~D.4), $d\check\alpha<0.001$
(Lemma~D.10) and $\check\eta\le0.001/d$ (Lemmas~D.5 to~D.7), as well as bounds on
$\check\eta$ in terms of the Lipschitz constant $L$ of $f$ (Lemmas~D.1 and~D.2) and
their Assumption~B.1. Since $f\equiv0$ is $1$-Lipschitz, we take $L=1$. Their
Algorithm~2 then sets $\check\alpha=\Theta(1/d)$ and $\check\eta=\Theta(1/d)$, and hence
$\check\alpha=\check\eta=c/d$ is their own choice. While they state Algorithm~2 for polytopes,
its caption says that other barriers change only the barrier parameter
$\check\nu$ and the Hessian
approximation. For $c<0.001$, these values satisfy
$\check\alpha<0.001/d$, $d\check\alpha<0.001$ and $\check\eta\le0.001/d$, as well as
$\check\eta\le1/(10d)$ of Lemma~D.1, the bound $\check\eta\le1/(80d)$ used in its proof,
and $\check\eta\le0.01$ of Lemma~D.2. Moreover, Assumption~B.1(ii) holds by part~(iv), and
Assumption~B.1(iii) fails by part~(iii). We did not recheck
Assumption~B.1(i), which~\cite{gkmsz24} prove for the log-det barrier.

Part~(ii) applies at these parameters once $d\ge\max(10^{4},44/c)$, and it shows
that the conclusion of Lemma~D.8 is false under its displayed hypothesis, for
the kernel that its proof writes out. By the bound on $\check p_{\min}$ in
part~(ii), the identity $\dtv=1-\check p_{\min}$ of that proof gives $\dtv>0.99$ as well.
The proof also asks for $\Pr[\check b_{2}<\check\tau]\ge0.99$, where
$\check a_{2}=e^{-\check b_{2}}$, $\check\tau\in[1,5]$ and $e^{-\check\tau}\ge0.01$. However, this step fails
here, since $\check b_{2}<\check\tau$ means $\check a_{2}>e^{-\check\tau}>e^{-4.83}$, which by
parts~(i) and~(ii) has probability at most $0.0016+e^{-995d}$. Moreover, for
$d\ge\max(10^{4},45.6/c)$, the step fails for every $\check\tau\in[1,5]$. In the
reading that adds the hypotheses of the lemmas in its proof, a hypothesis that
fails here is Assumption~B.1(iii), which is the input that their spectrahedral
theorem takes from their Lemma~G.9. Recall that part~(ii) does not cover the exact
process of their Definition~C.4, which uses $\widehat{\check\Phi}(u)$ at the proposal.

Finally, we remark that these are statements about the conclusions and proof steps of two lemmas at
one point of one instance. They show neither that the rate stated
in~\cite{gkmsz24} for spectrahedra is false nor that their walk
fails to mix. Indeed, the center carries tiny mass under the uniform target, and
Lemma~\ref{II:lem:arrow} is not a mixing lower bound.
\end{remark}

\bibliographystyle{alpha}
\bibliography{ref}

\appendix

%%% This file is the structure for appendix content
%%% TeX files for body contents should be named as:
%%% 50_xxxx.tex
%%% 51_xxxx.tex
%%% ...

%\input{50_roadmap}

\phantomsection
\section*{Acknowledgments and AI disclosure}
\addcontentsline{toc}{section}{Acknowledgments and AI disclosure}
\label{sec:ai-disclosure}

The proof ideas in this paper were generated by two large language
models, GPT-5.6 Sol (OpenAI) and Claude Fable 5 (Anthropic), and the
resulting proofs were verified jointly by the same two models,
including independent adversarial line-by-line reviews of every
argument. We used Claude Opus 5.5 (Anthropic) to polish the writing.
The authors reviewed the manuscript and take full
responsibility for its content.

\section{Approximate metrics and exact-arithmetic complexity}
\label{I:sec:approx}

In this appendix, we prove the implementation guarantees in
Theorem~\ref{thm:intro_approximate_metrics}. So far, our analysis has used exact
metrics because the derivative of the exact Lewis map produces the
cancellation in Section~\ref{I:sec:cancellation} and the exact log-det Hessian
gives the endpoint identities in Section~\ref{II:sec:ledger}. However, an
arithmetic implementation need not reproduce either metric exactly. We
therefore keep the walks that use the exact metrics and compare their Gaussian
proposals with approximate-metric proposals. This comparison leaves both arguments built on the proposal-ratio identities unchanged.

In both settings, the comparison has the same form. Specifically, a one-sided
padding makes the approximate metric dominate the exact metric, and an $O(1/d)$
relative error keeps the two Gaussian densities within a constant factor on
the acceptance event established above. With this padding and this relative
error, we show that the proposal overlap and conductance bounds still hold with
different universal constants. For polytopes, we obtain the approximation from
high-precision Lewis weights. For spectrahedra, we use the TensorSRHT
approximation of the log-det Hessian from~\cite[Lemma~G.12]{gkmsz24}, together
with a Metropolis filter that draws independent sketches at the current point
and at the proposal. We first treat deterministic Lewis-weight approximations
for polytopes, then isolate the perturbation argument for randomized
metrics in the sketched walk, and apply it to a sketched log-det metric for
spectrahedra.

Recall the notation of the polytope proof:
\[
 H(x)=A_x^\top W_x^{c_p}A_x,\qquad
 G(x)=\kappa H(x)+L^2I,\qquad
 \bar\nu_{\mathrm{LS}}:=\kappa n^{2/p}d^{1-2/p}.
\]
Here $q_x,R_x$, and $P$ denote the exact proposal, ratio, and lazy kernel.
For the weights, we use the iteration of Gribling, Sidford, and
Zhang~\cite{gsz26}, whose high-precision guarantee we now state in this
notation.

\subsection{High-precision Lewis weights}

We fix $0<\varepsilon_w<1/2$, and at each $x\in\intK$ we compute a positive
vector $\widetilde w_x$ by a deterministic measurable map
$\mathsf{ApproxLewis}_{\varepsilon_w}$ such that
\begin{equation}
 (1-\varepsilon_w)w_{x,i}\le \widetilde w_{x,i}
 \le(1+\varepsilon_w)w_{x,i}
 \qquad(i=1,\ldots,n).
 \label{I:eq:weight-estimate}
\end{equation}

To make the proposal kernel a deterministic measurable function of the
current state, we use a fixed initialization and a fixed iteration count.
The iteration uses the parameters
\[
 \alpha_{\mathrm{LW}}:=\frac{2}{p-2},\qquad
 b_{\mathrm{LW}}:=\alpha_{\mathrm{LW}}^{-1},\qquad
 \bar b_{\mathrm{LW}}:=\max\{1,b_{\mathrm{LW}}\},\qquad
 L_{\mathrm{LW}}:=32p\bar b_{\mathrm{LW}},
\]
and
\[
 \bar\varepsilon_w:=\frac{\varepsilon_w}{2(1+\alpha_{\mathrm{LW}})},
 \qquad
 T_{\mathrm{LW}}:=
 \lceil4L_{\mathrm{LW}}
 \log\frac{2n}{\bar\varepsilon_w}\rceil .
\]
Starting from $v_{\mathrm{LW}}^{(0)}=\mathbf{1}$ and writing
$D_{\mathrm{LW}}^{(t)}=\Diag(v_{\mathrm{LW}}^{(t)})$, for each
$t=0,\ldots,T_{\mathrm{LW}}-1$ we compute the row-quadratic vector, which can be
obtained from one full computation of the leverage-score vector,
\[
 u_{\mathrm{LW},i}^{(t)}:=(A_x)_{i,\cdot}
     (A_x^\top D_{\mathrm{LW}}^{(t)}A_x)^{-1}(A_x)_{i,\cdot}^\top,
 \qquad
 \varphi_{\mathrm{LW},i}^{(t)}
 :=\frac{(u_{\mathrm{LW},i}^{(t)})^{b_{\mathrm{LW}}}}
          {v_{\mathrm{LW},i}^{(t)}},
\]
and then update
\begin{equation}
 v_{\mathrm{LW},i}^{(t+1)}:=
 (1+\frac{\varphi_{\mathrm{LW},i}^{(t)}-1}
 {L_{\mathrm{LW}}})v_{\mathrm{LW},i}^{(t)}
 \qquad(i=1,\ldots,n).
 \label{I:eq:gsz-update}
\end{equation}
After one final computation of $u_{\mathrm{LW}}^{(T_{\mathrm{LW}})}$, the map returns
\begin{equation}
 \widetilde w_{x,i}:=
 (u_{\mathrm{LW},i}^{(T_{\mathrm{LW}})})^{1+b_{\mathrm{LW}}}.
 \label{I:eq:gsz-output}
\end{equation}
By the high-precision guarantee of Gribling, Sidford, and Zhang, the output of
this iteration satisfies Eq.~\eqref{I:eq:weight-estimate}. Moreover, since the
iteration count and the all-ones initialization are prescribed, the output is
a deterministic Borel function of $A_x$, and hence of $x$.

To obtain the one-sided comparison required below from the two-sided estimate
in Eq.~\eqref{I:eq:weight-estimate}, we inflate the approximate metric by its
worst possible downward error.

\begin{definition}[Approximate-metric LS walk]
\label{I:def:approx-walk}
For $\widetilde W_x:=\Diag(\widetilde w_x)$, define
\begin{equation}
 \widetilde H(x):=A_x^\top\widetilde W_x^{c_p}A_x,
 \qquad
 \bar H(x):=(1-\varepsilon_w)^{-c_p}\widetilde H(x),
 \qquad
 \bar G(x):=\kappa\bar H(x)+L^2I .
 \label{I:eq:padded-metric}
\end{equation}
At $x$, propose
\[
 Y\sim\bar q_x:=N(x,\frac{r^2}{d}\bar G(x)^{-1}).
\]
When $Y\in\intK$, use the exact conditional Metropolis ratio
\begin{equation*}
 \bar R_x(Y):=
 \frac{e^{-f(Y)}\bar q_Y(x)}{e^{-f(x)}\bar q_x(Y)}.
\end{equation*}
Proposals outside $K$ are rejected, and the chain is $\tfrac12$-lazy.
Write $\bar P$ for the resulting lazy kernel.
\end{definition}

\begin{algorithm}[H]
\caption{One step of the approximate-metric LS walk}
\label{I:alg:approx-walk}
\begin{algorithmic}[1]
\Require Current cache $(x,s_x,f(x),\bar G(x),C_x,
         \log\det\bar G(x))$ with $C_xC_x^\top=\bar G(x)$;
         accuracy $\varepsilon_w$.
\State With probability $1/2$, \Return the current cache. \Comment{laziness}
\State Draw $\xi\sim N(0,I_d)$, set
       $h\gets\frac{r}{\sqrt d}C_x^{-\top}\xi$, and set
       $y\gets x+h$ and $s_y\gets s_x+\mA h$.
\If{some coordinate of $s_y$ is nonpositive}
  \State \Return the current cache.
\EndIf
\State $\widetilde w_y\gets
       \Call{ApproxLewis}{A_y,p,\varepsilon_w}$.
\State Form $\bar G(y)$ from Eq.~\eqref{I:eq:padded-metric}, factor it,
       and evaluate $f(y)$.
\State Compute
\Statex \hspace{1.5em}$\displaystyle
 \ell\gets f(x)-f(y)+\frac12\log\frac{\det\bar G(y)}{\det\bar G(x)}
 -\frac{d}{2r^2}h^\top(\bar G(y)-\bar G(x))h.$
\State Draw $U_{\mathrm{acc}}\sim\mathrm{Unif}(0,1)$; if
       $\log U_{\mathrm{acc}}\le\min\{0,\ell\}$, \Return the cache at $y$;
       otherwise \Return the current cache.
\end{algorithmic}
\end{algorithm}

Our perturbation argument does not differentiate the approximate
weights. Instead, it first compares $\bar G$ with $G$ pointwise and then
transfers the acceptance and overlap bounds of the walk with the exact metric.

\begin{lemma}[Weight error implies metric error]
\label{I:lem:weight-to-metric}
Let
\begin{equation*}
 \gamma_w:=(\frac{1+\varepsilon_w}
                         {1-\varepsilon_w})^{c_p}
 =e^{\theta_w},
 \qquad
 \theta_w:=c_p\log\frac{1+\varepsilon_w}{1-\varepsilon_w}.
\end{equation*}
Then, uniformly on $\intK$,
\begin{equation}
 H(x)\preceq\bar H(x)\preceq\gamma_wH(x),
 \qquad
 G(x)\preceq\bar G(x)\preceq\gamma_wG(x).
 \label{I:eq:metric-sandwich}
\end{equation}
For $\varepsilon_w\le1/2$, $0\le\theta_w\le4\varepsilon_w$.
\end{lemma}

\begin{proof}
Since $c_p\ge0$, the coordinatewise estimate
Eq.~\eqref{I:eq:weight-estimate} and the positive semidefiniteness of each row
outer product imply
\[
 (1-\varepsilon_w)^{c_p}H(x)\preceq\widetilde H(x)
 \preceq(1+\varepsilon_w)^{c_p}H(x),
\]
and multiplying by $(1-\varepsilon_w)^{-c_p}$ proves the first sandwich.
Next, adding $L^2I$ and using $\gamma_w\ge1$, we obtain
\[
 \kappa\bar H+L^2I\preceq
 \gamma_w\kappa H+L^2I\preceq\gamma_wG,
\]
which proves the second sandwich. Finally, the remaining bound follows from
$c_p\le1$ and $\log((1+t)/(1-t))\le4t$ for $0\le t\le1/2$.
\end{proof}
\subsection{From metric approximation to mixing}

The sandwich in Eq.~\eqref{I:eq:metric-sandwich} has two consequences. First,
it makes the centered proposals close in total variation. Second, it
controls the change of each Gaussian density at the two endpoints of an
accepted step. The second consequence explains why the accuracy scale is
$1/d$, since a Gaussian log density contains a determinant contribution from
all $d$ directions.

\begin{lemma}[Gaussian proposal perturbation]
\label{I:lem:proposal-perturbation}
There is a universal $C_q$ such that, whenever $\theta_w\le1$,
\begin{equation*}
 \sup_{x\in\intK}\dtv(q_x,\bar q_x)
 \le C_q\sqrt d\,\theta_w=:\tau_w .
\end{equation*}
\end{lemma}

\begin{proof}
Applying the centered Gaussian KL formula to the eigenvalues
$\lambda_1,\ldots,\lambda_d\in[1,e^{\theta_w}]$
of $G(x)^{-1/2}\bar G(x)G(x)^{-1/2}$, we obtain
\[
 D_{\mathrm{KL}}(q_x\Vert\bar q_x)
 =\frac12\sum_{j=1}^d(\lambda_j-1-\log\lambda_j)
 \le C d\theta_w^2,
\]
where the inequality follows from
$0\le\lambda-1-\log\lambda\le C(\log\lambda)^2$ on $[1,e]$.
Now the claim follows from Lemma~\ref{I:lem:pinsker}.
\end{proof}

We need the density comparison at both endpoints. Fortunately, the required
event is already available from the acceptance proof for the exact metric.

\begin{lemma}[Two-endpoint event for the exact proposal]
\label{I:lem:reference-good-event}
There are universal constants $C_A,C_N<\infty$ and the same
$\epsilon<1/10$ as in Lemma~\ref{I:lem:acceptance} such that, for every
$x\in\intK$, the set
\[
 \mathcal A_x:=\{y\in\intK:
 \log R_x(y)\ge-C_A,\quad
 \norm{y-x}_{G(x)}^2\le C_Nr^2,\quad
 \norm{y-x}_{G(y)}^2\le C_Nr^2\}
\]
satisfies $q_x(\mathcal A_x)\ge1-\epsilon$.
\end{lemma}

\begin{proof}
We intersect the good-set event, the concentration event, and the Gaussian
norm event from the proof of Lemma~\ref{I:lem:acceptance}. On this
intersection, the proposal is feasible, $\log R_x(Y)\ge-C_A$, and
\[
 \norm{Y-x}_{G(x)}^2=\frac{r^2}{d}\norm{\xi}^2\le Cr^2,
\]
where the equality follows from the Gaussian proposal representation and the
inequality follows from the norm event.
Moreover, for the universal choice of $r$, this local norm is below $1/2$.
Thus, applying Lemmas~\ref{I:lem:ssc-G} and~\ref{I:lem:distortion} along the
feasible segment $[x,Y]$, we obtain $G(Y)\preceq CG(x)$, and hence the
reverse norm bound after enlarging $C_N$.
\end{proof}

\begin{lemma}[Acceptance survives high-precision weights]
\label{I:lem:approx-acceptance}
There is a universal $c_{\mathrm{app}}>0$ such that, if
$\varepsilon_w\le c_{\mathrm{app}}/d$, then uniformly in $x\in\intK$,
\begin{equation*}
 \Pr_{Y\sim\bar q_x}[Y\in K\ \text{and}\
          \log\bar R_x(Y)\ge-\bar C_A]
 \ge1-\bar\epsilon,
\end{equation*}
where $\bar C_A<\infty$ is a universal constant and
$\bar\epsilon:=\epsilon+\tau_w\le\epsilon+4C_qc_{\mathrm{app}}$ is
bounded by a universal constant.
\end{lemma}

\begin{proof}
We reduce $c_{\mathrm{app}}$, if necessary, so that
$4c_{\mathrm{app}}\le1$, which gives $\theta_w\le4\varepsilon_w\le1$ and
$e^{\theta_w}-1\le(e-1)\theta_w$. For $z,u\in\intK$, we define
\[
 \mathcal E_z(u):=\log\frac{\bar q_z(u)}{q_z(u)}
 =\frac12\log\frac{\det\bar G(z)}{\det G(z)}
 -\frac{d}{2r^2}
   (\norm{u-z}_{\bar G(z)}^2-\norm{u-z}_{G(z)}^2).
\]
By Eq.~\eqref{I:eq:metric-sandwich}, we have
\begin{equation*}
 |\mathcal E_z(u)|\le\frac{d\theta_w}{2}
 +\frac{d(e^{\theta_w}-1)}{2r^2}\norm{u-z}_{G(z)}^2.
\end{equation*}
Consequently, for $y\in\mathcal A_x$,
$|\mathcal E_x(y)|+|\mathcal E_y(x)|\le Cd\theta_w$. Since
\[
 \log\bar R_x(y)=\log R_x(y)+\mathcal E_y(x)-\mathcal E_x(y),
\]
we obtain $\log\bar R_x(y)\ge-C_A-Cd\theta_w$ from
Lemma~\ref{I:lem:reference-good-event}. Under the condition
$\varepsilon_w\le c_{\mathrm{app}}/d$, this is a universal lower bound.
Finally, by Lemma~\ref{I:lem:proposal-perturbation}, we have
\[
 \bar q_x(\mathcal A_x)\ge q_x(\mathcal A_x)-\tau_w
 \ge1-\epsilon-\tau_w.
\]
Since the boundary of the full-dimensional convex body $K$ has
Lebesgue measure zero, writing $Y\in K$ rather than $Y\in\intK$ in
the statement does not change this probability.
\end{proof}

\begin{lemma}[Approximate proposal overlap]
\label{I:lem:approx-proposal-overlap}
The constant $c_{\mathrm{app}}$ can be chosen so that
Eq.~\eqref{I:eq:close-pair} implies
\[
 \dtv(\bar q_x,\bar q_z)\le\bar\delta_q
 :=\delta_q+2\tau_w<1-2\bar\epsilon .
\]
\end{lemma}

\begin{proof}
The displayed upper bound follows from the triangle inequality together with
Lemmas~\ref{I:lem:proposal-overlap} and~\ref{I:lem:proposal-perturbation}.
It remains to show the strict inequality. In the proof for the exact metric,
we choose $\delta_q<1-2\epsilon$, which leaves a fixed positive margin
$m_0:=1-\delta_q-2\epsilon$. Since
$\bar\delta_q+2\bar\epsilon=\delta_q+2\epsilon+4\tau_w$, it suffices to
reduce $c_{\mathrm{app}}$ so that $4\tau_w<m_0$.
\end{proof}

\begin{theorem}[Mixing with approximate Lewis weights]
\label{I:thm:approx-main}
Under the hypotheses of Theorem~\ref{I:thm:main-formal}, suppose the
deterministic measurable estimate Eq.~\eqref{I:eq:weight-estimate}
satisfies $\varepsilon_w\le c_{\mathrm{app}}/d$.  Then the padded
approximate-metric LS walk of Definition~\ref{I:def:approx-walk} is exactly
reversible with stationary distribution $\pi$.  For every $w$-warm initial
law $\mu_0$ and every $0<\delta<1$,
\begin{equation}
 \dtv(\mu_0\bar P^t,\pi)\le\delta
 \qquad\text{for all}\qquad
 t\ge\Ot((d^2+dL^2R^2)\log(w/\delta)).
 \label{I:eq:approx-mixing}
\end{equation}
\end{theorem}

\begin{proof}
\emph{Detailed balance.}
Since $x\mapsto\bar G(x)$ is deterministic and measurable, the density of an
accepted move satisfies
\[
 e^{-f(x)}\bar q_x(y)(1\wedge\bar R_x(y))
 =\min\{e^{-f(x)}\bar q_x(y),e^{-f(y)}\bar q_y(x)\}.
\]
Thus, after adding rejection and laziness, we obtain an exactly
$\pi$-reversible kernel.

\emph{Local kernel overlap.}
By Lemmas~\ref{I:lem:approx-acceptance}
and~\ref{I:lem:approx-proposal-overlap}, the hypotheses of the
metric-independent kernel-overlap argument in
Lemma~\ref{I:lem:kernel-overlap} hold with barred constants. Hence, the lazy
kernels $\bar P_x,\bar P_z$ have a fixed positive common mass whenever
Eq.~\eqref{I:eq:coupling-radius} holds, where we still measure closeness in
the exact metric $G$. Since the proof of Lemma~\ref{I:lem:conductance}
uses only this local kernel overlap, reversibility, and the $\nu_G$-symmetry
of the exact metric, it applies with the barred common-mass constant and
gives
\[
 \Phi(\bar P)\ge
 \frac{cr}{\sqrt{d(\bar\nu_{\mathrm{LS}}+L^2R^2)}}.
\]

\emph{Conductance and mixing.}
Eq.~\eqref{I:eq:approx-mixing} follows from the same lazy Cheeger
and warm-start $L^2$ argument as in the proof of
Theorem~\ref{I:thm:main-formal}.
\end{proof}

Since the spectrahedral approximation is randomized, we also need a
seed-averaged version of the preceding perturbation argument, with
independent seeds drawn at the current point and at the proposal.

\begin{lemma}[Perturbation for random metrics with independent endpoint seeds]
\label{I:lem:randomized-metric}
Let $\omega_{\mathrm{s}}$ have a state-independent base law, let
$\bar G(x;\omega_{\mathrm{s}})$ be a measurable positive-definite output,
and let $\bar q_{x,\omega_{\mathrm{s}}}$ be its conditional Gaussian
proposal density.  Draw $\omega_{\mathrm{s}}$ at $x$, propose $y$, draw an
independent seed $\omega_{\mathrm{s}}'$ at $y$ when $y\in\intK$, and accept
with probability
\[
 1\wedge\frac{e^{-f(y)}\bar q_{y,\omega_{\mathrm{s}}'}(x)}
                  {e^{-f(x)}\bar q_{x,\omega_{\mathrm{s}}}(y)}.
\]
Proposals outside $\intK$ are rejected.
The resulting marginal kernel is exactly $\pi$-reversible.

Suppose that, uniformly in $x$, a good-seed event
$\mathcal G_x(\omega_{\mathrm{s}})$ has probability at least $1-\beta_0$
and
\[
 \dtv(q_x,\bar q_{x,\omega_{\mathrm{s}}})\le\tau
 \qquad\text{on }\mathcal G_x(\omega_{\mathrm{s}}).
\]
Suppose also that the exact proposals have overlap at least
$1-\delta_q$ for the close pairs under consideration.  For every $x$, let
$\mathcal B_x\subseteq\intK$ be measurable with
$q_x(\mathcal B_x^c)\le\epsilon$.  Assume that for $y\in\mathcal B_x$ the
conditional Metropolis ratio is at least $e^{-C}$ whenever the seeds at
$x$ and $y$ are both good, where $C\ge0$.  Then, before the fixed laziness
factor, the moving kernels have common mass at least
\begin{equation}
 e^{-C}(1-\beta_0)
 (1-\delta_q-2\epsilon-4\tau-2\beta_0).
 \label{I:eq:two-seed-common-mass}
\end{equation}
\end{lemma}

\begin{proof}
\emph{Reversibility.}
The marginal accepted flow is
\[
 \E_{\omega_{\mathrm{s}},\omega_{\mathrm{s}}'}[\min\{e^{-f(x)}\bar q_{x,\omega_{\mathrm{s}}}(y),
                 e^{-f(y)}\bar q_{y,\omega_{\mathrm{s}}'}(x)\}],
\]
which is symmetric after interchanging the two independent seeds. This is
the independent-sketch construction of~\cite[Lemma~D.12]{gkmsz24}.

\emph{Common mass on good seeds.}
We define the good-seed proposal subdensity
\[
 q_x^{\mathrm{good}}(y):=\E_{\omega_{\mathrm{s}}}[\ind{\mathcal G_x(\omega_{\mathrm{s}})}
       \bar q_{x,\omega_{\mathrm{s}}}(y)].
\]
By the total-variation hypothesis and the identity between total variation
and half the $L^1$ distance, we have
\[
 \norm{q_x^{\mathrm{good}}-q_x}_{L^1}\le 2\tau+\beta_0.
\]
Thus, for each close pair, we obtain
\begin{equation*}
 \int\min\{q_x^{\mathrm{good}}(y),q_z^{\mathrm{good}}(y)\}\,dy
 \ge1-\delta_q-2\tau-2\beta_0.
\end{equation*}
Indeed,
\begin{align*}
 \int\min\{q_x^{\mathrm{good}},q_z^{\mathrm{good}}\}
 &=\frac12(\norm{q_x^{\mathrm{good}}}_1
                 +\norm{q_z^{\mathrm{good}}}_1\\
 &\hspace{4em}
   -\norm{q_x^{\mathrm{good}}-q_z^{\mathrm{good}}}_1).
\end{align*}
Since
$\norm{q_x^{\mathrm{good}}}_1,\norm{q_z^{\mathrm{good}}}_1
\ge1-\beta_0$, the claim follows from the triangle inequality through
$q_x,q_z$. The same total-variation hypothesis also gives
\[
 q_x^{\mathrm{good}}(\mathcal B_x^c)\le\epsilon+\tau,
 \qquad
 q_z^{\mathrm{good}}(\mathcal B_z^c)\le\epsilon+\tau.
\]
Therefore, restricting the common proposal mass to
$\mathcal B_x\cap\mathcal B_z$ removes at most $2\epsilon+2\tau$. At each
remaining proposal $y$, requiring the independent reverse seed to be good
contributes the factor $1-\beta_0$. Finally, applying the conditional lower
bound $e^{-C}$, we obtain Eq.~\eqref{I:eq:two-seed-common-mass}.
\end{proof}

Since the good-seed guarantee is pointwise in the current state, fresh
independent endpoint seeds, exact conditional densities, and a
positive-definite metric on every seed suffice, and no event needs to hold
simultaneously over the mixing horizon.
\subsection{Arithmetic complexity for polytopes}

Since we cache the metric at the current point, each feasible proposal
requires at most one new Lewis-weight computation.  Combining this per-step
cost with Theorem~\ref{I:thm:approx-main}, we obtain the total arithmetic
bound below.  Throughout the two arithmetic corollaries, we use
$\log(e w/\delta)=1+\log(w/\delta)$ to absorb the one-time metric build and
evaluation of $f$.  This is possible since $w\ge1$, $0<\delta<1$, and the
displayed geometric prefactors are at least one.

\begin{corollary}[Arithmetic complexity for polytopes]
\label{I:cor:arithmetic-runtime}
In the exact-arithmetic model described below,
run Algorithm~\ref{I:alg:approx-walk} with
$\varepsilon_w=c_{\mathrm{app}}/d$.  It can
be implemented with
\begin{equation}
 \Ot(nd^{\omega-1}+d^\omega)
 \label{I:eq:per-step-runtime}
\end{equation}
arithmetic operations and at most one new evaluation of $f$ per Markov step.
Here $\omega$ denotes the exponent of square-matrix multiplication.
Consequently, for every $0<\delta<1$, starting from an exact draw
$X_0\sim\mu_0$, where $\mu_0$ is $w$-warm, the walk reaches total-variation
error $\delta$ using
\begin{equation}
 \Ot(
 (d^2+dL^2R^2)(nd^{\omega-1}+d^\omega)
 \log\frac{e w}{\delta})
 \label{I:eq:total-runtime}
\end{equation}
arithmetic operations and
$\Ot((d^2+dL^2R^2)\log(e w/\delta))$ evaluations of $f$.
\end{corollary}

\begin{proof}
We first bound the cost of the Lewis weights.  The slack-scaled matrix $A_x$
has $n$ rows and $d$ columns, and the fixed-schedule estimator in
Eqs.~\eqref{I:eq:gsz-update}--\eqref{I:eq:gsz-output} returns the
coordinatewise estimate Eq.~\eqref{I:eq:weight-estimate} in
\[
 O(p^2\log\frac{np\alpha_{\mathrm{LW}}}{\varepsilon_w})
\]
iterations, each of which computes the full leverage-score vector exactly.
Here $p=\max\{4,2(1+\log n)\}$ and $\varepsilon_w=c_{\mathrm{app}}/d$, and
thus the number of iterations is polylogarithmic in $n,d$.  To compute one full leverage-score vector for an $n\times d$ matrix, we form
and factor its weighted Gram matrix and evaluate all row quadratic forms,
which takes $O(nd^{\omega-1})$ arithmetic operations.  Hence, one call
to $\mathsf{ApproxLewis}_{\varepsilon_w}$ costs $\Ot(nd^{\omega-1})$
arithmetic operations.

For one transition, forming the slacks and checking
membership costs $O(\nnz(\mA))$, and assembling the weighted Gram matrix is
covered by $O(nd^{\omega-1})$.  Factoring $\bar G$, drawing the Gaussian from
that factor, computing its determinant, and evaluating the two quadratic
forms cost $O(d^\omega)$.  Since we cache the current metric,
factorization, determinant, slacks, and target value, and test laziness
before any proposal work, each nonlazy iteration computes at most one
new feasible-endpoint metric and one new target value.  On acceptance,
these become the next cache.  Since the initial cache is built once, there are
at most $t+1$ builds in $t$ steps.  This proves
Eq.~\eqref{I:eq:per-step-runtime}.

Finally, combining the per-step bound with Theorem~\ref{I:thm:approx-main}
and including the initial build and evaluation of $f$, we obtain
Eq.~\eqref{I:eq:total-runtime} and the stated number of evaluations of $f$.
\end{proof}

We note that the per-step cost in Eq.~\eqref{I:eq:per-step-runtime} is the
one stated for the robust Lee--Sidford walk
in~\cite[Theorems~1.1 and~1.2]{gkmsz24}, where this cost consists of computing an
approximate Lee--Sidford Hessian in $\Ot(nd^{\omega-1})$ time plus a further
$d^\omega$ per step.  In our implementation, the Lewis weights instead come
from the high-precision iteration of~\cite{gsz26}.

\begin{remark}[Exact-arithmetic model]
Field operations, comparisons, scalar powers and roots, logarithms and
exponentials, and exact Gaussian and uniform variates have unit cost.  Our
count includes the Lewis iterations, slack and membership computations,
metric factorization, the two Gaussian quadratic forms, and one new target
evaluation at a feasible proposal.

We assume that the initial warm draw and the parameters $L,R,w$ are supplied,
and we do not count finite-bit conditioning, approximate Gaussian generation,
or an inexact Metropolis decision.  Thus, the corollary bounds the number of
arithmetic operations and is not a bit-complexity result.  A bit-complexity
result would additionally require controlling the precision of
ill-conditioned diagonal scalings.
\end{remark}

\subsection{A sketched log-det metric}

For spectrahedra, we use $\kappa_{\mathrm{ld}}$ to denote the scale written
there as $\kappa$, and we define
\[
 H_{\mathrm{ld}}(x):=\nabla^2(-\log\det S(x)),
 \qquad
 G_{\mathrm{ld}}(x):=\kappa_{\mathrm{ld}}H_{\mathrm{ld}}(x)+L^2I.
\]
Here, $G_{\mathrm{ld}}$ is the exact metric, and $n$ denotes the pencil block
size.  With $\operatorname{vec}$ denoting column stacking, let
$\mathsf A\in\R^{d\times n^2}$ have $i$-th row
$\operatorname{vec}(A_i)^\top$, and define
\begin{equation}
 B_x:=\mathsf A(S(x)^{-1/2}\otimes S(x)^{-1/2}).
 \label{II:eq:sketch-factor}
\end{equation}
Then $H_{\mathrm{ld}}(x)=B_xB_x^\top$.  If $S(x)=L_xL_x^\top$, we may replace
the symmetric inverse square root in Eq.~\eqref{II:eq:sketch-factor} by
$T_x=L_x^{-\top}$, since $T_xT_x^\top=S(x)^{-1}$.  The corresponding rows
are then $\operatorname{vec}(T_x^\top A_iT_x)^\top$.

For fixed $0<\varepsilon_s<1/10$ and $0<\beta_s<1/10$, we compress this Gram
factor with a TensorSRHT before forming the $d\times d$ metric.
Specifically, for a fresh TensorSRHT seed $\omega_{\mathrm{s}}$, let
$\Pi_{\omega_{\mathrm{s}}}\in\R^{s\times n^2}$ be the transform
of~\cite[Definition~G.10]{gkmsz24}, where, for a sufficiently large
universal $C_{\mathrm{TS}}$,
\begin{equation}
 s:=\lceil C_{\mathrm{TS}}\varepsilon_s^{-2}d
 \log^3\frac{nd}{\varepsilon_s\beta_s}\rceil,
 \qquad
 Z_{x,\omega_{\mathrm{s}}}:=
 \Pi_{\omega_{\mathrm{s}}}B_x^\top\in\R^{s\times d}.
 \label{II:eq:sketch-size}
\end{equation}
If $n$ is not a Hadamard size, we let $n'$ be the next power of two, and we
replace the pencil by $S(x)\oplus I_{n'-n}$ and every direction by
$A_i\oplus0_{n'-n}$.  This leaves the feasible body and the log-det Hessian
unchanged.  Moreover, since the whitened directions become $\Ab_i(x)\oplus0$,
the leverage operator becomes $\Psib(x)\oplus0$, and $\psistar$ is unchanged.
Since $n'<2n$, this padding changes the mixing and arithmetic bounds by at
most a universal constant.  Henceforth, we relabel $n'$ as $n$.

We then define a positive-definite output on every seed by
\begin{equation*}
 \widehat H(x;\omega_{\mathrm{s}}):=
 \begin{cases}
 Z_{x,\omega_{\mathrm{s}}}^\top Z_{x,\omega_{\mathrm{s}}},
   &Z_{x,\omega_{\mathrm{s}}}^\top
    Z_{x,\omega_{\mathrm{s}}}\succ0,\\
 I_d,&\text{otherwise},
 \end{cases}
\end{equation*}
and apply one-sided padding:
\begin{equation*}
 \bar H_{\mathrm{ld}}(x;\omega_{\mathrm{s}})
 :=(1-\varepsilon_s)^{-2}\widehat H(x;\omega_{\mathrm{s}}),
 \qquad
 \bar G_{\mathrm{ld}}(x;\omega_{\mathrm{s}})
 :=\kappa_{\mathrm{ld}}\bar H_{\mathrm{ld}}(x;\omega_{\mathrm{s}})+L^2I.
\end{equation*}
While the fallback makes the conditional proposal well defined for every
seed, we note that the approximation statement below is restricted to the
sketch-success event.

\begin{lemma}[Tensor sketch implies metric approximation]
\label{II:lem:sketch-sandwich}
For each fixed $x\in\intK$, with probability at least $1-\beta_s$ over
$\omega_{\mathrm{s}}$,
\begin{equation*}
 H_{\mathrm{ld}}(x)\preceq
 \bar H_{\mathrm{ld}}(x;\omega_{\mathrm{s}})
 \preceq e^{\theta_s}H_{\mathrm{ld}}(x),
 \qquad
 G_{\mathrm{ld}}(x)\preceq
 \bar G_{\mathrm{ld}}(x;\omega_{\mathrm{s}})
 \preceq e^{\theta_s}G_{\mathrm{ld}}(x),
\end{equation*}
where
\[
 \theta_s:=2\log\frac{1+\varepsilon_s}{1-\varepsilon_s}
 \le8\varepsilon_s.
\]
\end{lemma}

\begin{proof}
By Lemma~G.12 of~\cite{gkmsz24}, with probability at least $1-\beta_s$, we
have, simultaneously for all $v\in\R^d$,
\[
 (1-\varepsilon_s)\norm{B_x^\top v}
 \le\norm{\Pi_{\omega_{\mathrm{s}}}B_x^\top v}
 \le(1+\varepsilon_s)\norm{B_x^\top v}.
\]
Squaring both sides, we obtain
\[
 (1-\varepsilon_s)^2H_{\mathrm{ld}}(x)
 \preceq\widehat H(x;\omega_{\mathrm{s}})
 \preceq(1+\varepsilon_s)^2H_{\mathrm{ld}}(x).
\]
On this event, the middle matrix is positive definite, and thus the fallback
is not used.  The first sandwich then follows from the padding, and adding $L^2I$ gives
the second in the same way as in Lemma~\ref{I:lem:weight-to-metric}.
Finally, we have $2\log((1+t)/(1-t))\le8t$ for $0\le t\le1/2$.
\end{proof}

We now define the sketched walk, which uses independent sketches at the
current point and at the proposal.

\begin{definition}[Sketched log-det walk]
\label{II:def:sketched-logdet-walk}
At every nonlazy step from $x$, draw a fresh seed $\omega_{\mathrm{s}}$
and propose
\[
 Y\sim\bar q^{\mathrm{ld}}_{x,\omega_{\mathrm{s}}}
 :=N(x,\frac{r^2}{d}
 \bar G_{\mathrm{ld}}(x;\omega_{\mathrm{s}})^{-1}).
\]
Reject if $Y\notin\intK$.  Otherwise draw an independent fresh seed
$\omega_{\mathrm{s}}'$ at $Y$ and accept with probability
\begin{equation}
 1\wedge
 \frac{e^{-f(Y)}\bar q^{\mathrm{ld}}_{Y,\omega_{\mathrm{s}}'}(x)}
      {e^{-f(x)}\bar q^{\mathrm{ld}}_{x,\omega_{\mathrm{s}}}(Y)}.
 \label{II:eq:two-seed-logdet-ratio}
\end{equation}
Both conditional densities in Eq.~\eqref{II:eq:two-seed-logdet-ratio} are
evaluated exactly, and the chain is $\tfrac12$-lazy.
Write $\bar P_{\mathrm{ld}}$ for the resulting lazy kernel.
\end{definition}

\begin{remark}[Fresh seeds]
We deliberately draw the sketches at the two endpoints independently.
Reusing the reverse seed after an accepted step would make the next forward
seed depend on the chain history, and this would require a different,
lifted-state analysis.
\end{remark}

In our implementation, we compute the exact metric directly when $n<s$ and
use the sketched metric when $n\ge s$.  Algorithm~\ref{II:alg:direct-or-sketched-walk}
states both branches in full.  In the sketched
branch, neither random metric is stored in the returned cache, since the next
nonlazy transition draws a fresh forward seed.

\begin{algorithm}[ht]
\caption{One step of the log-det walk with the exact or the sketched metric}
\label{II:alg:direct-or-sketched-walk}
\begin{algorithmic}[1]
\Require Current cache $(x,S_x,L_x,f(x))$ with
         $S_x=S(x)=L_xL_x^\top$, the fixed
         $\varepsilon_s,\beta_s$, and $s$ from
         Eq.~\eqref{II:eq:sketch-size}.
\State With probability $1/2$, \Return the current cache. \Comment{laziness}
\If{$n<s$}
  \State Form $\mathsf G_x\gets G_{\mathrm{ld}}(x)$.
\Else
  \State Draw a fresh TensorSRHT seed $\omega_{\mathrm{s}}$ and form
         $\mathsf G_x\gets
         \bar G_{\mathrm{ld}}(x;\omega_{\mathrm{s}})$.
\EndIf
\State Factor $\mathsf G_x=C_xC_x^\top$ and compute
       $\log\det\mathsf G_x$.
\State Draw $\xi\sim N(0,I_d)$, set
       $h\gets\frac{r}{\sqrt d}C_x^{-\top}\xi$, and set
       $y\gets x+h$ and
       $S_y\gets S_x+\sum_{i=1}^d h_iA_i$.
\If{$S_y\not\succ0$}
  \State \Return the current cache.
\EndIf
\State Factor $S_y=L_yL_y^\top$.
\If{$n<s$}
  \State Form $\mathsf G_y\gets G_{\mathrm{ld}}(y)$.
\Else
  \State Draw an independent fresh TensorSRHT seed
         $\omega_{\mathrm{s}}'$ and form
         $\mathsf G_y\gets
         \bar G_{\mathrm{ld}}(y;\omega_{\mathrm{s}}')$.
\EndIf
\State Factor $\mathsf G_y=C_yC_y^\top$, compute
       $\log\det\mathsf G_y$, and evaluate $f(y)$.
\State Compute the exact conditional log Metropolis ratio
\Statex \hspace{1.5em}$\displaystyle
 \ell\gets f(x)-f(y)
 +\frac12\log\frac{\det\mathsf G_y}{\det\mathsf G_x}
 -\frac{d}{2r^2}h^\top(\mathsf G_y-\mathsf G_x)h.$
\State Draw $U_{\mathrm{acc}}\sim\mathrm{Unif}(0,1)$.  If
       $\log U_{\mathrm{acc}}\le\min\{0,\ell\}$, \Return
       $(y,S_y,L_y,f(y))$, and otherwise \Return the current cache.
\end{algorithmic}
\end{algorithm}

\begin{theorem}[Mixing with a sketched log-det metric]
\label{II:thm:sketched-logdet-mixing}
Under the hypotheses of Theorem~\ref{II:thm:main-formal}, there are universal
constants $c_s,\bar\beta_s>0$ with the following property.  Let
$\widehat\psi\ge\psistar$ be a supplied upper bound and choose
$\kappa_{\mathrm{ld}}$ by replacing $\psistar$ with $\widehat\psi$ in
Eq.~\eqref{II:eq:rho-choice}.  If
$\varepsilon_s\le c_s/d$ and $\beta_s\le\bar\beta_s$, then the walk of
Definition~\ref{II:def:sketched-logdet-walk} is exactly reversible with
stationary distribution $\pi$.  For every $w$-warm initial law $\mu_0$ and
every $0<\delta<1$, it satisfies
\begin{equation}
 \dtv(\mu_0\bar P_{\mathrm{ld}}^t,\pi)\le\delta
 \qquad\text{for all}\qquad
 t\ge\Ot(
 (\widehat\psi nd+dL^2R^2)\log(w/\delta)).
 \label{II:eq:sketched-logdet-mixing}
\end{equation}
\end{theorem}

\begin{remark}[Reading the bound]
Although the sketch guarantee is pointwise at the current state, the failed
seeds are part of the transition kernel itself, and hence we do not need a
union bound over the mixing horizon.  Moreover, since $\psistar\le d$, the
choice $\widehat\psi=d$ is always valid and gives the unconditional order
$\Ot((nd^2+dL^2R^2)\log(w/\delta))$.
\end{remark}

\begin{proof}
We first consider the walk with the exact metric.  Since
$\widehat\psi\ge\psistar$, the proof of Lemma~\ref{II:lem:acceptance} applies
with $\psistar$ replaced by $\widehat\psi$.  We use $q_x^{\mathrm{ld}}$ and
$R_x^{\mathrm{ld}}$ to denote the exact spectrahedral proposal and ratio.  The good
event used in the proof of Lemma~\ref{II:lem:acceptance} already contains the
Gaussian norm event, and hence
\[
 \norm{Y-x}_{G_{\mathrm{ld}}(x)}^2
 =\frac{r^2}{d}\norm{\xi}^2\le Cr^2.
\]
By Lemma~\ref{II:lem:op-sc} and its distortion conclusion, we then have
$\norm{Y-x}_{G_{\mathrm{ld}}(Y)}^2\le Cr^2$ on the same feasible event.
Thus, with failure at most the same universal $\epsilon\le1/10$, the exact
proposal has a two-endpoint event on which
$\log R_x^{\mathrm{ld}}(Y)\ge-C_{\mathrm{acc}}$ and both displayed local
norms are $O(r^2)$.

Next, we compare the sketched and exact proposals conditionally on the seed.
On a good sketch seed, by Lemma~\ref{II:lem:sketch-sandwich} and the Gaussian
KL calculation of Lemma~\ref{I:lem:proposal-perturbation}, we have
\begin{equation}
 \dtv(q_x^{\mathrm{ld}},
       \bar q^{\mathrm{ld}}_{x,\omega_{\mathrm{s}}})
 \le C\sqrt d\,\theta_s=: \tau_s.
 \label{II:eq:sketch-proposal-tv}
\end{equation}
Moreover, for $z,u\in\intK$ and a good seed at $z$, we have
\begin{equation}
 |\log\frac{\bar q^{\mathrm{ld}}_{z,\omega_{\mathrm{s}}}(u)}
                        {q_z^{\mathrm{ld}}(u)}|
 \le\frac{d\theta_s}{2}
 +\frac{d(e^{\theta_s}-1)}{2r^2}
       \norm{u-z}_{G_{\mathrm{ld}}(z)}^2.
 \label{II:eq:sketch-density-error}
\end{equation}
For $\varepsilon_s\le c_s/d$, this bound is $O(1)$ at both endpoints on the
event above, and $\tau_s=O(c_s/\sqrt d)$.  Hence, on good
endpoint seeds, the conditional log Metropolis ratio is at least
$-\bar C_{\mathrm{acc}}$ for a universal $\bar C_{\mathrm{acc}}<\infty$.

To prove reversibility and to bound the common mass, we apply the subdensity argument for good seeds from Lemma~\ref{I:lem:randomized-metric}.  Since this argument is
barrier-independent, we can apply it with
Eqs.~\eqref{II:eq:sketch-proposal-tv} and
\eqref{II:eq:sketch-density-error}, the overlap of the exact proposals from
Lemma~\ref{II:lem:proposal-overlap}, and the independently good reverse seed.
Before the fixed laziness factor, the resulting common mass of the moving
kernels is at least
\[
 e^{-\bar C_{\mathrm{acc}}}(1-\beta_s)
 (1-\frac12-2\epsilon-4\tau_s-2\beta_s).
\]
We choose the universal $c_s,\bar\beta_s$ so that this lower bound is
positive.  Moreover, the seed-interchange argument of
Lemma~\ref{I:lem:randomized-metric} shows that the marginal accepted flow is
exactly symmetric, and hence the chain is exactly reversible.

It remains to bound the conductance and the mixing time.  Since the proofs of
Lemmas~\ref{II:lem:kernel-overlap}, \ref{II:lem:symmetry},
and~\ref{II:lem:conductance} use only a positive universal common-mass
constant and the geometry of the exact metric, they apply to the
sketched kernel.  Finally, the warm-start argument in Theorem~\ref{II:thm:main-formal} gives Eq.~\eqref{II:eq:sketched-logdet-mixing}.
\end{proof}
\subsection{Arithmetic complexity for spectrahedra}

At each state, we use the direct construction when $n<s$ and the sketched
construction when $n\ge s$, and charge every transition to the cheaper
metric build in its dimension regime.

For positive integers $a,b,c$, we use $\mathcal T_{\mathrm{mat}}(a,b,c)$ to
denote the arithmetic cost of multiplying an $a\times b$ matrix by a
$b\times c$ matrix.  At $\varepsilon_s=\Theta(1/d)$ and constant $\beta_s$,
we obtain the sketch size $s=\Ot(d^3)$ from Eq.~\eqref{II:eq:sketch-size}.
A complete fresh sketched-metric build has cost
\begin{equation}
 \mathsf C_{\mathrm{sk}}(n,d,s):=
 O(dn^2+n^\omega
 +d\mathcal T_{\mathrm{mat}}(n,n,s)
 +\mathcal T_{\mathrm{mat}}(d,s,d)+d^\omega).
 \label{II:eq:sketch-build-cost}
\end{equation}
Here the first two terms form and factor the slack, the third is the
TensorSRHT construction of $Z$ from \cite[Lemma~G.12]{gkmsz24}, the fourth
forms $Z^\top Z$, and the last tests and factors the $d\times d$ metric.
The $\Ot(n^2)$ Hadamard work is absorbed. We list the fourth term
separately since \cite[Lemma~G.12]{gkmsz24} constructs $Z$ but does not
form its Gram matrix.

Similarly, the direct exact build has cost
\begin{equation}
 \mathsf C_{\mathrm{ex}}(n,d):=
 O(dn^2+dn^\omega+n^2d^{\omega-1}+d^\omega),
 \label{II:eq:exact-logdet-build-cost}
\end{equation}
which comes from forming the $d$ slack-congruenced directions and their
Frobenius Gram matrix $H_{\mathrm{ld}}$ and from factoring $G_{\mathrm{ld}}$.
When $n<s$, the implementation runs the walk of
Theorem~\ref{II:thm:main-formal} with the exact metric and uses this direct
build, while for $n\ge s$ it runs the sketched walk of
Definition~\ref{II:def:sketched-logdet-walk} and uses
Eq.~\eqref{II:eq:sketch-build-cost}.  In this rule, we fix
$\varepsilon_s=c_s/d$ and $\beta_s=\bar\beta_s/2$, and we take $s$ from
Eq.~\eqref{II:eq:sketch-size}.

\begin{corollary}[Arithmetic complexity for spectrahedra]
\label{II:thm:sketched-arithmetic}
Under the hypotheses and parameter choices of
Theorem~\ref{II:thm:sketched-logdet-mixing}, consider the implementation that
uses the direct build when $n<s$ and the sketched build when $n\ge s$.  In the
exact-arithmetic model described below, each Markov step uses
\begin{equation}
 \Ot(n^\omega+n^2d^{\omega-1}
 +\min\{dn^\omega,n^2d^{3\omega-5}\})
 \label{II:eq:per-step-runtime}
\end{equation}
arithmetic operations and at most one new evaluation of $f$.
Here $\omega$ denotes the exponent of square-matrix multiplication.  For every
$0<\delta<1$, starting from an exact draw $X_0\sim\mu_0$, where $\mu_0$ is
$w$-warm, it reaches total-variation error $\delta$ using
\begin{equation}
 \Ot(
 (\widehat\psi nd+dL^2R^2)
 (n^\omega+n^2d^{\omega-1}+\min\{dn^\omega,n^2d^{3\omega-5}\})
 \log\frac{e w}{\delta})
 \label{II:eq:total-runtime}
\end{equation}
arithmetic operations and
$\Ot((\widehat\psi nd+dL^2R^2)\log(e w/\delta))$ evaluations of $f$.
With the unconditional choice $\widehat\psi=d$, replace
$\widehat\psi nd$ by $nd^2$ in both total bounds.
\end{corollary}

\begin{proof}
\emph{Direct branch.} Full-dimensional boundedness implies $d\le n(n+1)/2=O(n^2)$.  If $n<s=\Ot(d^3)$, we use the direct build
Eq.~\eqref{II:eq:exact-logdet-build-cost}.  The terms $dn^2$ and
$d^\omega$ of this build are $O(dn^\omega)$ and $O(n^2d^{\omega-1})$,
respectively, and thus
$\mathsf C_{\mathrm{ex}}=O(dn^\omega+n^2d^{\omega-1})$.  Moreover, since
$\omega\ge2$ and $n\le\Ot(d^3)$, we have
\[
 dn^\omega=n^2d\,n^{\omega-2}\le\Ot(n^2d^{3\omega-5}).
\]
Hence $dn^\omega\le\Ot(\min\{dn^\omega,n^2d^{3\omega-5}\})$, and
$\mathsf C_{\mathrm{ex}}$ satisfies Eq.~\eqref{II:eq:per-step-runtime}.

\emph{Sketched branch.} If $n\ge s$, we use the sketch. In this case, standard square blocking gives
\[
 d\mathcal T_{\mathrm{mat}}(n,n,s)
 =\Ot(n^2d^{3\omega-5}),
\]
which is the rectangular-multiplication bound used
in~\cite[Section~2.4]{gkmsz24}.  We also have
$\mathcal T_{\mathrm{mat}}(d,s,d)=O(sd^{\omega-1})$.  Because $s=\Ot(d^3)$
and $n\ge s$, this term and every other term in
Eq.~\eqref{II:eq:sketch-build-cost} are absorbed by
$\Ot(n^\omega+n^2d^{3\omega-5})$.  Moreover, since $s=\Omega(d^3)$ by
Eq.~\eqref{II:eq:sketch-size}, the condition $n\ge s$ implies
\[
 n^2d^{3\omega-5}=dn^\omega(d^3/n)^{\omega-2}=O(dn^\omega).
\]
Hence $n^2d^{3\omega-5}=O(\min\{dn^\omega,n^2d^{3\omega-5}\})$, and
$\mathsf C_{\mathrm{sk}}$ satisfies Eq.~\eqref{II:eq:per-step-runtime}.

\emph{Transitions and total cost.} A nonlazy feasible transition uses one
fresh forward build and one fresh reverse build, which only costs a constant
factor of two.  Although slack factors and the current target value may be cached,
random metrics are not reused.  Hence Eq.~\eqref{II:eq:per-step-runtime}
holds.  Multiplying it by the number of steps in Eq.~\eqref{II:eq:sketched-logdet-mixing} in the
sketched branch, and by the number of steps from the proof of Theorem~\ref{II:thm:main-formal} with
$\psistar$ replaced by $\widehat\psi$ in the exact branch, we obtain the
total transition cost.  Adding the initial metric build and the evaluation
of $f$ gives Eq.~\eqref{II:eq:total-runtime}.
\end{proof}

We remark that when $n\ge s$, the per-step bound
Eq.~\eqref{II:eq:per-step-runtime} reduces to
$\Ot(n^\omega+n^2d^{3\omega-5})$, which is the per-step cost stated
in~\cite[Theorem~1.3]{gkmsz24}.  When $n<s$, it reduces to
$\Ot(dn^\omega+n^2d^{\omega-1})$, which is the cost that
\cite[Section~2.4]{gkmsz24} records for computing the log-det Hessian
exactly.

\begin{remark}[Arithmetic model and scope]
For the spectrahedral bound, we charge unit cost to field operations,
comparisons, scalar elementary functions, and exact Gaussian, uniform, and
Rademacher variates.  The count includes forming and factoring the slack,
applying the TensorSRHT, forming its Gram matrix, factoring the proposal
metric, evaluating both conditional Gaussian densities, and making the
Metropolis decision of the sketched walk, with its independent sketches at
the current point and at the proposal.  The input supplies an exact warm draw
and the parameters $L,R,w$ and $\widehat\psi$. Choosing
$\widehat\psi=d$ avoids having to certify the leverage supremum.

While the polytope implementation uses a deterministic metric, the
spectrahedral implementation draws fresh independent sketches at both endpoints.  In both
cases, the Metropolis ratio uses the actual conditional proposal densities,
and the comparison above transfers acceptance, overlap, and conductance from
the walk with the exact metric to the implemented kernel.

The bounds exclude finite-bit conditioning, rounded Gaussian
sampling, an inexact Metropolis decision, and the construction of the warm
start.  Thus our bounds only count arithmetic operations and do not cover the
bit complexity.
\end{remark}

\end{document}